\documentclass[USletter, 11pt]{article}
\usepackage[top=1in, bottom=1in, left=1in, right=1in]{geometry}

\pdfoutput=1 

\usepackage[utf8]{inputenc}
\usepackage[T1]{fontenc}
\usepackage{lmodern}
\usepackage{microtype}
\usepackage{enumerate}

\usepackage[USenglish]{babel}

\usepackage{graphicx} 
\usepackage{xcolor}
\usepackage{amsmath,amssymb,amsthm}
\usepackage{mathtools}
\usepackage[]{hyperref}
\hypersetup{
    colorlinks,
    linkcolor={red!50!black},
    citecolor={blue!50!black},
    urlcolor=.
}
\usepackage[normalem]{ulem}
\usepackage{float}
\usepackage[capitalize]{cleveref}
\usepackage{todonotes}
\usepackage{tabularx}
\usepackage[square,numbers]{natbib}
\usepackage{tikz}
\usetikzlibrary{calc}
\usetikzlibrary{positioning,backgrounds,patterns,calc,matrix,shapes,decorations.pathreplacing,decorations.pathmorphing,math}
\tikzset{crossing/.style={cross out, draw=red, minimum size=2*(#1-\pgflinewidth), inner sep=0pt, outer sep=1pt, very thick}, crossing/.default={4pt}}
\usetikzlibrary{arrows.meta,bending,decorations.markings,hobby,patterns,calc}
\newcolumntype{C}[1]{>{\centering\let\newline\\\arraybackslash\hspace{0pt}}m{#1}}

\usepackage{thmtools, thm-restate}

\newtheorem{theorem}{Theorem}
\newtheorem{claim}[theorem]{Claim}
\newtheorem{lemma}[theorem]{Lemma}

\newtheorem{observation}[theorem]{Observation}

\newtheorem{definition}[theorem]{Definition}

\usepackage[ruled, vlined, linesnumbered, nofillcomment]{algorithm2e}

\DontPrintSemicolon
\SetKwInOut{Input}{Input}
\SetKwInOut{Output}{Output}
\SetKwProg{Procedure}{Procedure}{}{}

\newcommand{\cclass} [1] {\textnormal{\textsf{#1}}}

\newcommand{\problem}[1]{\textnormal{\textsc{#1}}}
\newcommand{\BI}{\problem{Global Reach Improvement}}
\newcommand{\MKCENTER}{\problem{Metric} $k$-\problem{Center}}
\newcommand{\MtwoKCENTER}{\problem{Metric} $2k$-\problem{Center}}
\newcommand{\MplusoneKCENTER}{\problem{Metric} $(k+1)$-\problem{Center}}
\newcommand{\DR}{\problem{Diameter Reduction}}
\newcommand{\SRI}{\problem{Reach Improvement}}
\newcommand{\RI}{\problem{Single-Source Reach Improvement}}
\newcommand{\gsc}{\textsc{Gap Set Cover}}
\newcommand{\setc}{\textsc{Set Cover}}
\newcommand{\reachOPT}{\beta^*}
\newcommand{\broadOPT}{\beta^*}

\newcommand{\broadprime}{\beta'}

\newcommand{\reachprime}{\beta'}

\newcommand{\proxradius}{\mu_r}
\newcommand{\proxdiameter}{\mu_d}

\newcommand{\HS}{\problem{Hitting Set}}
\newcommand{\centerv}{u}

\newcommand{\proxG}[3]{\mathcal{P}_{#1}({#2}, {#3})}
\newcommand{\prox}[2]{\mathcal{P}({#1}, {#2})}

\newcommand{\broad}[1]{\beta({#1})}

\newcommand{\reach}[2]{\beta_{#1}({#2})}

\newcommand{\metricfunc}{\phi}
\newcommand{\sourcev}{v_s}
\newcommand{\sourceV}{V_s}

\newcommand{\pmin}{\alpha_{min}}
\newcommand{\pmax}{\alpha_{max}}

\newcommand{\W}{\mathcal{W}}
\newcommand{\R}{\mathbb{R}}
\newcommand{\C}{\mathcal{C}}
\newcommand{\cE}{\mathcal{E}}

\newcommand{\poly}{\text{poly}}

\newcommand{\argmin}{\text{argmin}}

\newenvironment{inequality}
  {\crefalias{equation}{inequality}\begin{equation}}
  {\end{equation}\ignorespacesafterend}
\crefname{inequality}{inequality}{inequality}

\usepackage{graphicx}
\usepackage{wrapfig}
\graphicspath{ {./figs/} }

\newcommand{\su}[1]{^{(#1)}}

\newcommand{\problemdef}[3]{
	\begin{center}
	\begin{minipage}{0.95\columnwidth}
		\noindent
		\textsc{#1}
		\vspace{5pt}\\
		\setlength{\tabcolsep}{3pt}
		\begin{tabularx}{\textwidth}{@{}lX@{}}
			\textbf{Input:}     & #2 \\
			\textbf{Question:}  & #3
		\end{tabularx}
	\end{minipage}
	\end{center}
}

\usepackage{authblk}

\title{Optimizing Probabilistic Propagation in Graphs by Adding Edges}
\date{}

\author{Aditya Bhaskara}
\author{Alex Crane}
\author{Shweta Jain}
\author{Md Mumtahin Habib Ullah Mazumder}
\author{Blair~D.~Sullivan}
\author{Prasanth Yalamanchili}

\affil{University of Utah}

\begin{document}

\maketitle
\thispagestyle{empty}

\begin{abstract}
Probabilistic graphs are an abstraction that allow us to study randomized propagation in graphs. In a probabilistic graph, each edge is ``active'' with a certain probability, independent of the other edges. For two vertices $u,v$, a classic quantity of interest, that we refer to as the \emph{proximity} $\proxG{G}{u}{v}$, is the probability that there exists a path between $u$ and $v$ all of whose edges are active. For a given subset of vertices $\sourceV$, the \emph{reach} is defined as the minimum over pairs $u \in \sourceV$ and $v \in V$ of the proximity $\proxG{G}{u}{v}$. This quantity has been studied in the literature in the context of multicast in unreliable communication networks and in social network analysis.

We study the problem of \emph{improving} the reach in a probabilistic graph via edge augmentation. Formally, given a budget $k$ of edge additions and a set of source vertices $\sourceV$, the goal of \SRI{} is to maximize the reach of $\sourceV$ by adding at most $k$ new edges to the graph. The problem was introduced in earlier empirical work in the algorithmic fairness community~\cite{bashardoust2023reducing}, but lacked any formal guarantees. In this work, we provide the first approximation guarantees and hardness results for \SRI{}. 

We prove that the existence of a good augmentation implies a cluster structure for the graph in an appropriate metric. We use this structural result to analyze a novel algorithm that outputs a $k$-edge augmentation with an objective value that is poly($\broadOPT{}$), where $\broadOPT{}$ is the objective value for the optimal augmentation. When we are allowed slack in the edge budget $k$, we give an algorithm that adds $O(k \log n)$ edges, and a multiplicative approximation to the objective value. Our arguments rely on new probabilistic tools for analyzing proximity, inspired by techniques in percolation theory; these tools may be of broader interest. We also prove lower bounds, showing that significantly better approximation algorithms are unlikely, under known hardness assumptions related to gap variants of the classic \setc{} problem.
\end{abstract}

\newpage
\setcounter{page}{1}
\section{Introduction}\label{sec:intro}

We study communication between the nodes of a \emph{probabilistic} graph, i.e., a graph in which each edge is \emph{active} only with a certain probability. Probabilistic graphs allow us to model a variety of randomized diffusion processes in networks; some classic examples include information propagation in social networks (e.g., via the classic independent cascade model~\cite{kempe03maximizing,kempe2005influential,borgs2014maximizing,wang2012}), the spread of an epidemic in a population~\cite{newman2002spread,physRevE.76.036113,srivastava2025}, and finally, communication between nodes in a network whose whose links (edges) are prone to failures~\cite{callaway2000network,goldschmidt1994reliability,ke2020reliability,albert2000error}. 

Formally, a probabilistic graph $G$ is defined using a set of vertices $V$, a set of edges $E$, and a parameter $\alpha_e \in [0,1]$ for each $e \in E$, which is the probability that $e$ is ``active.'' We assume that edges are bidirectional (i.e., if an edge $uv$ is active, then $u$ can communicate with $v$ and vice versa), and that edges are active independently of one another. For a pair of vertices $u, v \in V$, we can define the \emph{proximity} $\proxG{G}{u}{v}$ as the probability that there exists a path from $u$ to $v$ in $G$, using only active edges. It is well-known (see, e.g., \cite{chen2010scalable}) that computing $\proxG{G}{u}{v}$ exactly is \#P-hard. Intuitively, this is because there can be exponentially many paths between two vertices, sharing edges in complex ways. However, in practice, Monte Carlo simulations can often be used to obtain approximations to $\proxG{G}{u}{v}$, as long as the probabilities estimated are not too small~\cite{fishman2007comparison,li2015recursive}.

In this paper, we study the problem of \emph{augmenting} a graph $G$ via edge addition in order to improve the proximity between vertices. We focus on applications in which a subset of the vertices wish to communicate reliably with all vertices in the graph. In the context of a social network, this corresponds to a subset of the users (content generators) being able to ``reach'' all the users in the graph; in a communication network, this corresponds to a subset of nodes being able to perform a broadcast. Formally, we have a set of \emph{source} vertices $V_s$, and we focus on the objective function:
\[ \reach{G}{\sourceV} := \min_{u \in \sourceV} \min_{v \in V} \proxG{G}{u}{v}.  \]

With this objective, we study the following augmentation problem: given a budget $k$ and a set of source vertices $\sourceV$, add at most $k$ edges to maximize $\reach{G'}{\sourceV}$ for the resulting graph $G'$. We call this problem \SRI{}, and our goal will be to develop approximation algorithms and understand its complexity. For a formal definition of \SRI{} (which also requires $\alpha_e$ values for all possible edges), see Section~\ref{sec:preliminaries}. 

The special case of $\sourceV = V$, which we will subsequently refer to as \BI{}, was studied in the work of~\cite{bashardoust2023reducing}, where they motivate the quantity $\reach{G}{V}$ from the perspective of \emph{fairness in information access} in social networks. However,~\cite{bashardoust2023reducing} provides only an empirical study of the problem; no formal approximation guarantees or hardness results are known. One of our main contributions is to provide such guarantees. Another related line of research focuses on the problem of edge augmentation to minimize the \emph{shortest path distance} between pairs of nodes (for the cases of $\sourceV$ being a singleton or $\sourceV = V$). While these works (e.g.,~\cite{li1992minimum,demaine2010minimizing,bilo2012improved,frati2015augmenting}) develop approximation algorithms, they rely on the relatively simple structure of the shortest path metric, which fails to hold in our setting. These, as well as other related works, are discussed in Section~\ref{sec:warmup}.

\subsection{Our Results}\label{sec:our-results}

We design approximation algorithms and prove lower bounds for the \SRI{} problem.  
The following notation will help describe our results. Let $G = (V, E, \{\alpha_e\})$ denote the given probabilistic graph, and let $\sourceV$ be the set of source vertices. We write $n = |V|$, and use $k$ to denote the edge addition budget. For every (potential) edge $e\in \binom{V}{2}$, we assume that the value of $\alpha_e$ is known (given as input) and that $\alpha_e \in [\pmin,\pmax]$. We denote the optimal objective value by $\broadOPT$.

We first focus on the setting where $\sourceV = V$, as studied in~\cite{bashardoust2023reducing}. This setting turns out to capture most of our algorithmic ideas. Our result is the following:

\begin{theorem}[Informal version of Theorem~\ref{thm:ballpacking-theorem}]\label{thm:informal-main-algo}
For any probabilistic graph $G$ and budget parameter $k$, there is a polynomial-time algorithm that outputs a set $S$ of at most $k$ edges to add to $G$, such that for the augmented graph $G+S$ we have
\[ \reach{G+S}{V} \ge (\broadOPT)^4 \cdot \poly \left( \frac{\alpha_{\min}}{k} \right).  \]
\end{theorem}

Our approximation guarantee is ``polynomial'' in $\broadOPT$. As we will see, the quantities $\proxG{G}{u}{v}$ form a metric after taking the logarithm. So in a sense, the polynomial factor above actually corresponds to a constant factor approximation (plus additive factors) after taking logarithms. 

Next, we generalize the result to the case of arbitrary source vertices $\sourceV$. We obtain a slightly weaker result in general, but match the result of~\Cref{thm:informal-main-algo} when $|\sourceV|=1$.

\begin{theorem}[Informal version of~\Cref{claim:ss-ballpacking-theorem,claim:sr-ballpacking-theorem}]\label{thm:informal-subset-source}
For any probabilistic graph $G$, budget parameter $k$, and set of source nodes $\sourceV$, there is a polynomial-time algorithm that outputs a set $S$ of at most $k$ edges to add to $G$, such that for the augmented graph $G+S$ we have
\[ \reach{G+S}{\sourceV} \ge (\broadOPT)^8 \cdot \poly \left( \frac{\alpha_{\min}}{k} \right).  \]

Furthermore, when $\sourceV$ is a singleton $\{\sourcev\}$, the guarantee matches that of Theorem~\ref{thm:informal-main-algo}:
\[ \reach{G+S}{\{\sourcev\}} \ge (\broadOPT)^4 \cdot \poly \left( \frac{\alpha_{\min}}{k} \right).  \]
\end{theorem}

Our algorithms rely on the following key idea: if the addition of $k$ edges improves the reach to $\beta^*$, then for every $v \in V$, at least one of the endpoints of the $k$ added edges must be ``close'' to $v$. This implies that balls of small radius around the endpoints of the optimal edge additions must cover all the vertices in $V$. This structural result (\Cref{thm:strengthened-ball-growing-theorem}) is used algorithmically: we find a mutually well-separated set of vertices, and argue that this set cannot be too large. By taking care when we pick these vertices (each newly selected vertex is separated yet ``close enough'' to previously chosen vertices), we show how to add at most $k$ edges while obtaining the desired approximation.

A significant technical challenge in our analysis is reasoning about the proximity function $\proxG{G}{u}{v}$, because there can be many $uv$-paths, potentially sharing edges. We develop a probabilistic tool we call the \emph{Splitting Lemma} that plays a key role in our analysis. It turns out to be a consequence of the Van den Berg-Kesten inequality from percolation theory~\cite{BK,bollobas2006percolation}, and may be of independent interest.\looseness=-1

It is natural to ask if the polynomial dependence on~$\broadOPT$ is necessary. We show that this is indeed the case, under well-known hardness assumptions. In fact, we show that a polynomial loss is unavoidable even for \emph{bicriteria} approximation algorithms, i.e., algorithms that can output $O(k)$ edges while trying to compete with an optimum that adds only $k$ edges:

\begin{theorem}[Informal version of Theorem~\ref{thm:broadcasthardness}] \label{thm:informal-hardness}
For any constant $c\ge 1$, under the \normalfont{\gsc{}} assumption (\Cref{asmp:hardgsc}), unless  $\cclass{P} = \cclass{NP}$, there is no polynomial-time algorithm for \SRI{} that adds most $ck$ edges to $G$ and guarantees $\reach{G'}{\sourceV} \ge (\broadOPT)^{1.1}$ for the resulting $G'$.
\end{theorem}

We remark that the hardness result also holds for the cases $\sourceV = V$ and $|\sourceV|=1$ (Theorem~\ref{thm:sshardness}), as well as for uniform edge activation probabilities.  Our reductions are from a known (and \cclass{NP}-hard~\cite{feige1998,feige2010}) variant of \setc{}, known as \gsc{}, which asks us to distinguish between instances where $k$ sets cover the entire universe (the YES-case) and those where any $ck$ sets cover at most a $(1-\delta)$-fraction of the universe (the NO-case), for some $\delta = \Omega_c(1)$. 

Motivated by the connection to \setc{}, we can ask if there exist algorithms that obtain \emph{linear} approximations to $\broadOPT$ if we allow $O(k \log n)$ edge additions. Our final algorithmic result shows that this is indeed the case:
\begin{theorem}[Informal version of Theorem~\ref{thm:main-klogn-algo}]\label{thm:informal-main-klogn-algo}
For any probabilistic graph $G$, budget parameter $k$, and set of source nodes $\sourceV$, there is a polynomial-time algorithm that outputs a set $S$ of $O(k \log n)$ edges to add to $G$, such that for the augmented graph $G+S$ we have
\[ \reach{G+S}{\sourceV} \ge \broadOPT \cdot \poly\left(\frac{\alpha_{\min}}{k}\right).  \]
\end{theorem}

Recall that one way to analyze the greedy algorithm for \setc{} is by observing that the ``coverage'' function, which measures how many elements are covered by the chosen sets, is submodular. In our setting, it turns out that the objective $\beta_{G'}(\sourceV)$, viewed as a function of the set of added edges, is not submodular (even after taking the logarithm). However, we show that there always exists a near-optimal \emph{star solution}, i.e., one in which all edge additions share a common endpoint, and furthermore that if we restrict our search space to star solutions, submodularity holds. These insights allow us to give a greedy algorithm based on a submodular potential function, but surprisingly, it is non-trivial to show that the optimal star solution optimizes our potential function. To do this, we once again use the Splitting Lemma.

\vspace{-6pt}
\paragraph{Extensions.} 
In~\Cref{appendix:constant-witnesses}, we give linear approximations for~\SRI{} which use $\poly(k)\cdot\log n$ edge additions. These results are weaker in general than~\Cref{thm:informal-main-klogn-algo}, but they may be preferable for certain values of $k, \pmin, \pmax$, and a parameter $\varepsilon$. The key step is to show that there is a near-optimal set $S$ of edge additions for which, given any $(u, v) \in \sourceV \times V$, only a constant number of edge additions from $S$ are needed to witness that $u$ and $v$ have high proximity in the augmented graph. We leverage this result to obtain a reduction to~\setc{}. Our technique may be of independent interest in the broader context of graph modification.

\subsection{Related work}

Our results contribute to the extensive literature on edge augmentation in order to improve various connectivity metrics in a graph. Perhaps closest to our work are the results on minimizing the diameter of a graph via edge additions. Specifically, the results of~\cite{li1992minimum, demaine2010minimizing, bilo2012improved,frati2015augmenting} also provide inspiration for some of our techniques. In the data mining literature, other metrics such as betweenness centrality have been studied from the perspective of edge augmentation (see, e.g.,~\cite{bergamini2018improving,papagelis2011suggesting,crescenzi2016greedily}). 
Augmentation to improve spectral metrics, such as the sum of effective resistances (known as the Kirchoff index), have also been studied in more recent works~\cite{Zhou_2025,achterberg2025}. Improving algebraic connectivity, specifically the Laplacian eigenvalue $\lambda_2(\mathcal{L}(G))$, has been studied in the work of~\cite{ghosh2007growing,wang2008algebraic}.

Transmission along the edges of a probabilistic graph has been extensively studied in probability theory as \emph{percolation}. Indeed, properties of the quantity $\proxG{G}{u}{v}$ are well-understood for (high-dimensional) grids and graphs with special structure~\cite{bollobas2006percolation}. Some of our key technical lemmas (e.g., Lemma~\ref{lemma:splitting}) rely on inequalities developed in this context. In the algorithms and data mining communities, probabilistic propagation has been used to model the spread of epidemics and information in networks; see the classic work of Kempe, Kleinberg, and Tardos~\cite{kempe03maximizing} on influence maximization. Here, the goal is to find a set of $k$ vertices such that a propagation process starting at these vertices ``reaches'' as many other vertices as possible in expectation. This work has had several applications to online advertising and information spread. However, it has been observed that influence maximization may end up unfairly isolating small groups of users~\cite{fish2019gaps,stoica2019fairness}. This has led to a significant body of research around \emph{fairness of information access} on social networks, e.g., \cite{tsang2019group, ali2019fairness, farnad2020unifying, rahmattalabi, becker2023improving}, and our work provides algorithms with approximation guarantees in this setting.

More broadly, beyond connectivity, problems such as shortest paths have been studied in probabilistic graphs. An extension of our notion of a probabilistic graph is one in which the \emph{length} of each edge is a random variable (for instance, in an application involving congestion in road networks). Here, a well-studied problem~\cite{loui1983optimal,nikolova2006optimal,nikolova2006stochastic,zhang2023finding} is that of finding the expected shortest path between given vertices. While these works are not directly related to our formulation, they are related to variants of \SRI{} in which one is interested in connectivity using paths of bounded length. It is interesting to ask if our techniques shed light on these settings.
\section{Preliminaries}\label{sec:preliminaries}
A \emph{probabilistic graph} $G = (V, E, \{\alpha_e \mid e \in \binom{V}{2}\})$ is a simple, undirected, connected graph $(V, E)$ and a set of activation probabilities for each $e \in \binom{V}{2}$. Note that we have abused the usual notation slightly by assuming activation probabilities for every possible edge, not just those in $E$. We will usually just write $\{\alpha_e\}$, and we will assume throughout that the activation probabilities are bounded below and above by constants $\pmin, \pmax \in (0, 1]$.
A \emph{realization} of the probabilistic graph $G$ (or a \emph{sampled graph}) is formed by deleting each edge in $E$ independently with probability $1 - \alpha_e$. Equivalently, we may think of ``activating'' each edge $e$ independently with probability $\alpha_e$.
The \emph{proximity function} $\proxG{G}{u}{v}$ is given by the probability with which $u$ and $v$ are connected in a sampled graph.
Recalling that Monte Carlo techniques may be used in practice~\cite{fishman2007comparison,li2015recursive}, we assume an oracle which computes pairwise proximity values in polynomial time.
The \emph{reach} of a vertex subset $\sourceV \subseteq V$ is $\min_{u \in \sourceV}\min_{v \in V} \proxG{G}{u}{v}$. We can now formally define our main problem.\looseness=-1 

\problemdef{\SRI{}}{A probabilistic graph $G = (V, E, \{\alpha_e\})$, a source-set $\sourceV \subseteq V$, and $k \in \mathbb N$.}
{Find a set $S \subseteq V^2 \setminus E$ of at most $k$ edge additions such that $\reach{G+S}{\sourceV}$ is maximized.}

We call the problem \BI{} if $\sourceV = V$, or \RI{} if $|\sourceV| = 1$.
We write $S^*$ for an optimal set of edge additions, $G^* = G +$~$S^*$ for the optimally augmented graph, and $\reach{G^*}{\sourceV}$ for the optimum reach achievable.
If $\sourceV = V$, we write $\broadOPT = \broad{G^*}$ for $\reach{G^*}{\sourceV}$, and if $\sourceV = \{\sourcev\}$ we write $\reach{G^*}{\sourcev}$ for $\reach{G^*}{\sourceV}$.
Many of our techniques for reasoning about proximity rely on
analyzing paths. We present here some non-standard notation.

\begin{definition}
    Let $p = v_0, v_1, v_2, \ldots, v_\ell$ be a (simple, unless noted otherwise) path. For $i < j$, we write $p[v_i, v_j] = v_i, v_{i+1}, \ldots, v_{j}$ for the \emph{segment} of $p$ from $v_i$ to $v_j$.
    We extend this notation to sets of paths when every path in the set shares a vertex. That is, for a set of paths $P$ from $w$ to $v$ which all use the vertex $u$, we write $P[w, u]$ and $P[u, v]$ for the \emph{segments} $\{p[w, u] \colon p \in P\}$ and $\{p[u, v] \colon p \in P\}$, respectively.
    The probability that at least one path in $P$ exists in a sampled graph is the \emph{contribution} of $P$, denoted $\Pr[P]$.\looseness=-1
\end{definition}

By definition, the proximity of $u$ to $v$ is the contribution of the set consisting of all $uv$-paths.
Observe that the negative logarithm of
proximity is a metric on the vertices of a graph. Formally, let $G = (V, E)$,
and define $\metricfunc \colon V^2 \rightarrow \R_{\geq 0}$ by $\metricfunc(u, v) = -\log\proxG{G}{u}{v}$. We claim that
$(V, \metricfunc)$ is a metric space. Symmetry and non-negativity are trivial, and every point has distance zero
to itself because $\prox{v}{v} = 1$ for all $v \in V$. The triangle inequality follows quickly from the
observation that for all $u, v, w \in V, \prox{u}{v} \geq \prox{u}{w} \cdot \prox{w}{v}$. 
Henceforth, we refer to $(V, \metricfunc)$ as the \emph{implied metric} of $G$.

We conclude this section with a brief summary of a first result for \BI{}, which will provide useful context for our later technical developments.
Intuitively, the following lemma says that edge additions ``far'' from a given vertex $u$ cannot drastically improve the reach of $u$.
Note that for a vertex subset $X \subseteq V$ and a vertex $u$, we write $\proxG{G}{u}{X} = \max_{v \in X}\proxG{G}{u}{v}$ for the proximity of $u$ to $X$.
\begin{restatable}{lemma}{ballgrowinglemma}\label{lemma:ball-growing}
    Let $G = (V, E, \{\alpha_e\})$, $u, v \in V$. Let $S \subseteq V^2 \setminus E$ be a set of $k$ edge additions, and $V(S)$ be the endpoints of the edges in $S$.
    Then $\proxG{G + S}{u}{v} \leq \proxG{G}{u}{v} + \proxG{G}{u}{V(S)}\cdot 2k\pmax$.
\end{restatable}

In~\Cref{appendix:k-center}, we prove~\Cref{lemma:ball-growing} by partitioning the paths from $u$ to $v$ in $G+S$ according to the first edge in $S$ appearing along the path (accounting for the orientation of the edge), and applying the union bound. We also show by example that the bound given by the lemma is asymptotically tight.
A simple consequence is a covering result for the implied metric of $G$. To see this, observe that if $\proxG{G}{u}{v}$ and $\proxG{G}{u}{V(S^*)}$ are both less than $\broadOPT/(1 + 2k\pmax)$, then by~\Cref{lemma:ball-growing}, $\proxG{G^*}{u}{v} < \broadOPT$, a contradiction. Thus, every $u \in V$ must have proximity at least $\broadOPT/(1+2k\pmax)$ to some vertex in $V(S^*)$.\looseness=-1

\begin{restatable}{corollary}{twokballcorollary}\label{cor:2kballs}
    There exist at most $2k$ balls of radius $-\log \frac{\broadOPT}{1 + 2k\pmax}$ which cover $V$.
\end{restatable}

\Cref{cor:2kballs} implies a straightforward reduction to \MtwoKCENTER{}, where edges are added to form a star on the selected centers.
Note that the dependence on $\broadOPT$ in the approximation guarantee is quartic instead of quadratic because we must use a $2$-approximation, e.g.,~\cite{gonzalez1985clustering,hochbaum1985best},
to solve \MtwoKCENTER{} in polynomial time. We defer the details to~\Cref{appendix:k-center}.

\begin{restatable}{proposition}{kcenterreduction}\label{thm:k-center-reduction}
    There exists a polynomial-time algorithm for \BI{} which produces a probabilistic graph with reach at least $\frac{(\broadOPT)^4\pmin^2}{(1+2k\pmax)^4}$ by adding at most $2k-1$ edges.
\end{restatable}

\section{A Polynomial Approximation for \BI{}}\label{sec:polynomial-approximation}

We now present our main algorithmic result for \BI{}, i.e., the case where $\sourceV = V$. Our guarantee will be nearly the same as that of~\Cref{thm:k-center-reduction}, up to poly($k, \pmin, \pmax$) factors, but we achieve it by adding at most $k$ edges, rather than $2k-1$.

\begin{restatable}{theorem}{ballpackingtheorem}\label{thm:ballpacking-theorem}
Let $\broadOPT$ be the optimum objective value for the \BI{} problem on a probabilistic graph $G$ and parameter $k$, as in~\Cref{sec:preliminaries}. For any constant $\varepsilon > 0$, there is a polynomial-time algorithm that finds a set $S$ of at most $k$ edge additions, such that
\[ \reach{G+S}{V} \ge  \frac{(\broadOPT{})^4 \pmin^{4}}{4^4k^8(1 + \varepsilon)^4}. \] 
\end{restatable}

\paragraph{Remark.} Somewhat surprisingly, we show that when $G$ is connected (as we have assumed), our algorithm adds at most $(k-1)$ edges.

\subsection{Warm-up: An Exponential Approximation}\label{sec:warmup}

To begin, we will give an algorithm that achieves the goal of adding exactly $k$ edges, but has an approximation factor that is exponential in $k$. The algorithm is based on trying to strengthen~\Cref{lemma:ball-growing} by analyzing what happens when a \emph{single edge} is added. This approach is an analog of the techniques used in~\cite{bilo2012improved} to obtain the best-known approximation factor for the \DR{} problem.  While weaker than our main result (\Cref{thm:ballpacking-theorem}), the analysis will help build intuition and motivate one of our key technical lemmas (\Cref{lemma:splitting}). We defer some of the details to \Cref{appendix:splitting-lemma,appendix:exponential-algs}; there, we also discuss an approach based on the related work of~\cite{li1992minimum}. 

Analogously to~\cite{bilo2012improved}, we say that a vertex subset $X \subseteq V$ is a $\beta$-\emph{independent set} in $G$ if $\proxG{G}{u}{v} < \beta$ for all $u, v \in X$. The following claim lets us bound the effect of a single edge addition.

\begin{restatable}{claim}{independentsetclaim}\label{claim:independent-set-claim}
If there exists a $\beta$-independent set $X$ in a probabilistic graph $G$, then for every $uv \in V^2 \setminus E$, there exists a $3\beta$-independent set of size $|X| - 1$ in $G+uv$. 
\end{restatable}

Before sketching the proof, let us see why~\Cref{claim:independent-set-claim} allows us to obtain an algorithm. By definition, for the optimal $k$-edge augmentation $G^*$ of $G$, there is no $\broadOPT$-independent set of size $2$ in $G^*$. Then by induction (using~\Cref{claim:independent-set-claim}) there is no $(\broadOPT/3^k)$-independent set of size $k+2$ in $G$.
Now, let $c_1, c_2, \ldots, c_{k+1}$ be the centers chosen by furthest point traversal (i.e., the Gonzalez algorithm~\cite{gonzalez1985clustering} for \MplusoneKCENTER{}) in the implied metric of $G$: point $c_1$ is chosen arbitrarily, and for $i=2, 3, \dots, k+1$, $c_i$ is chosen to be $\argmin_{v \in V}\max_{j < i} (\proxG{G}{v}{c_j})$, breaking ties arbitrarily. Since there is no $(\broadOPT/3^k)$-independent set of size $k+2$ in $G$, it follows that for every $v \in V$, there is some $c_i$ such that $\proxG{G}{v}{c_i} \geq \broadOPT/3^k$. Thus, by adding $k$ edges to form a star on the centers $c_1, c_2, \dots, c_{k+1}$ and using the triangle inequality for the implied metric, we obtain the following result.

\begin{restatable}{theorem}{independentsettheorem}\label{thm:independent-set-thm}
There exists a polynomial-time algorithm which produces reach (for $\sourceV = V$) at least $\frac{\broadOPT{}^2\pmin^2}{9^{k}}$ by adding at most $k$ edges. 
\end{restatable}

We will remark more on the bound given by~\Cref{thm:independent-set-thm} at the end of this subsection. Now, we outline the proof of~\Cref{claim:independent-set-claim}. Assume toward a contradiction that the claim is false, so no subset $X' \subseteq V$ (and in particular, no $X' \subseteq X$) of size $|X| - 1$ is a $3\beta$-independent set in $G+uv$.
Then, we can argue that upon adding $uv$, either three distinct vertices of $X$ ``got closer'' or two \emph{pairs} of vertices got closer. For the sake of our outline, we focus on the latter case, specifically that there exist four distinct vertices $x_i, x_j, y_i, y_j \in X$ with $\proxG{G+uv}{x_i}{x_j}$ and $\proxG{G+uv}{y_i}{y_j}$ both $\geq 3\beta$. The other case is similar; see~\Cref{appendix:bilo-tecnique}. Since the vertices formed a $\beta$-independent set in $G$, we can argue that paths from $x_i$ to $x_j$ using the edge $uv$ have contribution at least $2\beta$ in $G+uv$, and furthermore, without loss of generality, the subset $P_u$ of these paths on which $u$ precedes $v$ has contribution at least $2\beta/2 = \beta$. By a similar argument (after possibly relabeling $y_i, y_j$), paths from $y_i$ to $y_j$ using $uv$ in the orientation $u \rightarrow v$ have contribution at least $\beta$ in $G + uv$.

Now the key step in the argument is proving that $\Pr[P_u[x_i, u]] \cdot \Pr[P_u[u, x_j]] \geq \beta$. Once we have this, since all paths in $P_u[u, x_j]$ begin with the edge $uv$, we have that $\Pr[P_u[x_i, u]] \cdot \Pr[P_u[v, x_j]] \geq \beta$.
Moreover, since all edges used by paths in $P_u[x_i, u]$ and $P_u[v, x_j]$ exist in $G$, we may conclude that $\proxG{G}{x_i}{u}\cdot\proxG{G}{v}{x_j} \geq \beta$, and by a similar argument $\proxG{G}{y_i}{u}\cdot\proxG{G}{v}{y_j} \geq \beta$.
Using these inequalities, it follows (see~\Cref{appendix:bilo-tecnique}) that either $\proxG{G}{x_i}{y_i} \geq \beta$ or $\proxG{G}{x_j}{y_j} \geq \beta$, which is a contradiction because $x_i, x_j, y_i, y_j$ are distinct vertices in the $\beta$-independent set $X$.

It remains to prove the key step described above. Note that the analog of this step in the shortest-path context is a triviality:
if there exists a path of length $\ell$ from $x_i$ to $x_j$ via the vertex $u$, then $d(x_i, u) + d(u, x_j) \leq \ell$.
When reasoning about a single path of contribution $\mu$, this lifts to our setting, i.e., we may ``split'' the path at an internal vertex $u$ and observe that $\proxG{G}{i}{u} \cdot \proxG{G}{u}{j} \geq \mu$.
Unfortunately, unlike for shortest path distances, to understand proximity we need to reason about \emph{sets} of paths, so we need stronger techniques.
We show the following.

\begin{restatable}[The Splitting Lemma]{lemma}{splittinglemma}\label{lemma:splitting}
    Let $G = (V, E)$, $i, j, u \in V$. Let $P_u$ be a set of simple paths from $i$ to $j$ which use the vertex $u$. Then $\Pr[P_u] \leq \Pr[P_u[i, u]] \cdot \Pr[P_u[u, j]]$.
\end{restatable}

Essentially, the Splitting Lemma says that we may reason about sets of paths which all use a vertex $u$ in much the same way as we reason about a single path, i.e., we may ``split'' the paths at $u$ and obtain a useful lower bound on the product $\proxG{G}{i}{u} \cdot \proxG{G}{u}{j}$.
To understand why the lemma is non-trivial,
define $\cE_i$ (resp., $\cE_j$) to be the event that one of the paths from $P_u[i, u]$ (resp., $P_u[u, j]$) exists in a graph sampled from $G$ (note that the lemma is true for any set of simple paths $P_u$). Observe that the edges used by paths in $P_u[i, u]$ may not be disjoint from those used by paths in $P_u[u, j]$. In this case, events $\cE_i$ and $\cE_j$ are positively correlated, i.e., $\Pr[ \cE_i | \cE_j ] \ge \Pr[ \cE_i ]$. This implies that $\Pr[ \cE_i \cap \cE_j ] \ge \Pr[\cE_i] \cdot\Pr[\cE_j]$. But unfortunately, this inequality goes in the reverse direction of what we would like, as our goal is to lower-bound the product $\Pr[\cE_i] \cdot\Pr[\cE_j]$. 
To overcome this, the key is to observe that the paths in $P_u$ are simple. This implies that if a path $p \in P_u$ exists in the sampled graph, we can find two vertex-disjoint (not considering $u$) paths $p_i \in P_u[i, u]$ and $p_j \in P_u[u, j]$ that are also in the sampled graph. 

To make this idea more formal, let $\cE_{ij}^*$ be the event that in a sampled graph, there exists $p_i \in P_u[i, u]$ and $p_j \in P_u[u, j]$ that both occur, and moreover that $p_i, p_j$ are edge-disjoint. Let $\cE_u$ be the event that at least one path from $P_u$ exists.
By definition, $\Pr[\cE_u] \leq \Pr[\cE_{ij}^*]$, so it suffices to show that $\Pr[\cE_{ij}^*] \leq \Pr[\cE_i]\cdot\Pr[\cE_j]$.
We do this by applying the van den Berg-Kesten inequality~\cite{BK}, one of the fundamental inequalities in percolation theory~\cite{bollobas2006percolation}. Informally, the inequality is applicable because in a sampled graph satisfying event $\cE_{ij}^*$, the paths $p_i$ and $p_j$ can also be seen as ``disjoint certificates'' for events $\cE_i$ and $\cE_j$. We formalize this argument in~\Cref{appendix:splitting-lemma}.

We conclude by remarking that the exponential dependence on $k$ in~\Cref{thm:independent-set-thm} is caused by the loss in~\Cref{claim:independent-set-claim}. That is, the claim only guarantees the existence of a $3\beta$-independent set after adding a single edge, so inductive application of the claim yields an exponential loss.
A natural question is whether the lossless version of the claim, i.e., replacing $3\beta$ with $\beta$ (or, for example, $\beta (1+1/k)$) in the statement, can be proven. This turns out to be impossible in general; we give a counterexample in~\Cref{appendix:bilo-tecnique-lb}.
A second approach is to try to develop a claim which analyzes the structure of independent sets when batches of edges are added to a graph. 
More formally, we might hope that for some functions $f, g$, if a $\beta$-independent set $X$ exists in $G$, and $f(k)$ edges are added to form $G'$, then a $g(k, \beta)$-independent set of size at least $|X| - f(k)$ exists in $G'$.
As stated,~\Cref{claim:independent-set-claim} corresponds to the functions $f(k) = 1$ and $g(k, \beta) = 3\beta$. However, an algorithm guaranteeing reach $(\broadOPT)^{O(1)}\cdot\poly(k, \pmin)$ could be obtained, for instance, if $f \in \Omega(k)$ and $g \in \poly(k, \beta)$ or if $f \in \Omega(k/\log k)$ and $g \in O(1)$.
Unfortunately, the natural extension of the argument of~\Cref{claim:independent-set-claim} leads to a combinatorial explosion, so we cannot improve on the $\exp(k)$ bound. We discuss this in more detail in~\Cref{appendix:bilo-tecnique-lb}. Thus, to achieve~\Cref{thm:ballpacking-theorem} we will develop a new algorithmic technique.\looseness=-1

\subsection{Polynomial Approximation: Proof of~\Cref{thm:ballpacking-theorem}}\label{sec:ballpacking}

Now we will develop the remaining tools needed to prove~\Cref{thm:ballpacking-theorem}.
The first is a stronger structural result that shows that $V$ can be covered using $2k$ balls of a small \emph{proximity radius} and one set of small \emph{proximity diameter}. Let us now formalize this notation. For any $v \in V$, the \emph{ball of proximity radius} $\proxradius$ is defined as
\[ B_v (\proxradius) := \{u \in V~:~ \proxG{G}{v}{u} \ge \proxradius\}. \]

Note that a ball grows larger as $\proxradius$ decreases. Likewise, we say that a set of vertices $U$ has proximity diameter $\proxdiameter$ if for all $u,v \in U$, we have $\proxG{G}{u}{v} \ge \proxdiameter$. In this notation,~\Cref{cor:2kballs} showed that if $G$ can be augmented using $k$ edges to obtain $G^*$ with $\broadOPT = \broad{G^*}$, then a set of at most $2k$ balls of proximity radius $\broadOPT \cdot \poly(k,\pmax)$ suffice to cover the vertices of $G$. We now show the following result, which improves the radius significantly, at the expense of an additional ball.

\begin{restatable}{theorem}{strengthenedballgrowingtheorem}\label{thm:strengthened-ball-growing-theorem}
Suppose there is a $k$-edge augmentation of a probabilistic graph $G$ that yields $G^*$ with reach $\broadOPT = \broad{G^*}$. Let $\proxdiameter = \frac{\broadOPT}{4k^2}$ and $\proxradius = \sqrt{\proxdiameter}$. Then, we can cover $V$ as $V = \left( \cup_{i=1}^{2k} F_i \right) \cup U$, where $F_i$ are balls of proximity radius $\proxradius$ and $U$ is a set with proximity diameter at least $\proxdiameter$. 
\end{restatable}

The proof proceeds as follows. Let $c_1, c_2, \dots, c_\ell$, with $\ell \le 2k$, be the endpoints of the $k$ edges added to obtain $G^*$. We will use $F_i = B_{c_i} (\proxradius)$. The key step will be to show that the set of vertices not covered by these balls has proximity diameter at least~$\proxdiameter$. To argue this, we will use the Splitting Lemma. The details will be presented below, in~\Cref{sec:proof-of-main-theorem}. 

Next, we present the second tool needed for~\Cref{thm:ballpacking-theorem}, which is a new algorithmic insight: we start by obtaining $2k+1$ balls of radius roughly $\proxdiameter$ that cover $V$. But crucially, the balls will be chosen such that each ball will be \emph{neighboring} one (or more) other balls. This lets us consider an auxiliary graph whose vertices are the centers of the balls, and edges are between centers of neighboring balls. The dominating set in this graph turns out to have size at most $k$, which then leads to an augmentation with at most $(k-1)$ edges and the desired approximation guarantee. 

Let us now formalize this outline. 
First, we assume as input to our algorithm a guess $\beta'$ for the value of $\broadOPT$ such that $\broadOPT \geq \beta' \geq \frac{\broadOPT}{1 + \varepsilon}$. In~\Cref{appendix:ballpacking}, we show how to obtain an arbitrarily good estimate for $\broadOPT$
via a binary search. Here, we proceed as if we know $\beta'$ and we set $\proxdiameter = \frac{\beta'}{4k^2}$ and $\proxradius = \sqrt{\proxdiameter}$. Note that since $\beta' \leq \broadOPT$,~\Cref{thm:strengthened-ball-growing-theorem} is true for these values of $\proxdiameter$ and $\proxradius$.  
Next, for a set of centers $C \subset V$ and a proximity parameter $r$, we define the auxiliary graph $H_r^C = (C, E_H)$ of $G$ as one in which we have a vertex for each $c \in C$, and an edge $c_ic_j$ exists iff $\proxG{G}{c_i}{c_j} \geq r$.
We can now state the algorithm.

\vspace{6pt}

\noindent\begin{minipage}{\linewidth}
    \begin{algorithm}[H]


      Let $v$ be any vertex in $V$. Initialize $C = \{v\}$. \\
      \While{$\exists \ c \in V$ s.t. $\proxG{G}{c}{C} \in [\proxdiameter\pmin, \proxdiameter)$\label{step:pickw}} { 
        $C := C \cup \{c\}$\\ 
      }
      Construct the auxiliary graph $H_r^C$ of $G$ for $r=\proxdiameter\pmin$. \\
      Find a spanning forest $F$ in $H$. \\
      Find a 2-coloring of $F$ and let $D\subset C$ be the set of vertices in the smaller color class. \label{step:picksmallercolorclass} \\
      Pick any vertex $c \in D$ as the center and let $\hat{S}$ be the edges of a star on $D$ centered at $c$. \\
      \KwRet{$\hat{S}$}
      \caption{\label{alg:ballpacking}}
    \end{algorithm}
  \end{minipage}

  \vspace{6pt}

To analyze Algorithm~\ref{alg:ballpacking}, we begin with a simple lemma about properties that hold true for the set of selected centers.

\begin{restatable}{lemma}{ballpackingtightness}\label{claim:ballpacking-tightness}
Let $c_i$ denote the $i^{th}$ vertex added to $C$ by Algorithm~\ref{alg:ballpacking}. At the end of the {\bf while} loop in Line~\ref{step:pickw}, all of the following properties hold true for the set $C$:
\begin{enumerate}[(i)]
    \item $\forall c_i, c_j \in C, i \neq j, \proxG{G}{c_i}{c_j} < \proxdiameter$, \label{claim:ballpacking-farness-claim}
    \item $\forall c \in C$, $\proxG{G}{c}{C \setminus \{c\}} \geq \proxdiameter\pmin$, and \label{claim:ballpacking-closest-center-claim}
    \item $\forall c_i, c_j \in C, i \neq j, B_{c_i} (\proxradius)$ and $B_{c_j}(\proxradius)$ are disjoint. \label{claim:ballpacking-disjoint-claim} 
\end{enumerate}
\end{restatable}

\begin{proof}
Without loss of generality, assume that $c_i$ was added to $C$ before $c_j$. At the time of adding $c_j$, we have $\proxG{G}{c_i}{c_j} < \proxdiameter$, by definition (Line~\ref{step:pickw} of the algorithm). Thus, the first property holds. 
For property~(\ref{claim:ballpacking-closest-center-claim}), note that at the time of adding $c_i$, $\proxG{G}{c_i}{C} \geq \proxdiameter\pmin$. Since the proximity can only reduce after adding more $c_j$, the property follows. 
For property~(\ref{claim:ballpacking-disjoint-claim}), suppose the claim is false. Let $c_i$ and $c_j$ be two centers such that $u \in B_{c_i}(\proxradius) \cap B_{c_j}(\proxradius)$. Then by the triangle inequality in the implied metric of $G$, $\proxG{G}{c_i}{c_j} \geq \proxG{G}{c_i}{u} \cdot \proxG{G}{u}{c_j} \geq \proxradius^2 = \proxdiameter$ which contradicts property~(\ref{claim:ballpacking-farness-claim}).
\end{proof}

Using this, we show that the algorithm adds at most $(2k+1)$ vertices to $C$.

\begin{restatable}{claim}{ballpackingterminatesclaim}\label{claim:ballpacking-terminates-claim}
At the end of the {\bf while} loop in Line~\ref{step:pickw} of Algorithm~\ref{alg:ballpacking}, $|C| \leq 2k+1$.
\end{restatable}
\begin{proof}
Recall that by~\Cref{thm:strengthened-ball-growing-theorem}, we can write $V = \left( \cup_{i=1}^{2k} F_i \right) \cup U$, where $F_i$ are balls of proximity radius $\proxradius$ and $U$ is a set of proximity diameter $\proxdiameter$. By property~(\ref{claim:ballpacking-farness-claim}) of~\Cref{claim:ballpacking-tightness}, the algorithm chooses at most one $c$ from each of the $F_i$, and at most one $c$ from $U$. This implies that $|C| \le 2k+1$.
\end{proof}

Interestingly, it is not yet clear that the {\bf while} loop in Algorithm~\ref{alg:ballpacking} ends up ``covering'' all the vertices of $V$. Specifically, there may not exist any $c$ that satisfies the condition of the {\bf while} loop, but there could be $v \in V$ that have $\proxG{G}{v}{C} < \proxdiameter \pmin$. The following technical lemma shows that as long as $G$ is connected, this cannot happen (in other words, $\proxG{G}{v}{C}$ values change ``smoothly''). 

\begin{restatable}{lemma}{bucketing}\label{lem:bucketing}
Let $G$ be a connected probabilistic graph, $v$ be a vertex in $G$ and $U \subseteq V\setminus \{v\}$. Suppose $r$ is any parameter such that $\proxG{G}{v}{U} < r \leq 1$. Then there exists a $v' \in V$ that satisfies $\proxG{G}{v'}{U} \in [r\pmin, r)$.
\end{restatable}

The proof is deferred to~\Cref{appendix:ballpacking}. Using~\Cref{lem:bucketing}, we obtain the last step in the analysis.

\begin{restatable}{claim}{ballpackingcompactnessclaim}\label{claim:ballpacking-compactness-claim}
At the end of the {\bf while} loop in Line~\ref{step:pickw} of Algorithm~\ref{alg:ballpacking}, every vertex in $G$ has proximity at least $\proxdiameter$ to $C$.
\end{restatable}
\begin{proof}
If the claim is not true, then $\exists v \in V$ such that $\proxG{G}{v}{C} < \proxdiameter$. If $\proxG{G}{v}{C} \in [\proxdiameter\pmin, \proxdiameter)$, then $v$ satisfies the condition in Line~\ref{step:pickw} of the algorithm, which implies the {\bf while} loop could not have ended. So assume $\proxG{G}{v}{C} < \proxdiameter\pmin$. But by~\Cref{lem:bucketing}, setting $r=\proxdiameter$, $\exists v' \in V$ such that $\proxG{G}{v'}{C} \in [\proxdiameter\pmin, \proxdiameter)$. Thus, $v'$ meets the criteria in Line~\ref{step:pickw}, which again contradicts the fact that the {\bf while} loop has ended. 
\end{proof}

\subsubsection{Proofs of the Main Results,~\Cref{thm:ballpacking-theorem,thm:strengthened-ball-growing-theorem}}\label{sec:proof-of-main-theorem}
First we sketch the proof of~\Cref{thm:ballpacking-theorem}, assuming~\Cref{thm:strengthened-ball-growing-theorem} and the preceding analysis. We show that Algorithm~\ref{alg:ballpacking} satisfies the requirements of the theorem. By~\Cref{claim:ballpacking-terminates-claim}, at the end of the {\bf while} loop in Line~\ref{step:pickw} of the algorithm, $|C| \le 2k+1$. In the auxiliary graph $H_r^C$, for all $c_i, c_j \in C, i \neq j$, we place an edge if $\proxG{G}{c_i}{c_j} \geq \proxdiameter\pmin$. By property~(\ref{claim:ballpacking-closest-center-claim}) of~\Cref{claim:ballpacking-tightness}, for every $c_i \in C$ at least one $c_j$ satisfies this condition. In other words, the degree of every vertex in $H_r^C$  is at least one. Thus, the smaller color class $D$ picked in Line~\ref{step:picksmallercolorclass} has at most $k$ vertices and the star centered at $c$ has at most $k-1$ edges. By then using~\Cref{claim:ballpacking-compactness-claim}, property~(\ref{claim:ballpacking-closest-center-claim}) of~\Cref{claim:ballpacking-tightness}, and the triangle inequality in the implied metric of $G$, we conclude that for all $v \in V$, $\proxG{G+\hat{S}}{v}{D} \ge \proxdiameter^2\pmin$. Therefore, the resultant reach is at least
\[ \proxdiameter^4\pmin^4 = \frac{(\broadprime)^4\pmin^{4}}{4^4k^8} \geq \frac{(\broadOPT)^4 \pmin^{4}}{4^4k^8(1 + \varepsilon)^4}. \]

In~\Cref{appendix:ballpacking}, we complete the proof by showing how to repeat Algorithm~\ref{alg:ballpacking} in a binary search fashion to ensure the inequality above holds, i.e., to ensure a good guess for the value of $\broadOPT$.

Now, we conclude this section by proving the key structural result,~\Cref{thm:strengthened-ball-growing-theorem}.
Let $S^*$ be the set of edges added to $G$ to obtain $G^*$ and let $c_1, c_2, \dots, c_{\ell}$ (with $\ell \le 2k$) be the endpoints of the edges in $S^*$. Assign each vertex $v \in V$ to its closest $c_i$, breaking ties arbitrarily. We will say that the set of vertices assigned to $c_i$ constitute the \emph{cluster} of $c_i$ and that $c_i$ is the \emph{center} of its cluster. 
We will say that a vertex $v$ is \emph{good} if $\proxG{G}{v}{c} \geq \proxradius$ where $c$ is the center of the cluster to which $v$ belongs in $G^*$, and \emph{bad} otherwise. 

\begin{restatable}{claim}{badverticesclaim}\label{claim:bad-vertices-claim}
The set of all bad vertices has proximity diameter at least $\proxdiameter$ in $G$.
\end{restatable}

\begin{proof}
Let $u$ and $v$ be any two bad vertices. Consider the set of all paths from $u$ to $v$ in $G^*$ and form equivalence classes based on the leading vertex of the first edge from $S^*$ and the trailing vertex of the last edge from $S^*$ on the path. Since the paths are simple, for every class these vertices must be unique. Thus, there are at most $\ell \cdot (\ell - 1) + 1 < 4k^2$ equivalence classes --- including the empty class (i.e., the class corresponding to no new edges). At least one of these equivalence classes must have contribution $> \frac{\broadOPT}{4k^2} = \proxdiameter$ to the proximity. If this contribution comes from the empty class, we are already done, since it means that $\proxG{G}{u}{v} \ge \proxdiameter$. Thus, let us assume that every path in the class with the largest contribution has at least one edge from $S^*$. Let $c_i$ be the leading vertex of the first edge from $S^*$ and $c_j$ the trailing vertex of the last edge from $S^*$ corresponding to this class. Let $P_{ij}$ represent the set of simple paths (between $u$ and $v$) in this class. Then, applying the Splitting Lemma (\Cref{lemma:splitting}) at $c_i$, we have that 
\[
\proxdiameter = \frac{\broadOPT}{4k^2} < \Pr[P_{ij}] \leq \Pr[P_{ij}[u, c_i]] \cdot \Pr[P_{ij}[c_i, v]]. 
\]
Applying the Splitting Lemma once again, to $P_{ij}[c_i, v]$ at $c_j$, we have that 
\[
\Pr[P_{ij}[c_i, v]] \leq \Pr[P_{ij}[c_i, c_j]] \cdot \Pr[P_{ij}[c_j, v]].
\]
Combining the two inequalities, we get that $\proxdiameter < \Pr[P_{ij}[u, c_i]] \cdot \Pr[P_{ij}[c_i, c_j]] \cdot \Pr[P_{ij}[c_j, v]] \leq \Pr[P_{ij}[u, c_i]] \cdot \Pr[P_{ij}[c_j, v]]$. This implies that  either $\Pr[P_{ij}[u, c_i]] > \proxradius$ or $\Pr[P_{ij}[c_j, v]] > \proxradius$.
By construction, no new edges appear along paths in $P_{ij}[u, c_i]$ nor along paths in $P_{ij}[c_j, v]$, so it follows that either $\proxG{G}{u}{c_i} > \proxradius$ or $\proxG{G}{c_j}{v} > \proxradius$, which is a contradiction since both $u$ and $v$ were assumed to be bad. Thus, $\proxG{G}{u}{v} \geq \proxdiameter$. This completes the proof of the claim.
\end{proof}

Once we have the claim,~\Cref{thm:strengthened-ball-growing-theorem} follows quickly. The set of bad vertices forms the set $U$, and the good vertices are covered by the balls $F_i$ centered at $c_i$. Since $\ell \le 2k$, the theorem follows.

\section{A Linear Approximation for \SRI{}}\label{sec:star-submodular}

Here, we show how to combine the Splitting Lemma, submodular optimization, and an additional existential result (\Cref{thm:star-structure}) to achieve a linear approximation for the most general variant of~\SRI{}.
We defer technical details to~\Cref{appendix:submodularity}.

\begin{restatable}{theorem}{submodalgtheorem}\label{thm:main-klogn-algo}
  Let $\reach{G^*}{\sourceV}$ be the optimum objective value for the \SRI{} problem given a probabilistic graph $G$ and parameter $k$, as in~\Cref{sec:preliminaries}. For any constant $\varepsilon > 0$, there is a polynomial time algorithm that finds a set $S$ of $O(k\log n)$ edge additions such that \[\reach{G+S}{\sourceV} \geq \frac{\reach{G^*}{\sourceV}\pmin^{2+\varepsilon}}{(1+\varepsilon)12k^2}.\]
\end{restatable}

Our idea is to define an appropriate function and use its submodularity to obtain an algorithm. Unfortunately, the natural candidate ---improvement in reach when a set of edges is added--- is not a submodular function (see~\Cref{obs:broadcast-not-submodular} in~\Cref{appendix:submodularity}).
However, when we restrict edge additions to edges out of a ``center'' vertex $u$, submodularity holds for the proximities of $u$ to other vertices.
Specifically, let $E^u$ be $\{uv : v \in V\}$. For a subset $S \subseteq E^u$, define 
$g_v (S) :=  \log \proxG{G+S}{v}{u}$.

\begin{restatable}{lemma}{submodularitylemma}\label{lem:submodularity}
  For any graph $G$, vertices $u, v$, the function $g_v : 2^{E^u} \mapsto \mathbb{R}$ is monotone and submodular.
\end{restatable}

By subtracting a term corresponding to the proximity in $G$ (without any edge additions), $g_v$ lets us measure the ``gain'' in the proximity between $u$ and $v$ provided by adding edges $S$.
In the rest of this section, we will show how to leverage this submodularity to obtain an algorithm which searches for \emph{star solutions}, i.e., solutions in which every added edge is incident on a shared endpoint $u$.
To achieve the approximation guarantee of~\Cref{thm:main-klogn-algo}, we first need to ensure that there exists a near-optimal star solution.
We prove the following in~\Cref{appendix:star-lemma}.

\begin{restatable}[The Star Lemma]{lemma}{starstructuretheorem}\label{thm:star-structure}
  Let $(G = (V, E, \{\alpha_e\}), \sourceV, k)$ be an instance of \SRI{} and let $S \subseteq V^2 \setminus E$ be a solution of size $k$ achieving $\beta = \reach{G+S}{\sourceV}$. Let $V(S)$ be the endpoints of $S$.
  Then the solution $S_{star}$ of size at most $2k - 1$ formed by creating a star on $V(S)$ (with an arbitrary endpoint $u$ chosen as the center) achieves $\reach{G+S_{star}}{\sourceV} \geq \frac{\beta\pmin^2}{12k^2} \vcentcolon = \beta_{star}$.
  Furthermore, for every vertex pair $(i, j) \in \sourceV \times V$ with $\proxG{G}{i}{j} < \beta_{star}$, the contribution in $G + S_{star}$ of paths from $i$ to $j$ using the vertex $u$ is at least $\beta_{star}$.
\end{restatable}

Our algorithm is based on a potential function that captures the proximity between all relevant vertex pairs $(i,j) \in \sourceV \times V$. Before describing it, we introduce an auxiliary function that is defined for one pair $(i,j) \in \sourceV \times V$, a given vertex $\centerv$, a ``target'' proximity value $\beta'$, and a set of edges $T$ all incident to $u$:
\[  \mu (i, j; T, \beta') := - \log \proxG{G+T}{\centerv}{i}  - \log \proxG{G+T}{\centerv}{j} + \log \beta'. \]

Note that by Lemma~\ref{lem:submodularity}, for any $(i,j)$ we have that $\mu(i,j; T, \beta')$ is submodular. The algorithm assumes as parameters a vertex $\centerv$ (we will need to run the algorithm for every choice of $\centerv$), a parameter $\varepsilon > 0$, and a target reach value $\beta'$. The algorithm is as follows:

\vspace{6pt}

\noindent\begin{minipage}{\linewidth}
    \begin{algorithm}[H]
      \SetKwData{Left}{left}
      \SetKwData{This}{this}
      \SetKwData{Up}{up}

      \SetKwFunction{Union}{Union}
      \SetKwFunction{FindCompress}{FindCompress}
      
      Initialize $S\su{0} = \emptyset$, $t=0$ \\
      Define $A = \{(i, j) \in \sourceV \times V: \proxG{G}{i}{j} < \beta' \}$; (call these \emph{active pairs}) \\
      For any $T$, define $\Psi (T) = \sum_{(i,j) \in A} \max \{0, \mu(i,j; T, \beta') \}$ \\
      \While{$\Psi(S\su{t}) > \log (1/\pmin^\varepsilon)$}{
        Find edge $e$ incident to $u$ that minimizes $\Psi(S\su{t-1} \cup \{e\})$ \\
        Increment $t$; define $S\su{t} = S\su{t-1} \cup \{e\}$
      }
      Return $S\su{t}$
      \caption{\label{alg:submod}}
    \end{algorithm}
  \end{minipage}

  \vspace{6pt}

While submodularity will ensure that the drop in potential is significant at every step, it turns out that it is non-trivial to prove that the optimal subset achieves low potential!
This is where we use the Splitting Lemma (\Cref{lemma:splitting}):

\begin{restatable}{lemma}{optisgood}\label{lem:opt-is-good}
Let $u$ be the center of the star solution $S_{star}$ obtained from~\Cref{thm:star-structure} and $\beta_{star}$ be the corresponding reach value. Suppose $(i, j) \in \sourceV \times V$ such that $\proxG{G}{i}{j} < \beta_{star}$. Then $\proxG{G+S_{star}}{i}{u} \cdot \proxG{G+S_{star}}{j}{u} \ge \beta_{star}$.
\end{restatable}
\begin{proof}
  Consider the set of all paths from $i$ to $j$ in $G + S_{star}$. Let $P_u$ be the subset of these paths which use the vertex $u$. By~\Cref{thm:star-structure}, $\Pr[P_u] \geq \beta_{star}$.   
  We now apply the Splitting Lemma (\Cref{lemma:splitting}), and complete the proof by noting that (by definition of proximity) $\proxG{G+S_{star}}{i}{u} \cdot \proxG{G+S_{star}}{j}{u} \geq \Pr[P_u[i, u]]\cdot\Pr[P_u[u, j]]$.
\end{proof}

Lemma~\ref{lem:opt-is-good} immediately implies that for $\Psi$ as defined in Algorithm~\ref{alg:submod}, for the star edges $S_{star}$ and $\beta' \leq \beta_{star}$, we have
$\Psi (S_{star}) = 0$.  This is because for every active pair $(i,j)$, the lemma implies that $\mu(i,j; S_{star}, \beta') \le 0$. We then have the following guarantee for the algorithm, at every step $t$.

\begin{restatable}{lemma}{potentialdeclinelemma}\label{lemma:potential-decline-lemma}
Let $S\su{t}$ be the set of added edges as defined in Algorithm~\ref{alg:submod}. For any $t \ge 1$, we have
$\Psi (S\su{t}) \le \Psi (S\su{t-1}) \left( 1 - \frac{1}{2k} \right)$.
\end{restatable}

\Cref{lemma:potential-decline-lemma} follows from a standard submodularity argument, but some extra care must be taken because of the $\max$ appearing in the definition of $\Psi$. 
The preceding lemmas give us all the tools we need to analyze Algorithm~\ref{alg:submod} and prove~\Cref{thm:main-klogn-algo}.
We again defer the details to~\Cref{appendix:submodularity}.

\section{Lower Bounds}\label{sec:hardness-bi}

In this section we give lower bounds for \BI{} and \RI{} which are special cases of \SRI{}. We note that, prior to our work, even \cclass{NP}-hardness was open~\cite{bashardoust2023reducing}.
We also note that our lower bounds hold even under uniform edge-sampling probabilities. For the rest of this section, we denote this uniform probability by $\alpha$.
Our results are as follows. 

\begin{restatable}{theorem}{broadcasthardness}\label{thm:broadcasthardness}
Let $\broadOPT$ be the optimum objective value for the \BI{} problem given a probabilistic graph $G$ and parameter $k$, as in~\Cref{sec:preliminaries}. For any constants $c' \ge 1$ and $\varepsilon>0$, unless $\cclass{P}=\cclass{NP}$,
there is no polynomial-time algorithm which can guarantee reach at least
$(\broadOPT{})^{\frac{6}{5}-\varepsilon}$ while adding at most $c'k$ edges.
\end{restatable}

\begin{restatable}{theorem}{sshardness}\label{thm:sshardness}
    Let $\reach{G^*}{\sourcev}$ be the optimum objective value for the \RI{} problem given a probabilistic graph $G$, parameter $k$, and source vertex $\sourcev$, as in~\Cref{sec:preliminaries}. For any constants $\varepsilon > 0$ and $c \geq 1$, unless $\cclass{P}=\cclass{NP}$, there is no
    polynomial-time algorithm which can guarantee reach at least $\big(\reach{G^*}{\sourcev})^{\frac{4}{3}-\varepsilon}$ while adding at most $ck$ edges. 
\end{restatable}

Here, we will present the construction and a sketch of the analysis for~\Cref{thm:broadcasthardness}. The full version is in~\Cref{appendix:hardness-bi}.
The construction and analysis for~\Cref{thm:sshardness} are similar, and can be found in~\Cref{appendix:ss-hardness}.

The reduction is from a variant of the \textsc{Set Cover} problem. Specifically, we rely on the following hardness assumption ~\cite{feige1998,feige2010}:

\begin{restatable}{assumption}{hardgapsetcover}\label{asmp:hardgsc}[\gsc{}]
    Let $c \ge 1$ be any constant. Given a collection of $n$ sets $S_1,S_2, \dots,  S_n \subseteq [m]$, it is \cclass{NP}-hard to distinguish between two following cases:
    \begin{itemize}
    \item \textbf{YES:} There are $k$ sets in the collection whose union is $[m]$. 
    \item \textbf{NO:} There exists a $\delta < 1$ such that the union of any $ck$ sets can cover at most $\delta\cdot m$ elements.
    \end{itemize}
    Furthermore, the hardness holds even when $m = \Theta(n)$ and $|S_i| = O(\text{polylog}(n))$.
\end{restatable}

We remark that Assumption~\ref{asmp:hardgsc} likely holds even with $c = (\log n)^{1-o(1)}$. With this stronger assumption, our hardness results can be improved to nearly match our algorithmic bounds. We omit the details. 

\begin{proof}[Proof Sketch of~\Cref{thm:broadcasthardness}] 
Our reduction from \textsc{Gap Set Cover} is as follows.

\vspace{6pt}

\noindent \underline{\textbf{Instance:}}  Given an instance of 
\gsc{} consisting of a collection of $n$ sets $S_1,S_2, \dots,  S_n \subseteq [m]$, we construct a \BI{} instance $(G=(V,E,\{\alpha_e\}),k)$ with $\pmin = \pmax = \alpha $ as follows:

We create a graph $G$ with a \emph{pivot vertex} $p$, vertices $s_i$ corresponding to sets $S_i$ (called \emph{set vertices}) and vertices $e_i$ corresponding to elements $i \in [m]$ (called \emph{element vertices}). Between every pair of set vertices $s_i, s_j$, we add a path of length $l$, where $l$ is an even integer parameter whose value will be specified later. These paths are mutually disjoint, and so there are $\binom{n}{2} (l-1)$ vertices along the paths. We call these \emph{set-set internal vertices}. Next, we add a path of length $l$ between $s_i$ and $e_j$ for all $j \in S_i$. That is, we connect a set vertex $s_i$ to all the element vertices $e_j$ corresponding to elements $j\in S_i$. Once again, these paths are all mutually disjoint. We call the vertices on the paths the \emph{set-element internal vertices}. Finally, we connect the pivot $p$ to each set vertex via mutually disjoint paths of length $l$. We call the internal vertices along these paths \emph{pivot-set internal vertices.}

\begin{figure}[t]\label{fig:hardness-mainpaper}
    \centering
    \includegraphics[width=\textwidth]{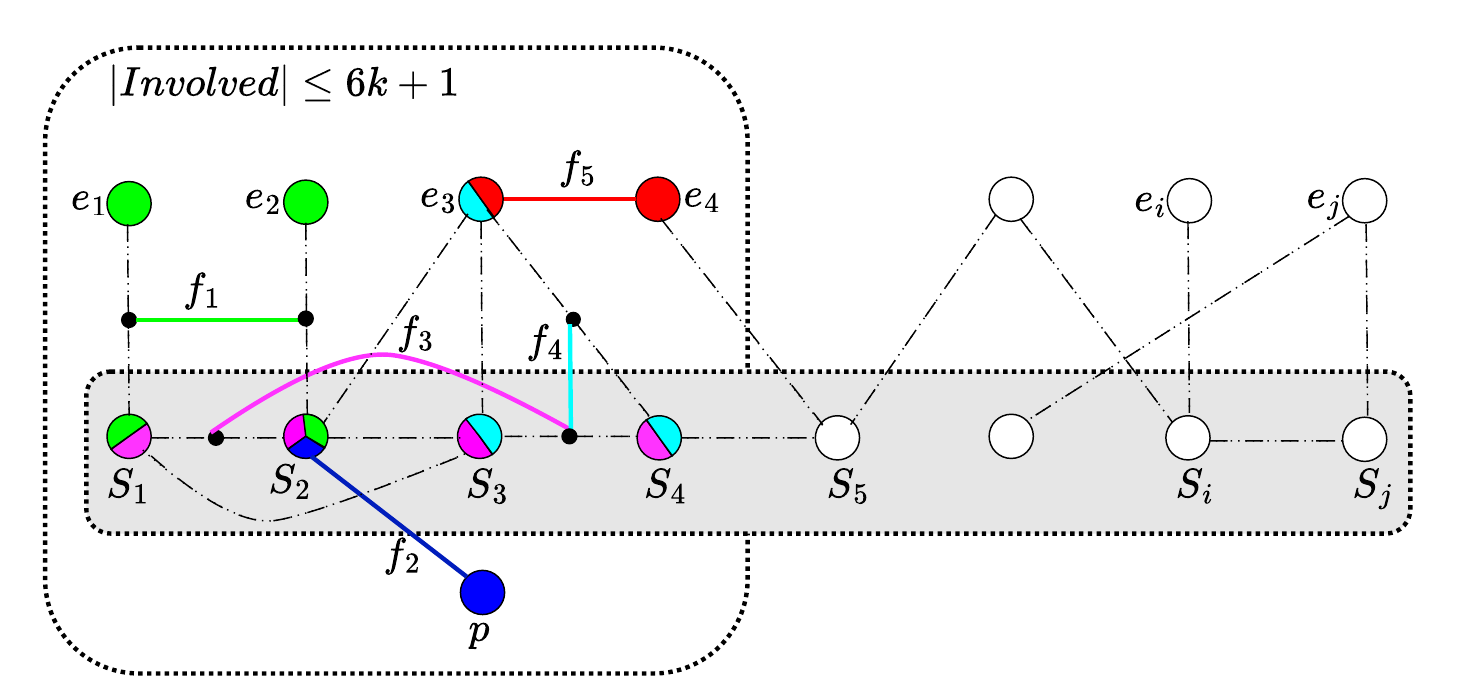}
    \caption{An example instance of \BI{} constructed by our hardness reduction, after adding $k$ edges. We use  $s_i$ for set vertices, $e_i$ for element vertices, and dot-dash lines to indicate the $l$-length paths joining these. Internal vertices which are the endpoint of a new edge are shown as smaller black vertices along these paths. Among the $k$ added edges (solid), we highlight examples of most types of possible edge additions as $f_1, \ldots, f_5$, and indicate using the associated color the set and element vertices they involve. For clarity, several types of edge additions incident to $p$ are omitted, along with the $l$-length paths between $p$ and each set vertex. Finally, $e_i$ and $e_j$ are a pair of element vertices outside the involved set that have a distance of at least $3l$ between them. These are the vertices that we use to prove (in the no-case) that the reach is at most $\alpha^{3l - \varepsilon}$.}    
\end{figure}
Now we argue about the reach that can be achieved after adding $k$ edges to such a graph $G$.

\vspace{6pt}

\noindent \underline{\textbf{YES-case:}} 
Consider the cover of $[m]$ that uses only $k$ sets. Now consider adding $k$ edges between the pivot vertex $p$ and the set vertices corresponding to the sets in the cover.
We claim that between any two vertices in the resulting graph, there is a path of length at most $2.5l + 1$, implying that $\broad{G} \geq \alpha^{2.5l+1}$.
The proof follows by a straightforward case analysis.

\vspace{6pt}

\noindent \underline{\textbf{NO-case:}}
In this case,~\Cref{asmp:hardgsc} says that the union of any $ck$ sets among $S_1, \ldots, S_n$ covers at most $(1 - \Omega(1))m$ elements (for any constant $c \geq 1$).
Now consider adding $r = c'k$ edges, where $r \leq ck/8$. Let $E'$ denote the set of added edges and let
$G'$ be the graph obtained from $G$ by adding the edges in $E'$.

We define a set of \emph{involved} set and element vertices as follows. A set vertex $s_i$ is said to be involved in edge $uv \in E'$
if for $x \in \{u, v\}$ we have (a) $s_i = x$, (b) $x$ is a set-set internal vertex and one of the end-points of the corresponding path
(the one containing $x$) is $s_i$, or (c) $x$ is a set-element internal vertex and the corresponding set vertex is $s_i$.
Analogously, we say that an element vertex $e_i$ is involved in edge $uv \in E'$ if for $x \in \{u, v\}$
either (a) $e_i = x$ or (b) $x$ is a set-element internal vertex and $e_i$ is the corresponding element.
We also generalize the notation slightly and say that a vertex is involved in the set of edges $E'$ if it is
involved in at least one of the edges in $E'$. The main claim due to our choice of parameters is the following.

\begin{claim}\label{claim:no-case-claim}
    There exist element vertices $e_i, e_j$ such that (a) neither is involved in $E'$, (b) none of the sets containing $i, j$
    are involved in $E'$, and (c) no $S_t$ contains both $i$ and $j$.
\end{claim}

The claim follows by noting that every $e \in E'$ involves no more than~$4$ set and~$2$ element vertices.
Thus, there is a set $I$ of no more than $6r$ set vertices such that all the set \emph{and} element vertices
involved in $E'$ are either in $I$ or are covered by $I$. From this, the claim can be proved by using our assumption in
the NO-case, followed by a counting argument.

The remainder of the proof argues that for $e_i, e_j$ satisfying~\Cref{claim:no-case-claim}, by choosing $l$ large enough
we can ensure that $\proxG{G+E'}{e_i}{e_j} < \alpha^{3l - \varepsilon}$, for any $\varepsilon > 0$. This is done by
first contracting the involved set and element vertices (and paths between them) into a ``hub'', followed by counting
paths of different lengths in the resulting graph. The details are somewhat technical; see~\Cref{appendix:hardness-bi}.
Since the reach in the YES-case is at least $\alpha^{2.5l + 1}$, the desired gap follows.

\end{proof}
\section{Generalizing our Algorithms}\label{sec:subset-source}

Note that the main result of~\Cref{sec:star-submodular} (\Cref{thm:main-klogn-algo}) applies to all variants of our problem. Meanwhile, the main result of~\Cref{sec:polynomial-approximation} (\Cref{thm:ballpacking-theorem}) applies only to \BI{}, i.e., the case where $\sourceV = V$.
Our task in this section will be to adapt those techniques to achieve polynomial approximations for \RI{} and for \SRI{}.
We begin with the former, for which we will show nearly the same guarantee as given by~\Cref{thm:ballpacking-theorem}.

\begin{restatable}{theorem}{ssballpackingtheorem}\label{claim:ss-ballpacking-theorem}
    Let $\reach{G^*}{\sourcev}$ be the optimum objective value for \RI{} given a probabilistic graph $G$, parameter $k$, and source vertex $\sourcev$, as in~\Cref{sec:preliminaries}.
    For any constant $\varepsilon > 0$, there exists a polynomial-time algorithm which finds a set $S$ of at most $k$ edge additions such that \[\reach{G+S}{\sourcev} \geq \frac{\reach{G^*}{\sourcev}^4\pmin^2}{(2k+2)^4(1+\varepsilon)^4}.\] 
\end{restatable}

The key insight is the following analogue to~\Cref{thm:strengthened-ball-growing-theorem}.
\begin{restatable}{theorem}{ssballgrowingtheorem}\label{thm:ss-ball-growing-theorem}
Let $G'$ be the graph obtained by augmenting $G$ with a set of $k$ edges $S$, let $V(S)$ be the endpoints of $S$, $\sourcev$ be the source vertex, and let $\reachprime = \reach{G'}{\sourcev}$ represent the reach of $\sourcev$ in $G'$. Let $\proxradius = \frac{\reachprime}{(2k+2)}$. Let $v$ be any vertex in $G$. Then either $\proxG{G}{\sourcev}{v} \geq \proxradius$ or $\proxG{G}{v}{V(S)} \geq \proxradius$.
\end{restatable}
We prove~\Cref{thm:ss-ball-growing-theorem} in~\Cref{appendix:subset-source}.~\Cref{thm:ss-ball-growing-cor} is an immediate consequence.
\begin{restatable}{corollary}{ssballgrowingcor}\label{thm:ss-ball-growing-cor}
Let $G$, $\sourcev$, and $\proxradius$ be as in~\Cref{thm:ss-ball-growing-theorem}. Then one ball of proximity radius $\proxradius$ centered at $\sourcev$ and at most $2k$ additional balls of proximity radius $\proxradius$ cover $G$. 
\end{restatable}

We now adapt Algorithm~\ref{alg:ballpacking} to the single-source setting.
Let $\proxdiameter = \proxradius^2$ where $\proxradius$ is as defined in~\cref{thm:ss-ball-growing-theorem}. 
As in~\Cref{sec:ballpacking}, we must assume as input a guess $\beta'$ for $\reach{G^*}{\sourcev}$, with $\reach{G^*}{\sourcev} \geq \beta' \geq \reach{G^*}{\sourcev}/(1+\varepsilon)$.
An arbitrarily good guess can once again be found via a binary search procedure.
We call Algorithm~\ref{alg:ballpacking} with a single change: we pick the source vertex as the first vertex to add to $C$, i.e., $v=\sourcev$.
Then~\cref{claim:ballpacking-tightness} and~\Cref{claim:ballpacking-compactness-claim} are true for the set $C$ at the end of the \textbf{while} loop of Line~\ref{step:pickw}.

Furthermore, in Line~\ref{step:picksmallercolorclass} of Algorithm~\ref{alg:ballpacking}, if $\sourcev \notin D$, we add $\sourcev$ to $D$ and set $\sourcev$ to be the center of the star. We can now prove the following claim that is analogous to~\cref{claim:ballpacking-terminates-claim}.

\begin{restatable}{claim}{ssballpackingterminatesclaim}\label{claim:ss-ballpacking-terminates-claim}
At the end of the \textbf{while} loop in Line~\ref{step:pickw}, $|C| \leq 2k+1$.
\end{restatable}

\begin{proof}
Suppose the claim is not true. Every vertex in $C$ is either $\sourcev$ or by~\Cref{thm:ss-ball-growing-theorem} has a new edge end point within the ball of proximity radius $\proxradius$ centered at the vertex. Since the balls of radius $\proxradius$ centered at vertices in $C$ are disjoint by property~(\ref{claim:ballpacking-disjoint-claim}) of~\Cref{claim:ballpacking-tightness}, and since there are at most $2k$ new edge endpoints, the claim follows.
\end{proof}

This completes the toolbox needed to prove~\Cref{claim:ss-ballpacking-theorem}. The proof is similar to that of~\Cref{thm:ballpacking-theorem}; see~\Cref{appendix:subset-source}.

We will now leverage this algorithm to get a polynomial approximation for \SRI{}, albeit with a slightly worse approximation factor.
We need only repeat the previous algorithm, choosing an arbitrary vertex in $\sourceV$ as $\sourcev$.

\begin{restatable}{theorem}{sr-ballpackingtheorem}\label{claim:sr-ballpacking-theorem}
    Let $\beta_{G^*}(\sourceV)$ be the optimum objective value for \SRI{} given a probabilistic graph $G$, parameter $k$, and source set $\sourceV$, as in~\Cref{sec:preliminaries}. For any constant $\varepsilon > 0$, there exists a polynomial-time algorithm which finds a set $S$ of at most $k$ edge additions such that
    \[\beta_{G+S}(\sourceV) \geq \frac{{\reach{G^*}{\sourceV}}^8\pmin^{4}}{(2k+2)^8(1+\varepsilon)^8}.\]
\end{restatable}
\begin{proof}
Let $u \in \sourceV$. 
Let $S_u$ be the optimal set of edges for \RI{} on $G$ with source-vertex $\sourcev = u$, and $\reach{G+S_u}{u}$ be the corresponding optimal reach.
Let $S^*$ be the optimal set of edges and $\reach{G+S^*}{\sourceV}$ denote the optimal reach for \SRI{} on $G$ with $\sourceV$ as the source set.
Then $\reach{G+S_u}{u} \geq \reach{G+S^*}{u} \geq \reach{G+S^*}{\sourceV}$. 
By~\Cref{claim:ss-ballpacking-theorem}, using the above algorithm for \RI{} with $u$ as the source vertex, we obtain a set of at most $k$ edges $\hat{S_u}$ such that 
\[\reach{G+\hat{S_u}}{u} \geq \frac{{\reach{G+S_u}{u}}^4\pmin^{2}}{(2k+2)^4(1+\varepsilon)^4} \geq \frac{{\reach{G+S^*}{\sourceV}}^4\pmin^{2}}{(2k+2)^4(1+\varepsilon)^4}.\]
By the triangle inequality, for all $(v, w) \in \sourceV \times V$, we have $\proxG{G+\hat{S_u}}{v}{w} \geq \proxG{G+\hat{S_u}}{v}{u}\cdot\proxG{G+\hat{S_u}}{u}{w}$.
Substituting, we have $\proxG{G+\hat{S_u}}{v}{w} \geq \frac{{\reach{G+S^*}{\sourceV}}^8\pmin^{4}}{(2k+2)^8(1+\varepsilon)^8}$, as desired.
\end{proof}

We conclude by noting that in the special case of~\RI{}, we can improve upon the generic guarantee given by~\Cref{thm:main-klogn-algo}.
The key insight is that, in the single-source context, we can prove a stronger analogue of the Star Lemma (\Cref{thm:star-structure}):

\begin{restatable}{lemma}{ssstarlemma}\label{lemma:ss-star-lemma}[The Single-Source Star Lemma]
    Let $G'$ be the graph obtained by augmenting a probabilistic graph $G$ with a set of $k$ edges $S$, let $V(S)$ be the endpoints of $S$, $\sourcev$ be the source vertex, and let $\beta' = \reach{G'}{\sourcev}$ represent the reach of $\sourcev$ in $G'$.
    Then the solution $S_{star}$ formed by adding (at most $2k$) edges to form a star on $V(S)$ centered at $\sourcev$ achieves $\reach{G+S_{star}}{\sourcev} \geq \frac{\reach{G'}{\sourcev}\pmin}{(2k+2)\pmax} \vcentcolon = \beta_{star}$.
    Furthermore, for every vertex $v \in V$, if we partition the paths from $\sourcev$ to $v$ in $G+S_{star}$ according to the (at most one) edge from $S_{star}$ used, then there exists an equivalence class with contribution at least $\beta_{star}$.
\end{restatable}

We prove~\Cref{lemma:ss-star-lemma} in~\Cref{appendix:constant-witnesses}. Using the lemma, the following may be achieved either via a direct submodularity-based approach (as in~\Cref{sec:star-submodular}) or via a black-box reduction to \HS{}.
We give the details of the latter approach in~\Cref{appendix:constant-witnesses}.

\begin{restatable}{theorem}{ssalg}\label{thm:ss-klogn}
    Let $\reach{G^*}{\sourcev}$ be the optimum objective value for \RI{} given a probabilistic graph $G$, parameter $k$, and source vertex $\sourcev$, as in~\Cref{sec:preliminaries}.
    For any constant $\varepsilon > 0$, there exists a polynomial-time algorithm which finds a set $S$ of $O(k\log n)$ edge additions such that
    \[\beta_{G+S}(\sourcev) \geq \frac{\reach{G^*}{\sourcev}\pmin}{(2k + 2)(1 + \varepsilon)} .\]
\end{restatable}

\section*{Acknowledgements}
This work was supported in part by the National Science Foundation under award IIS-1956286 to Blair D. Sullivan, under Grant \#2127309 to the Computing Research Association for the CIFellows 2021 Project,
and awards CCF-2008688 and CCF-2047288 to Aditya Bhaskara.

\appendix
\section{Proofs of~\Cref{lemma:ball-growing} and~\Cref{thm:k-center-reduction}}\label{appendix:k-center}
\ballgrowinglemma*
\begin{proof}    
    Consider all paths from $u$ to $v$ in $G+S$ which include at least one new edge (an edge in $S$). We partition these paths according to the first
    new edge to appear along the path. For each of the $k$ equivalence classes, we partition again according to the orientation of the first new edge. That is, for the equivalence class corresponding to new edge $wz$, we partition into paths using $wz$ with $w$ being the leading vertex, i.e., the orientation $w \rightarrow z$, and those for which $z$ is the leading vertex, i.e., the orientation $z \rightarrow w$.
    This procedure produces at most $2k$ equivalence classes in total.

    We now bound the contribution of each equivalence class. Consider the equivalence class defined by new edge $wz$, with $w$ being the leading vertex. Every path in this class
    begins with a segment from $u$ to $w$ which exists in $G$, and then extends via the edge $wz$. Consequently, the contribution
    of this set of paths is at most $\proxG{G}{u}{w} \cdot \alpha_{wz} \leq \proxG{G}{u}{V(S)} \cdot \pmax$. The result follows from applying the union bound on each
    of our $2k$ equivalence classes, plus the contribution of paths which use no new edges. 
\end{proof}

We note that the bound given by~\Cref{lemma:ball-growing} is asymptotically tight; see the example given in~\Cref{fig:tightness-of-ball-growing}.
In this case, $u$ is the center of a spider graph with legs of length $l$. Let $v, q_1, q_2, \ldots, q_k$ be the leaves. Assuming uniform sampling probabilities $\pmax = \pmin = \alpha$, it holds that $u$ has proximity $x = \alpha^l$ to each of $v, q_1, q_2, \ldots q_k$. 
Let $S = \{vq_1, vq_2, \ldots, vq_k\}$. With sufficiently small $\alpha$ and $x$, we have
\[
\proxG{G+S}{u}{v} = 1- (1-x\alpha)^{k}(1-x) \approx 1-(1 -kx\alpha)(1-x) =x+\Omega(xk\alpha).
\]

\begin{figure}[H]
    \scalebox{.75}{
    \begin{minipage}[t]{.4\textwidth}
    \begin{tikzpicture}
        \node[circle,fill,inner sep=2pt,label=left:$u$] (u) at (-1,-0.5) {};
        \node[circle,fill,inner sep=2pt] (11) at (-.5,0.30) {};
        \node[circle,fill,inner sep=2pt] (12) at (0,0.30) {};
        \node[circle,fill,inner sep=2pt] (13) at (0.5,0.30) {};
        \node[shape=circle,draw=black] (1n) at (1.85,0.20) {$q_1$};
        \node[circle,fill,inner sep=2pt] (21) at (-.5,-.55) {};
        \node[circle,fill,inner sep=2pt] (22) at (0,-.55) {};
        \node[circle,fill,inner sep=2pt] (23) at (0.5,-.55) {};
        \node[shape=circle,draw=black] (2n) at (1.95,-0.75) {$q_2$};
        \node[circle,fill,inner sep=2pt] (k1) at (-.5,-1.5) {};
        \node[circle,fill,inner sep=2pt] (k2) at (0,-1.5) {};
        \node[circle,fill,inner sep=2pt] (k3) at (0.5,-1.5) {};
        \node[shape=circle,draw=black] (kn) at (1.90,-1.65) {$q_k$};
        \node[circle,fill,inner sep=2pt] (l1) at (-.5,1) {};
        \node[circle,fill,inner sep=2pt] (l2) at (0,1) {};
        \node[circle,fill,inner sep=2pt] (l3) at (0.5,1) {};
        \node[shape=circle,draw=black] (v) at (1.75,1.20) {$v$};
        \draw (11) -- (u);
        \draw (21) -- (u);
        \draw (k1) -- (u);
        \draw (l1) -- (u);
        \draw[dashed] (1n) to[out=90, in=270] (v);
        \draw[dashed] (2n) to[out=30, in=330] (v);
        \draw[dashed] (kn) to[out=0, in=30] (v);

        \draw (11) -- (12) -- (13);
        \draw (21) -- (22) -- (23);
        \draw (k1) -- (k2) -- (k3);
        \draw (l1) -- (l2) -- (l3);
        
        \draw[dotted] (13) -- (1n);
        \draw[dotted] (23) -- (2n);
        \draw[dotted] (k3) -- (kn);
        \draw[dotted] (l3) -- (v);
        \end{tikzpicture}
    
    \end{minipage}
    } 
    \begin{minipage}[b]{.65\textwidth}
        
        \caption{An example demonstrating the tightness of the bound given in~\Cref{lemma:ball-growing}. The vertex $u$ is connected to $v$ and each $q_i$ via disjoint paths of length $l$.
        The edge addition set $S$ consists of (dashed) edges between $v$ and each $q_i$. 
        \bigskip
        }
        \label{fig:tightness-of-ball-growing}
    \end{minipage}
    
\end{figure}

\twokballcorollary*
\begin{proof}
    Let $S^*$ be the optimal set of $k$ edge additions, forming $G^* = G + S^*$ with $\broad{G^*} = \broadOPT$.
    Let $V(S^*)$ be the endpoints of the edges in $S^*$.
    If any vertex has reach at least $\broadOPT/(1+2k\pmax)$ in $G$, then the result is trivially true. Thus, for each vertex $u$, we may assume the existence of a vertex $v$ with $\proxG{G}{u}{v} < \broadOPT/(1 + 2k\pmax)$.
    If $\proxG{G}{u}{V(S^*)} < \broadOPT/(1 + 2k\pmax)$, then by~\Cref{lemma:ball-growing} we have $\proxG{G^*}{u}{v} < \broadOPT$, a contradiction.
    Then for every vertex $u$, $\proxG{G}{u}{V(S^*)} \geq \broadOPT/(1+2k\pmax)$. Then in the implied metric, $V$ is covered by at most $2k$ balls of radius $-\log (\broadOPT/(1+2k\pmax))$ centered at the vertices of $V(S^*)$.
\end{proof}

\kcenterreduction*
\begin{proof}
    Since~\Cref{cor:2kballs} guarantees the existence of $2k$ balls of radius $r = -\log \frac{\broadOPT}{1 + 2k\pmax}$ that cover $V$, 
    we can find a set of $2k$ centers $C = \{c_1, c_2,\dots,c_{2k}\}$ such that every vertex in $V$ is within distance $2r$ of some center in $C$ by using a $2$-approximation algorithm for \MtwoKCENTER{}~\cite{gonzalez1985clustering,hochbaum1985best}.
    We then add the $2k-1$ edges $c_1c_2, c_1c_3, \ldots, c_1c_{2k}$. We claim that the resulting probabilistic graph $G'$ formed by adding these $2k-1$ edges has reach at least
    $\frac{(\reachOPT)^4\pmin^2}{(1+2k\pmax)^4}$. To see this, consider two arbitrary vertices $u$ and $v$. By our construction of $C$, there exist (possibly non-distinct) centers 
    $c_u, c_v$ such that $\proxG{G'}{u}{c_u}, \proxG{G'}{v}{c_v} \geq \frac{(\broadOPT)^2}{(1 + 2k\pmax)^2}$. 
    Due to our edge additions, $\proxG{G'}{c_u}{c_1}, \proxG{G'}{c_v}{c_1} \geq \pmin$. Then we have
    \[
        \proxG{G'}{u}{v} \geq \proxG{G'}{u}{c_u}\cdot\proxG{G'}{c_u}{c_1}\cdot\proxG{G'}{c_1}{c_v}\cdot\proxG{G'}{c_v}{v} \geq \frac{(\reachOPT)^4\pmin^2}{(1+2k\pmax)^4}.
    \]
\end{proof}
\section{Full proof of the Splitting Lemma}\label{appendix:splitting-lemma}

The Splitting Lemma, stated formally in Lemma~\ref{lemma:splitting}, is a fundamental structural result that we invoke repeatedly throughout the paper.
 It captures a crucial aspect of proximity behavior in the context of probabilistic graphs. Specifically, while the metric-like property of proximity implies that $\proxG{G}{i}{j} \ge \proxG{G}{i}{u} \cdot \proxG{G}{j}{u}$ for all $i, j, u \in V$, this direction reflects a standard triangle inequality. 
 A more subtle question arises in the reverse setting: suppose the proximity $\proxG{G}{i}{j}$ is largely explained by paths that pass through an intermediate vertex $u$. What can then be said about $\proxG{G}{i}{u}$ and $\proxG{G}{j}{u}$ individually? 
 The Splitting Lemma addresses this question directly, providing a structural lower bound on the product $\proxG{G}{i}{u} \cdot \proxG{G}{j}{u}$ in terms of the total contribution of paths from $i$ to $j$ that pass through $u$. 
 Our proof leverages correlation inequalities, most notably the van den Berg-Kesten inequality, to handle dependencies induced by overlapping paths and to formally reason about the decomposition of such contributions.

  \begin{lemma}\label{lemma:bk}(Van den Berg-Kesten inequality~\cite{BK})
    Let $G = (V, E, \{\alpha_e\})$ be a probabilistic graph and let $L$ be the set of all possible subgraphs of $G$ formed by including each edge $e \in E$ independently with probability $\alpha_e$. 
    Then, for any two increasing events $A, B \subseteq L$ (i.e., events that are preserved under the addition of edges: if $H \in A$ and $H \subseteq H'$, then $H' \in A$), we have the following inequality
    \[
    \Pr[A \circ B] \leq \Pr[A] \cdot \Pr[B]
    \]
    where $A \circ B$ denotes the \emph{disjoint occurence} of $A$ and $B$ in the sampled graph. That is, there exist disjoint subsets of edges $E_A, E_B \subseteq E$ in the sampled graph such that the occurrence of $A$ is supported by the edges in $E_A$ and the occurrence of $B$ is supported by the edges in $E_B$, and moreover, $E_A \cap E_B = \emptyset$. 
    \end{lemma}

\splittinglemma*
\begin{proof}
  Let $\cE_i$ be the event that at least one path from the set $P_u[i,u]$ exists in the sampled subgraph, and let $\cE_j$ be the event that at least one path from the set $P_u[u,j]$ exists.
  We first observe that both $\cE_i$ and $\cE_j$ are increasing events because if a sampled subgraph $H$ contains a path from $i$ to $u$ (i.e., $H \in \cE_i$), then any supergraph $H' \supseteq H$ also contains that path, and hence $H' \in \cE_i$. The same reasoning applies to $\cE_j$.
  
  Next, we define the event $\cE_i \circ \cE_j$ to be the disjoint occurrence of $\cE_i$ and $\cE_j$: that is, there exist paths $p_i \in P_u[i,u]$ and $p_j \in P_u[u,j]$ in the sampled subgraph such that $p_i$ and $p_j$ are edge-disjoint.
  Since both $\cE_i$ and $\cE_j$ are increasing events, we can apply the van den Berg-Kesten inequality (Lemma~\ref{lemma:bk}) to obtain:
  \[
  \Pr[\cE_i \circ \cE_j] \leq \Pr[\cE_i] \cdot \Pr[\cE_j].
  \]

  We now relate the event $\cE_i \circ \cE_j$ to the existence of a path from $P_u$. 
  Let $p$ be a path in $P_u$. Then $p$ can be uniquely decomposed into two subpaths $p[i,u]$ and $p[u,j]$ consisting of the prefix from $i$ to $u$ and the suffix from $u$ to $j$, respectively. Note that $p[i,u] \in P_u[i,u]$ and $p[u,j] \in P_u[u,j]$. Moreover, since $p$ is simple, $p[i,u]$ and $p[u,j]$ must be edge disjoint (in fact, they must be vertex disjoint except for the vertex $u$). This implies that $p \in \cE_i \circ \cE_j$. Since this is true for every path $p \in P_u$, we have that $P_u \subseteq \cE_i \circ \cE_j$. Thus, 
  \[  
  \Pr[P_u] \leq \Pr[\cE_i \circ \cE_j] \leq \Pr[\cE_i] \cdot \Pr[\cE_j] = \Pr[P_u[i, u]] \cdot \Pr[P_u[u, j]].
  \]

\end{proof}
\section{Combining the Splitting Lemma with existing techniques}\label{appendix:exponential-algs}
In this section, we include a pair of technical exercises in which we explore the ability of the Splitting Lemma to
achieve approximations for \BI{} when combined with existing techniques which were developed for \DR{}.
The main takeaway is that while it is possible to achieve reach guarantees, these results have exponential dependence on $k$, thus motivating the further technical developments of~\Cref{sec:ballpacking}.

\subsection{Proof of~\Cref{thm:independent-set-thm}}\label{appendix:bilo-tecnique}

In a graph $G$, we say that a subset of vertices $X\subseteq V$ is a $\beta$-independent set if for all distinct $u,v \in X$, $\proxG{G}{u}{v} < \beta$. We now recall the Gonzalez clustering algorithm \cite{gonzalez1985clustering} and adapt it to our proximity-based setting. Given a graph $G$ and a value $k \in \mathbb{Z}^+$, the algorithm computes a $k$-clustering $\langle V_1, V_2, \ldots, V_k \rangle$ of $V(G)$ in two steps. First, it selects $k$ vertices $c_1, c_2, \ldots, c_k \in V(G)$ to serve as cluster centers: $c_1$ is chosen arbitrarily, and for each $i = 2, \ldots, k$, the vertex $c_i$ is selected to minimize the maximum proximity to the previously selected centers, i.e., $c_i \in \arg\min_{v \in V(G)} \max_{j \in [1,i-1]} \proxG{G}{v}{c_j}$. In the second step, each vertex $v \in V(G)$ is assigned to the cluster $V_i$ associated with the closest center $c_i$, i.e., $v \in V_i$ if $\proxG{G}{v}{c_i} \geq \proxG{G}{v}{c_j}$ for all $j \in [1,k]$, $j \neq i$. For any resulting cluster $V_i \subseteq V$, we define its proximity radius as: $r_G(V_i) := \max_{u\in V_{i}} \min_{v \in V_{i}} \proxG{G}{u}{v}$.
\begin{lemma}\label{lemma:lemma_independentSet}
 Let $G$ be a graph in which the size of the largest $\beta$-independent set is at most $k$. Then, the Gonzalez algorithm run on $G$ with parameter $k$ returns a clustering ${ V_1, V_2, \ldots, V_k }$ such that the proximity radius of each cluster satisfies $r_G(V_i) \geq \beta$ for all $i \in [1,k]$.
\end{lemma}

\begin{proof}
    To prove the lemma, suppose for contradiction that there exists a vertex $u \in V(G)$ such that $\max_{i \in [1,k]} \proxG{G}{u}{c_i} < \beta$. Then, by construction of the $k$-element set $\{c_1, c_2, \ldots, v_k\}$ via Gonzalez, we also have $\proxG{G}{c_i}{c_j} < \beta$ for all $i \neq j$. Thus, the set $\{c_1, c_2, \ldots, c_k\} \cup \{u\}$ forms a $\beta$-independent set, contradicting the assumption that the maximum cardinality of any $\beta$-independent set in $G$ is at most $k$. Therefore, it must be that $\min_{v \in V(G)} \max_{i \in [1,k]} \proxG{G}{v}{c_i} \geq \beta$. As a consequence, for each cluster $V_i$ centered at $c_i$, the cluster radius satisfies:
    \[
    r_G(V_i) \geq \min_{v \in V_i} \proxG{G}{v}{c_i} \geq \min_{v \in V(G)} \max_{j \in [1,k]} \proxG{G}{v}{c_j} \geq \beta.
    \]
\end{proof}

\independentsetclaim*
\begin{proof}
    Let $X = \{v_1, v_2, \ldots, v_\ell\}$ be a $\beta$-independent set in $G$, i.e., $\max_{a,b \in S} \proxG{G}{a}{b} < \beta$. 
    Assume toward a contradiction that there does not exist any $3\beta$-independent set of size $\ell - 1$ in $G+uv$.
    This implies that no $X' \subseteq X$ of size $\ell - 1$ is a $3\beta$-independent set in $G+uv$.
    It follows that one of two cases occurs. Either there are four distinct vertices in $X$, of which at most two may be in any given $3\beta$-independent set in $G+uv$, or there are three distinct vertices in $X$, no two of which have proximity less than $3\beta$.
    Below, we derive contradictions in both of these cases.

    \textbf{Case 1: Disjoint pairs become close:} Suppose there exist four distinct vertices $v_i, v_{i'}, v_j, v_{j'} \in X$ such that $\proxG{G+uv}{v_i}{v_j} \geq 3\beta$ and $\proxG{G+uv}{v_{i'}}{v_{j'}} \geq 3\beta$. Let 
    $d_1 = \log_\beta \proxG{G}{v_i}{u}$,
    $d_2 = \log_\beta \proxG{G}{v_j}{v}$,
    $d_3 = \log_\beta\proxG{G}{v_{i'}}{u}$, and
    $d_4 = \log_\beta\proxG{G}{v_{j'}}{v}$. 
    See~\Cref{fig:four-vertices-case} for a visual aid. Note that since $v_i, v_j, v_{i'}$, and $v_{j'}$ are distinct vertices in the $\beta$-independent set $X$, we have that $d_1 + d_3 > 1$ and $d_2 + d_4 > 1$.
        
    \begin{figure}[!h]
    \begin{minipage}[t]{.5\textwidth}
    \begin{tikzpicture}[scale=1.2, 
        vertex/.style={circle, draw, minimum size=8mm}, 
        dot/.style={circle, fill=black, inner sep=2pt},
        edge label/.style={midway, fill=white, inner sep=1pt},
        wavy/.style={decorate, decoration={snake, amplitude=1mm, segment length=5mm},dotted,thick}]

        \node[dot, label=above:$u$] (u) at (0,0) {};
        \node[dot, label=above:$v$] (v) at (2,0) {};
    
        \node[dot, label=above:$v_i$] (vi) at (-1.5,1) {};
        \node[dot, label=below:$v_{i'}$] (vip) at (-1.5,-1) {};
        \node[dot, label=above:$v_j$] (vj) at (3.5,1) {};
        \node[dot, label=below:$v_{j'}$] (vjp) at (3.5,-1) {};
    
        \draw[blue] (u) -- node[edge label, below = 1mm , text = black] {$e$} (v);
    
        \draw[wavy] (vi) -- node[edge label, above right] {$\beta^{d_1}$} (u);
        \draw[wavy] (v) -- node[edge label, above left] {$\beta^{d_2}$} (vj);
        \draw[wavy] (vip) -- node[edge label, below right] {$\beta^{d_3}$} (u);
        \draw[wavy] (v) -- node[edge label, below left] {$\beta^{d_4}$} (vjp);
    
    \end{tikzpicture}
    
    \end{minipage}
    \begin{minipage}[b]{.5\textwidth}
        
        \caption{An illustration of Case 1. In the graph $G$, $v_i, v_{i'}, v_j, v_{j'}$ are distinct vertices in a $\beta$-independent set. The proof reasons about proximities after the (blue) edge $e = uv$ is added to $G$.\looseness=-1}
        \label{fig:four-vertices-case}
    \end{minipage}      
    \end{figure}

    Since $\proxG{G}{v_i}{v_j} < \beta$ and $\proxG{G+uv}{v_i}{v_j} \geq 3\beta$, it must be the case that the contribution in $G+uv$ of $v_iv_j$-paths which use the edge $e = uv$ is lower bounded by $2\beta$. The edge $e$ may appear in both orientations, i.e., $u$ before $v$ or $v$ before $u$, on these paths. We partition these paths further according to which of $u$ or $v$ appears first along the path from $v_i$. Without loss of generality, we assume that the subset of these paths on which $u$ precedes $v$ has contribution at least $\beta$. We denote these paths by $P_u$.
    By the Splitting Lemma (\Cref{lemma:splitting}), $\beta \leq \Pr[P_u] \leq \Pr[P_u[v_i, u]]\cdot\Pr[P_u[u, v_j]]$. Applying the Splitting Lemma again to the set of paths $P_u[u, v_j]$, we have that \[\beta \leq \Pr[P_u[v_i, u]]\cdot \Pr[P_u[u, v]] \cdot \Pr[P_u[v, v_j]] \leq \Pr[P_u[v_i, u]]\cdot \Pr[P_u[v, v_j]].\]
    Since the edge $e$ does not appear along paths in $P_u[v_i, u]$, nor along those in $P_u[v, v_j]$, it follows that $\proxG{G}{v_i}{u}\cdot\proxG{G}{v}{v_j} \geq \beta$. Equivalently, $d_1 + d_2 \leq 1$. By a similar argument, up to possibly relabeling $v_{i'}$ and $v_{j'}$, we have that $d_3 + d_4 \leq 1$.

    Now we are ready to derive a contradiction. From $d_1 + d_3 > 1$, we have that $d_1 > 1 - d_3$. By substituting into the inequality $d_1 + d_2 \leq 1$, we obtain $d_2 \leq d_3$. Substituting again into $d_3 + d_4 \leq 1$, we obtain $d_2 + d_4 \leq 1$, but this is a contradiction, since we already have that $d_2 + d_4 > 1$. This concludes the argument.

    \textbf{Case 2: Three vertices become mutually close}: Assume there exist distinct $v_i, v_{i'}, v_{i''} \in X$ such that $\proxG{G+uv}{v_i}{v_{i'}}, \proxG{G+uv}{v_i}{v_{i''}}, \proxG{G+uv}{v_{i'}}{v_{i''}} \geq 3\beta$ in $G+uv$. 
    We define constants $d_1, d_2, d_3, d_4, d_5$, and $d_6$ in a similar fashion as the previous case; see~\Cref{fig:three-vertices-case}. Since $v_i, v_{i'}$, and $v_{i''}$ are distinct vertices in the $\beta$-independent set $X$, we have the following set of inequalities:
    \begin{enumerate}[(i)]
        \item $d_1 + d_3 > 1$,
        \item $d_1 + d_5 > 1$,
        \item $d_2 + d_4 > 1$,
        \item $d_2 + d_6 > 1$,
        \item $d_3 + d_5 > 1$, and
        \item $d_4 + d_6 > 1$.
    \end{enumerate}
    
    \begin{figure}[!h]
    \scalebox{.75}{
    \begin{minipage}[t]{.5\textwidth}
    \begin{tikzpicture}[scale=1.2, 
        vertex/.style={circle, draw, minimum size=8mm}, 
        dot/.style={circle, fill=black, inner sep=2pt},
        edge label/.style={midway, fill=white, inner sep=1pt},
        wavyedge/.style={decorate, decoration={snake, amplitude=1mm, segment length=9mm}, dotted, black,thick}
        ]
    
        \node[dot, label=left:$u$] (u) at (0,0) {};
        \node[dot, label=right:$v$] (v) at (2,0) {};
    
        \node[dot, label=above:$v_i$] (vi) at (-1,2) {};
        \node[dot, label=above:$v_{i'}$] (vip) at (3.5,2) {};
        \node[dot, label=below:$v_{i''}$] (vipp) at (1,-2) {};
    
        \draw[blue] (u) to[bend left=0] node[edge label, below = 1mm,text = black] {$e$} (v);
    
        \draw[wavyedge] (vi) to[bend left=15] node[edge label, left] {$\beta^{d_1}$} (u);
        \draw[wavyedge] (vi) to[bend left=15] node[edge label, above right = 1mm] {$\beta^{d_2}$} (v);
    
        \draw[wavyedge] (vip) to[bend left=-10] node[edge label, below right] {$\beta^{d_3}$} (u);
        \draw[wavyedge] (vip) to[bend left=10] node[edge label, below right] {$\beta^{d_4}$} (v);
    
        \draw[wavyedge] (vipp) to[bend left=10] node[edge label, below right] {$\beta^{d_6}$} (v);
        \draw[wavyedge] (vipp) to[bend right=10] node[edge label, below left] {$\beta^{d_5}$} (u);
    \end{tikzpicture}
    \end{minipage}
    } 
    \begin{minipage}[b]{.6\textwidth}
        
        \caption{An illustration of Case 2. In the graph $G$, $v_i, v_{i'}, v_{i''}$ are distinct vertices in a $\beta$-independent set. The proof reasons about proximities after the (blue) edge $e = uv$ is added to $G$.}
        \label{fig:three-vertices-case}
    \end{minipage}
    
    \end{figure}

    Since $\proxG{G}{v_i}{v_{i'}} < \beta$ and $\proxG{G+uv}{v_i}{v_{i'}} \geq 3\beta$, we have that $v_iv_{i'}$-paths using the edge $uv$ have contribution at least $2\beta$ in $G+uv$. Then, using the Splitting Lemma as in the previous case of this proof, we obtain that:
    \begin{enumerate}[(a)]
        \item Either $d_1 + d_4 \leq 1$ or $d_2 + d_3 \leq 1$, 
        \item Either $d_2 + d_5 \leq 1$ or $d_1 + d_6 \leq 1$, and
        \item Either $d_3 + d_6 \leq 1$ or $d_4 + d_5 \leq 1$.
    \end{enumerate}
    
    From inequalities (i), (ii), and (v), we have that at least two of $d_1, d_3$, and $d_5$ are $> \frac{1}{2}$. 
    Similarly, from inequalities (iii), (iv), and (vi), at least two of $d_2, d_4,$ and $d_6$ are $> \frac{1}{2}$.
    We will show that these conclusions are mutually exclusive.

    If $d_1$ and $d_3$ are both $> \frac{1}{2}$, then from (a) we have that either $d_2$ or $d_4$ is $< \frac{1}{2}$. Then there are two options. 
    If $d_2$ and $d_6$ are $> \frac{1}{2}$, then we use (a) to observe that $d_1 + d_4 \leq 1$ and (b) to observe that $d_2 + d_5 \leq 1$. Thus, if $d_1 \geq d_2$, we substitute into the former to achieve $d_2 + d_4 \leq 1$, and if $d_1 < d_2$ we substitute into the latter to achieve $d_1 + d_5 \leq 1$. These conclusions contradict (iii) and (ii), respectively.
    Thus, if $d_1$ and $d_3$ are both $> \frac{1}{2}$, then it must be the case that $d_4$ and $d_6$ are both $> \frac{1}{2}$.
    Now, we use (a) to observe that $d_2 + d_3 \leq 1$, and (c) to observe that $d_4 + d_5 \leq 1$. If $d_4 \leq d_3$, we substitute into the former to see that $d_2 + d_4 \leq 1$, and if $d_4 > d_3$ we substitute into the latter to see that $d_3 + d_5 \leq 1$. These conclusions contradict (iii) and (v), respectively.

    We have shown that $d_1$ and $d_3$ cannot both be greater than $\frac{1}{2}$. By symmetric arguments, at least one of $d_1, d_5$ and at least one of $d_3, d_5$ is $\leq \frac{1}{2}$. This is a contradiction, since we have already shown that at least two of $d_1, d_3$, and $d_5$ must be greater than $\frac{1}{2}$. Thus, the proof is complete. 
\end{proof}

\independentsettheorem*
    
\begin{proof}
Let $S^*$ be an optimal set of $k$ edges whose addition to $G$ yields a new graph $G^* = G + S^*$ with optimal reach value $\beta^*$, i.e., $\proxG{G^*}{u}{v} \geq \beta^*$ for all $u \neq v$. Now, we claim that no $(\beta^*/3^{k})$-independent set of size $k+2$ exists in $G$. 
We prove this by inducting on $k$. In the base case, for $k=0$, $S^* = \emptyset$, so $G^* = G$. Hence, we have $\proxG{G}{u}{v} \geq \beta^*$ for all $u,v \in V(G)$. Thus, no $\beta^*$-independent set of size $2$ exists in $G$. So, the base case holds. Now, as induction hypothesis, we assume that if a graph $G'$ can be augmented with at most $k-1$ edges to obtain reach at least $\beta^*$, then no $(\beta^*/3^{(k-1)})$-independent set of size $k+1$ exists in $G'$. Now, we prove the inductive step for $k$. For contradiction, we assume that $G$ contains a $(\beta^*/3^{k})$-independent set of size $k+2$. Given our budget $k$ and optimal set of augmenting edges $S^* = {e_1, e_2, ..., e_k}$ (ordered arbitrarily), we add the first edge $e_1$ and obtain the graph $G_1 = G + {e_1}$. Applying Claim \ref{claim:independent-set-claim}, the maximum size of a $3 \cdot (\beta^*/3^{k}) = \beta^*/3^{k-1}$-independent set in $G_1$ is at least $k+1$, but this contradicts the inductive hypothesis. Therefore, no $(\beta^*/3^{k})$-independent set of size $k+2$ exists in $G$.

Thus if we run the Gonzalez algorithm on $G$ with parameter $k+1$ (the desired number of cluster centers), then by Lemma~\ref{lemma:lemma_independentSet}, each cluster has proximity radius at least $\frac{\beta^*}{3^{k}}$. That is, $r_G(V_i) \geq \frac{\beta^*}{3^{k}}$ for all $i \in [k+1]$.

Now, we construct a new probabilistic graph by adding $k$ edges between $c_1$ and the remaining $k$ centers: $c_1c_2, c_1c_3, \ldots, c_1c_{k+1}$. We call this set of edges $\tilde{S}$. We now analyze the reach of the resulting graph. Consider any two arbitrary vertices $v_i \in V_i$ and $v_j \in V_j$. By construction, we have:
$\proxG{G}{v_i}{c_i} \geq \frac{\beta^*}{3^{k}}$, $\proxG{G+\{\tilde{S}\}}{c_i}{c_1} \geq \pmin$, $\proxG{G + \{\tilde{S}\}}{c_1}{c_j} \geq \pmin$, and, $\proxG{G}{c_j}{v_j} \geq \frac{\beta^*}{3^{k}}$. Thus we get that:
\[
\proxG{G + \{\tilde{S}\}}{v_i}{v_j} \geq \left(\frac{\beta^*}{3^{k}}\right) \cdot \pmin \cdot \pmin \cdot \left(\frac{\beta^*}{3^{k}}\right) = \frac{(\beta^*)^2 \pmin^2}{3^{2k}}.
\]

Hence, the resulting graph has reach at least $\frac{(\beta^*)^2\pmin^2}{3^{2k}}$, as claimed.
\end{proof}

\subsection{Necessity of the loss in~\Cref{claim:independent-set-claim}}\label{appendix:bilo-tecnique-lb}
Extending the discussion from~\Cref{sec:warmup}, we now give a simple counterexample that shows that a lossless variant of~\Cref{claim:independent-set-claim} is false.

\begin{lemma}\label{lem:independent-set-counterexample}
There exists a probabilistic graph $G$ with uniform edge probabilities $\alpha$ and a parameter $\beta$ such that (a) there exists a $\beta$-independent set in $G$ with three vertices, and (b) after adding one edge $e$ to $G$, any pair of vertices $u,v$ have $\proxG{G+e}{u}{v} \ge \beta (1+\frac{\alpha}{2})$. 
\end{lemma}

While the lemma falls short of showing a factor $3$ loss, it shows that a fully lossless version of~\Cref{claim:independent-set-claim} cannot hold.

\begin{proof}[Proof of~\Cref{lem:independent-set-counterexample}]
Consider the probabilistic graph in~\Cref{fig:tightness-of-ball-growing}, with sampling probability $\alpha_{uv} = \alpha$ for every possible edge $uv$, and suppose $\alpha$ is a constant $<1/2$. 
\begin{figure}[H]
    \scalebox{.75}{
    \begin{minipage}[t]{.4\textwidth}
        \begin{tikzpicture}[scale=1.2,
            every node/.style={inner sep=0pt},
            vertex/.style={circle,fill=black,inner sep=1.2pt},  
            bigvertex/.style={circle,fill=black,inner sep=2.2pt}, 
            dotstyle/.style={dotted, black, thick},
            solidedge/.style={black, thick},
            wavyedge/.style={decorate, decoration={snake, amplitude=0.5mm, segment length=3mm}, black, thick}
        ]
        
        \node[vertex, label={[label distance = 1mm]below left:$u$}] (u) at (0,0) {};
        
        \coordinate (ur1) at ($(u)+(45:0.5cm)$); 
        \coordinate (ur2) at ($(ur1)+(45:1.5cm)$); 
        \node[vertex,label = {[label distance = 1mm]right:$c_1$}] at (ur1) {};
        \node[bigvertex, label={[label distance = 1mm]right:$a_1$}] at (ur2) {};
        \draw[solidedge] (u) -- (ur1);
        \draw[dashed] (ur1) -- (ur2);
        \node at ($(ur1)!0.5!(ur2)+(0.4,0)$) {$\ell$};
        
        \coordinate (ul1) at ($(u)+(135:0.5cm)$); 
        \coordinate (ul2) at ($(ul1)+(135:1.5cm)$); 
        \node[vertex, label = {[label distance = 1mm]left:$c_2$}] at (ul1) {};
        \node[bigvertex, label={[label distance = 1mm]left:$a_2$}] at (ul2) {};
        \draw[solidedge] (u) -- (ul1);
        \draw[dashed] (ul1) -- (ul2);
        \node at ($(ul1)!0.5!(ul2)+(-0.4,0)$) {$\ell$};
        
        \coordinate (d1) at ($(u)+(270:1.5cm)$);
        \node[bigvertex, label={[label distance = 1mm]below:$a_3$}] at (d1) {};
        \draw[dashed] (u) -- (d1);
        \node at ($(u)!0.5!(d1)+(0.2,-0.3)$) {$\ell$};
        
        \draw[blue] (ul1) -- (ur1);
        \node at ($(ul1)!0.5!(ur1)+(0,0.2)$) {$e$};
        
        \end{tikzpicture}
    
    \end{minipage}
    } 
    \begin{minipage}[b]{.65\textwidth}
        
        \caption{Graph $G$ with a center vertex $u$ connected to three vertices $a_1, a_2, a_3$ via paths of length $\ell+1, \ell+1,\ell$ respectively. $c_1$ and $c_2$ are the first two vertices in the $\ell +1$ length paths that connect $u$ to $a_1$ and $a_2$ respectively. The new edge $e$ in blue is added between $c_1$ and $c_2$.
        \bigskip
        }
        \label{fig:tightness-of-ball-growing}
    \end{minipage}
    
\end{figure}

It is easy to see that the vertices $\{a_1, a_2, a_3\}$ form a $\beta$-independent set, for $\beta := \alpha^{2\ell+1}$. Additionally, we can see that after adding the edge $c_1 c_2$ as depicted, for every pair of vertices $u, v$, the proximity is at least $\alpha^{2\ell+1}(1+\alpha -\alpha^2) > \beta (1+\frac{\alpha}{2})$. 

\end{proof}

Now we consider variants of~\Cref{claim:independent-set-claim} which bound the structure of independent sets when batches of edges are added. Recalling the discussion from~\Cref{sec:warmup}, consider functions $f, g$, and variants of~\Cref{claim:independent-set-claim} of the following form: If there exists a $\beta$-independent set $X$ in $G$ and a set $S$ of $f(k)$ edges are added to $G$, then there exists a $g(k, \beta)$-independent set of size $|X| - f(k)$ in $G+S$.~\Cref{claim:independent-set-claim} corresponds to the functions $f(k) = 1$ and $g(k, \beta) = 3\beta$. We will illustrate why the natural extension of the argument used to prove~\Cref{claim:independent-set-claim} breaks down for $f(k) > 1$.

Let $S$ be any set of $f(k)$ edge additions. As in~\Cref{sec:ballpacking}, consider the auxiliary graph $H^X = (X, E_H)$ of $G + S$ where for $x_i, x_j \in X$, the edge $x_ix_j$ exists in $H^X$ if and only if $\proxG{G+S}{x_i}{x_j} \geq g(k, \beta)$.
To generalize the case analysis of~\Cref{claim:independent-set-claim}, we consider the structure of this auxiliary graph. In particular, the first case in the proof guarantees the existence of a matching of size $f(k) + 1$ in $H^X$. The second case in the proof guarantees a clique of size $f(k) + 2$. 
Consider the first case, and let $x_ix_j$, $y_iy_j$ be two edges of the matching in $H^X$. The first step in the argument is to see that the contribution in $G+S$ of $x_ix_j$-paths which use at least one edge from $S$ must be at least $g(k, \beta) - \beta$. The same can be said of $y_iy_j$-paths using at least one edge from $\beta$.
We want to use the Splitting Lemma (\Cref{lemma:splitting}) as in~\Cref{claim:independent-set-claim}, but crucially we can only apply this lemma at a single vertex at a time. The natural solution, seen several times in this work, is to partition the $x_ix_j$-paths according to the first and last endpoints of edges in $S$ encountered. If $g(k, \beta)$ is large enough, then we may assert that some set of paths in this partition has contribution at least $\beta$. Let $u, v$ be the first and last endpoints of edges in $S$ defining this set of paths. Then, applying the Splitting Lemma, we may conclude that $\proxG{G+S}{x_i}{u} \cdot \proxG{G+S}{v}{x_j} \geq \beta$.
By a similar argument, for appropriately selected $u', v'$, we have that $\proxG{G+S}{y_i}{u'} \cdot \proxG{G+S}{v'}{y_j} \geq \beta$. The key difficulty is that to derive the desired contradiction, i.e., that either $\proxG{G}{x_i}{y_i} \geq \beta$ or $\proxG{G}{x_j}{y_j} \geq \beta$, we need $u = u'$ and $v = v'$. 

In other words, implicitly in the proof of~\Cref{claim:independent-set-claim} we are using the Pigeonhole principle to assert that in the equivalence relation described above, the class with the largest contribution has the \emph{same} $u, v$ for at least two edges $x_ix_j$ and $y_iy_j$ of the matching in $H^X$. When $f(k) = 1$, this is obvious since there is only one equivalence class. Unfortunately, in general there are $\Omega(f^2)$ equivalence classes, but there are only $f(k) + 2 \in O(f)$ edges in the matching. A similar combinatorial explosion prevents us from generalizing the second case in the proof. 

\subsection{A second exponential algorithm}\label{appendix:li-technique}
In this section we include another exponential approximation algorithm for \BI{}. This time, our technique is similar to the \DR{} analysis given by~\cite{li1992minimum}.
The main idea is to strengthen~\Cref{cor:2kballs} by showing an appropriate relationship between $\reachOPT$ and the optimum objective value for \MplusoneKCENTER{} in the implied metric of $G$.
Though our approach is similar to~\cite{li1992minimum}, there are several challenges to overcome since our distance metric does not provide some of the guarantees of the shortest path metric which are crucial in their approach. In particular, while a given edge $uv$ may participate in shortest paths from a vertex $x$ in at most one orientation, i.e., $u \rightarrow v$ or $v \rightarrow u$, in our setting an edge may contribute to proximity in both orientations. We show how to overcome this, but with an exponential dependence on $k$.

We now describe our algorithm.
Let $E^*$ be the optimal set of edges to add to $G$ for \BI{}, $G^*=(V,E \cup E^*,\alpha)$, $\broadOPT$ be the reach of $G^*$, and $D^*=-\log \broadOPT{}$. We use the algorithms of Gonzalez~\cite{gonzalez1985clustering} or Hochbaum and Shmoys~\cite{hochbaum1985best} to obtain a 2-approximation for the \MKCENTER{} instance on the implied metric $(V, \metricfunc)$ of $G$ with parameter $k+1$. If we let $M_{k+1}$ be the optimal radius of the \MKCENTER{} instance and $\mu_{k+1}$ be the value returned by the algorithm then $\mu_{k+1} \leq 2M_{k+1}$. Let $X \subseteq V$ be the set of~$k+1$ vertices returned by the 2-approximation algorithm. Using the vertices of $X$ add a star in $G$ (with the center chosen arbitrarily from $X$) and let $G'$ be this new graph. Thus, we added at most $k$ new edges to $G$ and we will show that the reach of $G'$ (call it $\beta'$ and let $D'=-\log \beta'$) is at least $\frac{\broadOPT{}^4\pmin^2}{16^{k+1}}$.

For every vertex $s$ in $G'$ (and also in $G$), at least one of the centers from $X$ is within distance $\mu_{k+1}$ in the implied metric of $G$, and therefore also in the implied metric of $G'$. Thus, in the implied metric of $G'$, starting from any vertex $s$, by reaching the nearest center in distance at most $\mu_{k+1}$, then taking one edge($\alpha_e \ge \pmin$) to reach the root of the star we added, taking one more edge to the center closest to the destination vertex $t$, and then traveling a distance of at most $\mu_{k+1}$ again to reach $t$ -- in at most $2\mu_{k+1} - 2 \log \alpha$ distance, one can reach any vertex in the graph starting from any vertex. Thus, 

\begin{inequality}\label{ineqD}
    D^* \leq D' \leq 2\mu_{k+1} - 2\log \pmin \leq 2(2M_{k+1}) - 2\log\pmin \leq 4M_{k+1} - 2\log\pmin
\end{inequality}

We will now show that $M_{k+1}$ is upper-bounded by $D^* + k + 1$. Combined with~\cref{ineqD} this allows us to bound the approximation factor. 

To show that $M_{k+1} \leq D^* + k + 1$, we will show that $G$ is a feasible instance of \MKCENTER{} for parameter value $k+1$ and distance to center at most $D^* + k + 1$. We pick the centers in the following way: we first pick an arbitrary vertex $x$ as a center. Consider the set of at most $2k$ vertices that are the end points of $E^*$. For every new edge $(i,j)$ in $G^*$, if $\proxG{G^*\setminus (i,j)}{x}{i} \geq \proxG{G^*\setminus (i,j)}{x}{j}$ then we pick $j$ as a center. Otherwise, we pick $i$ as a center. In other words, we pick as a center whichever of $i$ and $j$ is farther from $x$ in the implied metric of $G^*\setminus (i,j)$. In this way, we get a set of at most $k+1$ centers. Call this set $\hat{X}$.
We now prove that $\hat{X}$ is a feasible solution.

\begin{restatable}{lemma}{singlecriteriaclosecenter}\label{clm:single-criteria-close-center}
    For every vertex $w \in V$ at least one of the vertices in $\hat{X}$ is within distance $D^*+k+1$ of $w$ in the implied metric of $G$.
\end{restatable}
\begin{proof}
    If $w$ is a center, the claim is trivially true. So, assume $w$ is not a center. Let $x$ be an arbitrary vertex, and let $P_{G^*}(w,x)$ be the set of all simple paths from $w$ to $x$ in $G^*$. Call the edges in $E^*$ \emph{new} edges. Consider the partition of $P_{G^*}(w,x)$ into $2k+1$ equivalence classes based on the first new (directed) edge they encounter going from $w$ to $x$. Let $C_{vu}$ represent the class of simple $(w,x)$ paths such that $(u,v)$ is the first new edge they encounter and $v$ appears before $u$ on these paths. Define $C_{uv}$ similarly. Suppose in the new edge $(u,v)$, $v$ is the center. Then, we will say that $C_{vu}$ and $C_{uv}$ are the \emph{center-first} and the \emph{center-second} classes, respectively. Moreover, since $v$ is a center, $\proxG{G^*\setminus(u,v)}{x}{v} \leq \proxG{G^*\setminus(u,v)}{x}{u}$. 

    If the contribution of the center-first class $C_{vu}$ is at least $\frac{\broadOPT{}}{2^{k+1}}$ then $\proxG{G}{w}{v} \geq \frac{\broadOPT{}}{2^{k+1}}$ and the claim is true. So, assume $\Pr[C_{vu}] \leq \frac{\broadOPT{}}{2^{k+1}}$. Then,
    \[
        \Pr[P_{G^*}(w,x) \setminus (C_{uv} \cup C_{vu})] + \Pr[C_{uv}] \geq \broadOPT{} - \frac{\broadOPT}{2^{k+1}}
    \]

    We now claim that the first summand accounts for most of this probability mass. 
    
    \begin{claim}
        $\Pr[P_{G^*}(w,x) \setminus (C_{uv} \cup C_{vu})] \geq \frac{1}{2}(\broadOPT{} - \frac{\broadOPT}{2^{k+1}})$.
    \end{claim}
    \begin{proof}[Proof of claim.]
        Suppose otherwise. Then the Splitting Lemma~(\Cref{lemma:splitting}) gives us that 
        \[ \Pr[C_{uv}[w, u]]\cdot\Pr[C_{uv}[u, x]] = \Pr[C_{uv}[w, u]]\cdot\alpha_{(u,v)}\cdot\Pr[C_{uv}[v, x]] \geq \frac{1}{2}(\broadOPT{} - \frac{\broadOPT}{2^{k+1}}). \]
        By noting that $\proxG{G^*\setminus (u, v)}{w}{u} \geq \Pr[C_{uv}[w, u]]$ and $\proxG{G^*\setminus (u, v)}{v}{x} \geq \Pr[C_{uv}[v, x]]$, we observe that
        \[
            \proxG{G^*\setminus (u, v)}{w}{u} \cdot \proxG{G^*\setminus (u, v)}{v}{x} \geq \frac{1}{2}(\broadOPT{} - \frac{\broadOPT}{2^{k+1}}).
        \]
        We can now manipulate definitions to obtain a contradiction.
        \begin{align*}
        \begin{aligned}
        &\Pr[P_{G^*}(w,x) \setminus (C_{uv} \cup C_{vu})] \\
        &\geq \Pr[P_{G^* \setminus (u,v)}(w,x)] \\
        &=\proxG{G^* \setminus (u,v)}{w}{x} \\
        &\geq \proxG{G^* \setminus (u,v)}{w}{u} \cdot\proxG{G^* \setminus (u,v)}{u}{x} \\
        &\geq \proxG{G^* \setminus (u,v)}{w}{u} \cdot\proxG{G^* \setminus (u,v)}{v}{x} \\
        &\geq (1/2)(\broadOPT{} - \broadOPT/2^{k+1}) \\
        \end{aligned}
        \begin{aligned}
            &\\
            \quad&(\text{\hfill since }P_{G^* \setminus (u,v)}(w,x) \subseteq P_{G^*}(w,x) \setminus (C_{uv} \cup C_{vu})) \\
            &(\text{by definition of proximity}) \\
            &(\text{by triangle inequality of the implied metric}) \\
            &(\text{since } \proxG{G^* \setminus (u,v)}{v}{x} \leq \proxG{G^* \setminus (u,v)}{u}{x}) \\
            &(\text{by our argument above}), \\
        \end{aligned}
    \end{align*}
    but we began this proof by supposing that $\Pr[P_{G^*}(w,x) \setminus (C_{uv} \cup C_{vu})] < \frac{1}{2}(\broadOPT{} - \frac{\broadOPT}{2^{k+1}})$.

    \end{proof}

    We repeat the same analysis by looking at the now remaining $2k-1$ classes. If $(u',v')$ is another new edge and the center-first class of $(u',v')$ has contribution at most $\frac{\broadOPT}{2^{k+1}}$ then the remaining $2k-3$ classes must have contribution at least 
    \[ \frac{1}{2}(\frac{\broadOPT{}}{2}(1-\frac{1}{2^{k}})-\frac{\broadOPT{}}{2^{k+1}})=\frac{\broadOPT{}}{4}(1-\frac{1}{2^{k}}-\frac{1}{2^{k}}) =\frac{\broadOPT{}}{4}(1-\frac{1}{2^{k-1}}). \]
    Thus, if after $t$ rounds the contribution of the remaining $2(k-t)+1$ classes is $\frac{\broadOPT{}}{2^t}(1-\frac{1}{2^{k-t+1}})$ then after 1 more round the contribution of the remaining $2(k-t)-1$ classes will be at least $\frac{\broadOPT{}}{2^{t+1}}(1-\frac{1}{2^{k-(t+1)+1}})$.

    Therefore, if the center-first classes of none of the $k$ new edges have contribution at least $\frac{\broadOPT{}}{2^{k+1}}$ in $G^*$ then the ``empty class'' (corresponding to paths that have no new edge) must have contribution $\geq \frac{\broadOPT{}}{2^k}(1-\frac{1}{2}) \geq \frac{\broadOPT{}}{2^{k+1}}$ in $G^*$ (and hence $G$). This completes the proof. 
\end{proof}

Putting it all together, we get the main result.

\begin{restatable}{theorem}{singlecriteriaapprox}\label{clm:single-criteria-approx}
    There exists a polynomial-time algorithm for~\BI{} which produces reach (for $\sourceV = V$) at least $\frac{(\broadOPT)^4\pmin^2}{16^{k+1}}$ by adding at most $k$ edges.
\end{restatable}

\section{Omitted details from the proof of~\Cref{thm:ballpacking-theorem}}\label{appendix:ballpacking}

\bucketing*
\begin{proof}
Let $u \in U$. Consider any $v$ to $u$ path. Since $G$ is connected, such a path always exists. Let $x$ be the first vertex on this path such that $\proxG{G}{x}{U} \geq r$ ($x$ can be $u$). Let $w$ be the vertex just before $x$ on this path ($w$ can be $v$). Then, by our choices of $x$ and $w$, $\proxG{G}{w}{U} < r$. Moreover, by the triangle inequality of the implied metric of $G$, $\proxG{G}{w}{U} \geq \pmin\proxG{G}{x}{U} \geq \pmin r$. Thus, $\proxG{G}{w}{U} \in [r\pmin, r)$. Setting $v'=w$, the claim is true.
\end{proof}

\ballpackingtheorem*
\begin{proof}[Full Proof]
Let $S$ be the set of optimal edges. We assume as input to our algorithm
a guess $\beta'$ for the value of $\broadOPT$ such that $\beta' \geq \frac{\broadOPT}{1 + \varepsilon}$. We set $\proxradius$ and $\proxdiameter$ accordingly. Note that for a guess $\beta' \leq \broadOPT$,~\Cref{thm:strengthened-ball-growing-theorem} is true for $\beta'$ (if the theorem happens to hold for some guess $\beta' > \broadOPT$, this will only improve the result of our algorithm). 

By~\Cref{claim:ballpacking-terminates-claim}, at the end of the {\bf while} loop in Line~\ref{step:pickw} of the algorithm, $|C| \le 2k+1$. In the auxiliary graph $H_r^C$, for all $c_i, c_j \in C, i \neq j$, we place an edge if $\proxG{G}{c_i}{c_j} \geq \proxdiameter\pmin$. By property~(\ref{claim:ballpacking-closest-center-claim}) of~\Cref{claim:ballpacking-tightness}, for every $c_i \in C$ at least one $c_j$ satisfies this condition. In other words, the degree of every vertex in $H_r^C$  is at least one. Thus, the smaller color class $D$ picked in Line~\ref{step:picksmallercolorclass} has at most $k$ vertices. Thus, $D$ has at most $k$ vertices and the star centered at $c$ has at most $k-1$ edges. By then using~\Cref{claim:ballpacking-compactness-claim}, property~(\ref{claim:ballpacking-closest-center-claim}) of~\Cref{claim:ballpacking-tightness}, and the triangle inequality in the implied metric of $G$, we conclude that for all $v \in V$, $\proxG{G+\hat{S}}{v}{D} \ge \proxdiameter^2\pmin$. Therefore, the resultant reach is at least
\[ \proxdiameter^4\pmin^4 = \frac{(\broadprime)^4\pmin^{4}}{4^4k^8} \geq \frac{(\broadOPT)^4 \pmin^{4}}{4^4k^8(1 + \varepsilon)^4}. \]

Finding the set $C$ and building the auxiliary graph can be done in polynomial time. A spanning forest of $H_r^C$ can be found in polynomial time using one of several standard algorithms, for eg., Kruskal's algorithm~\cite{kleinberg2005algorithm}. Furthermore, a forest can be 2-colored in polynomial time~\cite{cormen2009introduction} and the set $\hat{S}$ can also be found in polynomial time. Thus, the algorithm runs in polynomial time. 
 
It remains to show how we can estimate $\broadOPT$. We will do this by mimicking the technique of Demaine and Zadimoghaddam~\cite{demaine2010minimizing}. In the following, let $t$ denote the number of vertices added to $C$ by the algorithm. We note that $\broadOPT \leq 1$, so $\frac{\broadOPT}{\broad{G}} \leq \frac{1}{\broad{G}}$. Then for any $\varepsilon > 0$, there exists some integer $0 \leq i \leq \log_{1+\varepsilon} \frac{1}{\broad{G}}$ with the property that $\broad{G}(1 + \varepsilon)^i \leq \broadOPT \leq \broad{G}(1+\varepsilon)^{i+1}$. We conduct a binary search of integers in the interval $[0, \log_{1+\varepsilon}\frac{1}{\broad{G}}]$. Note that $\broad{G} \geq \pmin^n$, so this interval has polynomial length (for fixed $\varepsilon$).

For each tested integer $j$, we assume that $\broadOPT = \broad{G}(1 + \varepsilon)^j$, and execute the algorithm described above. By~\cref{claim:ballpacking-terminates-claim}, we can conclude that if the algorithm adds more than $2k+1$ vertices to $C$, then $\broad{G}(1 + \varepsilon)^j \geq \broadOPT$, and therefore that $j > i$. Let $j^*$ be the largest integer in the interval for which our algorithm adds at most $2k+1$ vertices to $C$.
Then we can conclude that $\broadOPT \leq \broad{G}(1+\varepsilon)^{j^*+1}$, and in this case our algorithm adds at most $t \leq 2k+1$ vertices to $C$ to produce reach at least 
\[
    \frac{({\broad{G}(1 + \varepsilon)^{j^*}})^4 \pmin^{4}}{4^4k^8} = 
    \frac{({\broad{G}(1 + \varepsilon)^{j^*+1}})^4 \pmin^{4}}{(1 + \varepsilon)^4 4^4k^8}
    \geq \frac{{\broadOPT}^4 \pmin^{4}}{(1 + \varepsilon)^4 4^4k^8}
\]
as desired.
\end{proof}

\section{Proof of the Star Lemma}\label{appendix:star-lemma}

\starstructuretheorem*
\begin{proof}
  Choose an arbitrary vertex $u \in V(S)$.
  Let $S_{star} = \{ij \ | \ v \in V(S) \setminus \{u\}\}$. Let $(i, j) \in \sourceV \times V$. Let $P_{ij}$ be the set of all paths from $u$ to $v$ in $G+S$, and note that by definition
  $\Pr[P_{ij}] = \proxG{G+S}{i}{j} \geq \reach{G+S}{\sourceV} = \beta'$.
  Let $p$ and $p'$ be two paths in $P_{ij}$. Let $l$ be the leading vertex of the first \emph{new edge} (an edge in $S$) appearing along $p$, and let $t$ be the trailing vertex of the last new edge appearing along $p$.
  Define $l'$ and $t'$ similarly. We impose an equivalence relation on $P_{ij}$ by declaring that $p$ and $p'$ are similar if $l = l'$ and $t = t'$.
  We reserve one equivalence class for the set of paths containing no new edges (the \emph{empty class}). Note that there are $|V(S)| \leq 2k$ possible values of $l$, and for each such value there are $|V(S)| - 1$ possible values for $t$.
  Thus, including the empty class there are at most $2k(2k - 1) + 1 < 4k^2$ equivalence classes. Then by the union bound, some equivalence class has contribution at least $\beta'/4k^2$.
  If the empty class meets this criteria, then we are done, since these paths also exist in $G+S_{star}$. Moreover, this class also satisfies the second conclusion of the lemma.
  Otherwise, choose one such equivalence class, defined by vertices $l$ and $t$, and call this class $\mathcal C$.

  We now further partition $\mathcal C$ into three subsets. The first, $\mathcal C_1$, is the set of paths in $\mathcal C$ on which $u$ precedes $l$. The second, $\mathcal C_2$,
  is those paths on which $t$ precedes $u$. $\mathcal C_3$ is $\mathcal C \setminus (\mathcal C_1 \cup \mathcal C_2)$. Note that if $u \in \{l, t\}$, then this partition is still well-defined, with $\mathcal C_3 = \mathcal C$ and $\mathcal C_1 \cup \mathcal C_2 = \emptyset$.
  Again using the union bound, the sum of the contributions of these three sets is an upper bound for the contribution of $\mathcal C$.
  It follows that at least one has contribution at least $\beta'/12k^2$.
  Let $\mathcal C'$ denote whichever of $\mathcal C_1, \mathcal C_2$, or $\mathcal C_3$ has the largest contribution.
  We will show how to replace $\mathcal C'$ with a new set of paths $Q$ which uses only edges in $E \cup S_{star}$.

  Note that $lu, ut \in S_{star}$. If $C' = \mathcal C_1$, then we form $Q$ by, for each $p \in \mathcal C'$, replacing $p[u, t]$ with the edge $ut$.
  If $\mathcal C' = \mathcal C_2$, then we replace $p[l, u]$ with the edge $lu$. Otherwise $\mathcal C' = \mathcal C_3$, in which case we replace $p[l, t]$ with the segment $l, u, t$.
  Observe that if $\mathcal C' = \mathcal C_1$, then $\mathcal C'[i, u] = Q[i, u]$ and $\mathcal C'[t, j] = Q[t, j]$.
  Similarly, if $\mathcal C' = \mathcal C_2$, then $\mathcal C'[i, l] = Q[i, l]$ and $\mathcal C'[u, j] = Q[u, j]$.
  Finally, if $\mathcal C' = \mathcal C_3$, then $\mathcal C'[i, l] = Q[i, l]$ and $\mathcal C'[t, j] = Q[t, j]$.
  In any case, we call the segment of $\mathcal C'$ on which $Q$ differs the \emph{middle segment} of $\mathcal C'$, denoted $\mathcal C_m'$, and we call the other two segments the \emph{beginning} and \emph{ending} segments, written $\mathcal C_b'$ and $\mathcal C_e'$, respectively.
  We define $Q_b, Q_m$, and $Q_e$ similarly.
  
  Now we claim that $Q$ has contribution at least $\frac{\beta'\pmin^2}{12k^2} = \beta_{star}$ in $G + S_{star}$.  
  Let $(p_1, p_2)$ be a pair of paths with $p_1 \in \mathcal C_b' = Q_b$ and $p_2 \in \mathcal C_e' = Q_e$.
  We call $(p_1, p_2)$ a \emph{nice path pair} if $p_1$ and $p_2$ are vertex-disjoint, and we say that $(p_1, p_2)$ \emph{exists} in a sampled graph if both paths exist. 
  Let $\cE_1$ be the event that a nice path pair exists in a sampled graph\footnote{
        Technically, $\cE_1$ is an event in two sample spaces, i.e., the spaces defined by sampling from $G + S$ and
        $G + S_{star}$. However, since edges are sampled independently and the set of possible nice path pairs is identical in both graphs,
        the event remains well-defined and has equal probability under both measures.
    }.
  Note that by construction, the vertex $u$ does not appear on any path in either $Q_b$ or $Q_e$. Then the edges
  of the paths in $Q_m$ are disjoint from the edges of paths in $Q_b$ and $Q_e$. Noting that edges are sampled independently,
  we now have that $\Pr[Q] = \Pr[Q_m]\cdot\Pr[\cE_1]$. Moreover, because $Q_m$ consists of a single path
  on at most two edges, i.e., either the edge $ut$, the edge $uw$, or the path $l, u, t$, we may write $\Pr[Q_m] \geq \pmin^2$,
  and conclude that $\Pr[Q] \geq \pmin^2\Pr[\cE_1]$. Next, we note that the existence in a sampled graph of a path in $\mathcal C'$ implies the existence of a nice path pair, and thus $\Pr[\cE_1] \geq \Pr[\mathcal C'] \geq \beta'/12k^2$.
  Then $\Pr[Q] \geq \beta_{star}$, as desired.
  Furthermore, the set $Q$ also satisfies the second conclusion of the lemma.
\end{proof}

\section{Omitted Proofs from~\Cref{sec:star-submodular}}\label{appendix:submodularity}

\begin{observation}\label{obs:broadcast-not-submodular}
    Neither reach nor the logarithm of reach is submodular with respect to edge additions.
  \end{observation}
  To understand~\Cref{obs:broadcast-not-submodular}, it is helpful to consider a small example: a sub-divided star with three leaves. For simplicity, consider the case in which $\sourceV = V$ and $\pmin = \pmax = \alpha$.
  Then our graph $G$ consists of a ``center'' vertex $v_0$ connected to each of three leaves $v_1, v_2$, and $v_3$ via
  disjoint paths on $\ell$ edges, where $\ell$ can be thought of as some large integer whose value depends on $\alpha$.
  The reach of $G$ is $\alpha^{2\ell}$.
  Now, we observe that $G + \{v_1v_3\}$ is isomorphic to $G + \{v_2v_3\}$, and both have reach less than $2\alpha^{2\ell}$.
  Meanwhile, in the graph $G + \{v_1v_3, v_2v_3\}$, every pair of vertices lies on some cycle of length at most $2\ell + 2$. Thus, the reach of
  this graph is at least $\alpha^{\ell + 1}$. Given a sufficiently large value of $\ell$, we now have
  \[
  \broad{G + \{v_1v_3\}} - \broad{G} < \alpha^{2\ell} < \alpha^{\ell + 1} - 2\alpha^{2\ell} < \broad{G + \{v_1v_3, v_2v_3\}} - \broad{G + \{v_2v_3\}},
  \]
  violating the definition of submodularity. A similar analysis on the same graph can be used to show that the logarithm of
  reach is not submodular. In this case, one need only show that $\frac{\broad{G + \{v_1v_3\}}}{\broad{G}} < \frac{\broad{G + \{v_1v_3, v_2v_3\}}}{\broad{G + \{v_2v_3\}}}$ for an appropriately selected value of $\ell$.

  \begin{figure}
    \centering
    \vspace{-13pt}
    \begin{tikzpicture}
    \draw (-.7, 0) circle (1cm);
    \draw (.7, 0) circle (1cm);
    \draw (0, -1) circle (1cm);
    \node[] at (-1.5, 1) {$\cE_0$};
    \node[] at (1.5, 1) {$\cE_1$};
    \node[] at (-1, -1.8) {$\cE_2$};
    \node[] at (1, 0.2) {$x_b$};
    \node[] at (.6, -.6) {$x_c$};
    \node[] at (0, -1.4) {$x_d$};
\end{tikzpicture}
    \caption{Events and probabilities used in the proof of~\Cref{lem:submodularity}.}\label{fig:venn-diagram}
  \end{figure}
  \submodularitylemma*
  \begin{proof}
    Since proximity only increases by adding edges (this fact can readily be seen by the sampling based definition of proximity), it is clear that $g_v$ is monotone. So let us focus on submodularity.
    
    We first study the two-edge setting: consider any $e_1, e_2 \in E^u$. We will claim that
    \begin{equation}\label{eq:submodular-main}
    \frac{\proxG{G+e_1+e_2}{v}{u}}{\proxG{G+e_2}{v}{u}} \le   \frac{\proxG{G+e_1}{v}{u}}{\proxG{G}{v}{u}}.
    \end{equation}
    
    The inequality is obvious if either $e_1$ or $e_2$ already exists in $G$, and so let us assume that this is not the case.
    
    Let us define $P_1$ (resp., $P_2$) to be the set of simple paths in $G+e_1$ (resp., $G+e_2$) that go from $v$ to $u$, \emph{ending in} the edge $e_1$ (resp., $e_2$). Let $P_0$ be the set of simple paths in $G$ that go from $v$ to $u$, (so they do not contain either of $e_1$ or $e_2$). The key observation is that \emph{any simple path} from $v$ to $u$ in $G+e_1+e_2$ must be in $P_0 \cup P_1 \cup P_2$. This is because no simple path contains both $e_1$ and $e_2$; it also cannot have $e_1$ or $e_2$ as an intermediate edge in the path.
    Now, define $\cE_0, \cE_1$, and $\cE_2$ to be the events that at least one path from $P_0$, $P_1$, or $P_2$ (respectively) exists in a sampled graph.
    Thus, the inequality~(\ref{eq:submodular-main}) is equivalent to:
    \[  \frac{ \Pr[ \cE_0 \cup \cE_1 \cup \cE_2 ]}{ \Pr[ \cE_0 \cup \cE_2 ]} \le \frac{ \Pr[ \cE_0 \cup \cE_1]}{ \Pr[ \cE_0 ]}. \]
    
    To simplify the following algebra, we introduce some variables; see~\Cref{fig:venn-diagram}. Specifically, we say that
    $x_a = \Pr[\cE_0]$, $x_b = \Pr[\bar{\cE_0} \cap \cE_1 \cap \bar{\cE_2}]$, $x_c = \Pr[ \bar{\cE_0} \cap \cE_1 \cap \cE_2 ]$, and $x_d = \Pr[ \bar{\cE_0} \cap \bar{\cE_1} \cap \cE_2]$.
    Using this notation, it is easy to check that
    \[
      \frac{ \Pr[ \cE_0 \cup \cE_1 \cup \cE_2 ]}{ \Pr[ \cE_0 \cup \cE_2 ]} = \frac{x_a + x_b + x_c + x_d}{x_a + x_c + x_d} \leq \frac{x_a + x_b + x_c}{x_a + x_c} \leq \frac{x_a + x_b + x_c}{x_a} = \frac{ \Pr[ \cE_0 \cup \cE_1]}{ \Pr[ \cE_0 ]},
    \]
    establishing the claim.
    
    The claim implies submodularity in a straightforward way: suppose $S \subseteq T \subseteq E^u$ and consider any $e \in E^u$. Again, the case $e \in T$ is trivial, so let us assume that $e \not\in T$. We can now consider $G+S$ to be the ``base'' graph, and use the argument above repeatedly, choosing $e_1 = e$, and an element of $T \setminus S$ as $e_2$. This gives us
    \[ \frac{\proxG{G+T+e}{v}{u}}{\proxG{G+T}{v}{u}} \le   \frac{\proxG{G+S+e}{v}{u}}{\proxG{G+S}{v}{u}}. \]
    Taking logarithms, we obtain that $g_v$ is submodular.
  \end{proof}

  \potentialdeclinelemma*
  \begin{proof}
    Let $e_1, e_2, \dots, e_{2k}$ be the edges in $E'$. Define the current set of active pairs as
    \[ A\su{t-1} = \{ (i,j) \in \sourceV \times V : \mu(i,j; S\su{t-1}) > 0 \}. \]
    
    For any such pair, the first observation is to note that 
    \begin{equation}\label{eq:submodular-avg1} \sum_{\ell \in [2k]} \left( \mu(i,j; S\su{t-1}) - \mu(i,j; S\su{t-1} \cup \{e_\ell \} ) \right) \ge \mu(i, j; S\su{t-1}).
    \end{equation}
    This follows via a standard argument, adding the elements $e_\ell$ in some order and using submodularity. For a technical reason, we note that this also implies that
    \begin{equation}\label{eq:submodular-avg2}
    \sum_{\ell \in [2k]} \left( \mu(i,j; S\su{t-1}) - \max\{ 0, \mu(i,j; S\su{t-1} \cup \{e_\ell \} )\} \right) \ge \mu(i, j; S\su{t-1}).
    \end{equation}
    This follows from~\eqref{eq:submodular-avg1}, because if for some $\ell$, $\mu (i,j; S\su{t-1}\cup\{e_\ell \}) < 0$, that term in the summation alone is $\ge$ RHS, and we only need to use the fact that every other term is non-negative. If all the $\mu$ are $\ge 0$, then~\eqref{eq:submodular-avg1} and~\eqref{eq:submodular-avg2} are identical.
    
    We can sum this over all pairs $(i,j)$ in $A\su{t-1}$, and noting that the RHS is exactly $\Psi(S\su{t-1})$, we have by averaging,
    \[ \exists \ell ~:~ \sum_{(i,j) \in A\su{t-1}} \left( \mu(i,j; S\su{t-1}) - \max\{ 0, \mu(i,j; S\su{t-1} \cup \{e_\ell \} ) \} \right) \ge \frac{1}{2k} \Psi (S\su{t-1}). \]
    
    This implies that $\Psi(S\su{t-1}) - \Psi (S\su{t}) \ge \frac{1}{2k} \Psi (S\su{t-1})$. The lemma follows by rearranging the terms.
  \end{proof}

  \submodalgtheorem*
  \begin{proof}
    We begin by assuming that we have a value $\reach{G^*}{\sourceV} \geq x \geq \reach{G^*}{\sourceV}/(1+\varepsilon)$. Having $x$, we set $\beta' = \frac{x\pmin^2}{12k^2}$, so~\cref{thm:star-structure} tells us that it is possible to achieve reach $\beta'$ by adding $O(k)$ edges incident to a single vertex. We now proceed with Algorithm~\ref{alg:submod}, trying each possible vertex $u$. Because $G$ is connected, we have $\proxG{G}{u}{i} \ge \pmin^n$ for every $u,i$. Thus, $\mu(i, j; S\su{0}) \le \frac{2n}{\varepsilon} \log (1/\pmin^\varepsilon)$ for all $i,j$, implying that the initial potential is at most $(2n^3/\varepsilon) \cdot \log(1/\pmin^\varepsilon)$. Since the potential drops by a factor at least $(1- \frac{1}{2k})$ in each iteration, and since the algorithm terminates when the potential reaches $\log(1/\pmin^{\varepsilon})$, we conclude that the number of iterations is $O(k \log n)$.
    Furthermore, when the algorithm terminates, we have $\mu(i, j; S\su{t}, \beta') \le \log \frac{1}{\pmin^\varepsilon}$ for all pairs $(i,j) \in \sourceV \times V$. This implies that for all $(i, j) \in \sourceV \times V$, 
    $$\proxG{G+S\su{t}}{i}{j} \ge \beta' \cdot \pmin^\varepsilon \geq \frac{\reach{G^*}{\sourceV}\pmin^{2+\varepsilon}}{(1+\varepsilon)12k^2}.$$
    
    To obtain the estimate $x$ for $\reach{G^*}{\sourceV}$, we use a technique similar to that of~\Cref{thm:ballpacking-theorem}. Specifically, we note that a standard
    analysis can be used to obtain a precise bound $b \in O(k\log n)$ on the number of edges added by our algorithm, given that $\reach{G^*}{\sourceV} \geq x$ and that we have guessed the correct vertex $u$. We
    conduct a binary search of integers in the interval $[0, \log_{1+\varepsilon} \frac{1}{\reach{G}{\sourceV}}]$. For each such integer $i$, we execute our algorithm with $x = \reach{G}{\sourceV}(1+\varepsilon)^i$. If the algorithm adds more than
    $b$ edges for every guess of the vertex $u$, then we conclude that $x > \reach{G^*}{\sourceV}$. Let $i^*$ be the largest guessed value such that the algorithm terminates (for some guess of vertex $u$) after adding at most $b$ edges.
    Then we have that
    $x = \reach{G}{\sourceV}(1+\varepsilon)^{i^*} \geq \reach{G^*}{\sourceV}/(1 + \varepsilon)$ and consequently the achieved reach matches or exceeds our desired bound. This completes the proof.
  \end{proof}
  
\section{Full Proof of~\Cref{thm:broadcasthardness}}\label{appendix:hardness-bi}

\broadcasthardness*

The reduction is from a variant of the \textsc{Set Cover} problem. Specifically, we rely on the following hardness assumption ~\cite{feige1998,feige2010}:

\hardgapsetcover*

Note that the assumption implies that doing even slightly better than the bicriteria guarantee of the greedy algorithm of~\Cref{sec:star-submodular} is \cclass{NP}-hard. For our reduction, we only need a bound of $(1 - \Omega(1))m$ in the No case, which is weaker. Likewise, we only require $|S_i| = o(\sqrt{n})$. We also remark that  Assumption~\ref{asmp:hardgsc} likely holds even with $c = (\log n)^{1-o(1)}$. With this stronger assumption, our hardness results can be improved to nearly match our algorithmic bounds. We omit the details. 

\begin{proof} 
Our reduction from \textsc{Gap Set Cover} is as follows.

\vspace{6pt}

\noindent \underline{\textbf{Instance:}}  Given an instance of 
\gsc{} consisting of a collection of $n$ sets $S_1,S_2, \dots,  S_n \subseteq [m]$. We construct a \BI{} instance of $(G=(V,E,\{\alpha_e\}),k)$ with $\pmin = \pmax = \alpha $ as follows:

We create a graph $G$ with a \emph{pivot vertex} $p$, vertices $s_i$ corresponding to sets $S_i$ (called \emph{set vertices}) and vertices $e_i$ corresponding to elements $i \in [m]$ (called \emph{element vertices}). Between every pair of set vertices $s_i, s_j$, we add a path of length $l$, where $l$ is an even integer parameter whose value will be specified later. These paths are mutually disjoint, and so there are $\binom{n}{2} (l-1)$ vertices along the paths. We call these \emph{set-set internal vertices}. Next, we add a path of length $l$ between $s_i$ and $e_j$ for all $j \in S_i$. (I.e., we connect a set vertex $s_i$ to all the element vertices $e_j$ corresponding to elements $j\in S_i$.). Once again, these paths are all mutually disjoint. We call the vertices on the paths the \emph{set-element internal vertices}. Finally, we connect the pivot to each set vertex via mutually disjoint paths of length $l$. We call the internal vertices along these paths \emph{pivot-set internal vertices.}

Now we argue about the maximum reach that can be achieved after adding $k$ edges to such a graph $G$.

\noindent \underline{\textbf{Yes-case:}} let there be a set cover of size $k$ that covers all elements. Suppose we consider adding $k$ edges between the pivot vertex $p$ and set vertices corresponding to sets in the set cover. We claim that between any two vertices in the resulting graph, there is a path of length at most $2.5 l + 1$, thus implying that the reach $\broad{G} \ge \alpha^{2.5l + 1}$. 

To see the claim, we argue separately for each vertex:

First, from the pivot vertex $p$, every vertex can be reached via a path of length $\le 2l$. To see this, note that the distance of each $e_i$ (and therefore also each set-element internal vertex) from $p$ is at most $l+1$; this follows because we have direct edges from $p$ to a set cover. Also, every set-vertex (and therefore also every pivot-set internal vertex) can be reached via a path of length at most $l$. Thus, all of the set-set internal vertices can be reached from $p$ via a path of length $2l$ (indeed, this can be made $1.5l$ by choosing the closer set-vertex).

Second, from any set vertex, we can reach every other set vertex using a path of length $l$, and thus every element vertex with a path of length $\le 2l$. Further, any of the internal vertices can also be reached via a path of length $\le 2l$, as can the pivot.

Third, from any element vertex $e_i$, we can reach every other element vertex $e_j$ with a path of length $\le 2l+2$ (by going to the set vertex covering $i$, going to the pivot, then to the set vertex covering $j$, then going to $e_j$). Further, any set vertex can be reached via a path of length at most $2l$ (going to a set vertex covering $i$ and taking the length $l$ path to the desired set vertex). Any set-set internal vertex can thus be reached via a path of length at most $2.5 l$: we can go from the target vertex to the closest set vertex ---with a path of length $\le 0.5l$--- and from there to $e_i$ by a path of length $\le 2l$ as before. Moreover, any pivot-set internal vertex can be reached via a path of length at most $2l$, using a path of length at most $l + 1$ to the pivot and then proceeding to the target vertex via path of length at most $l - 1$. Finally, any set-element internal vertex can be reached via a path of length $(2.5l + 1)$: from the target, we can either go to an element vertex via a path of length $(0.5l-1)$ or a set vertex via a path of length $(0.5l+1)$, and using the above, this implies that we can get to the target by a path of length $2.5l+1$.

Fourth, from any pivot-set internal vertex $y$, we can reach any other pivot-set internal vertex via a path of length $<2l$ by first traveling to the pivot, and then to the target vertex. Moreover, we can reach any set-element internal vertex $x$ via a path of length at most $2.5l$. We begin by traveling to either the pivot or the set vertex corresponding to $y$ (whichever is closer) via a path of length at most $0.5l$. We then continue to $x$ via at most $2l$ additional edges, as argued above.

Next, from any set-set internal vertex, we can reach any set vertex with a path of length $\le 1.5l$, and thus we can reach every other vertex with a path of length $\le 2.5l$.

Finally, from a set-element internal vertex, it only remains to show that we can reach any other set-element internal vertex using a short path (other cases are covered above by symmetry). Consider two set-element internal vertices $x$ and $x'$, with $x$ (resp., $x'$) on the path from $s_i$ to $e_j$ (resp., $s_{i'}$ to $e_{j'}$). We show that there is a \emph{cycle} with $5l+2$ edges in the graph that contains $x, x'$. Note that this implies that there is a path of length $\le 2.5l + 1$ (i.e., the shorter path on the cycle). As shown in~\Cref{fig:hardness-loop}, note that there is a path of length $2l+2$ between $e_j$ and $e_{j'}$ (using the covering set vertices, as seen above), and there is also the path going to $s_i$, to $s_{i'}$, then to $e_{j'}$. 
\begin{figure}
    \centering
    \includegraphics[width= 0.5\textwidth]{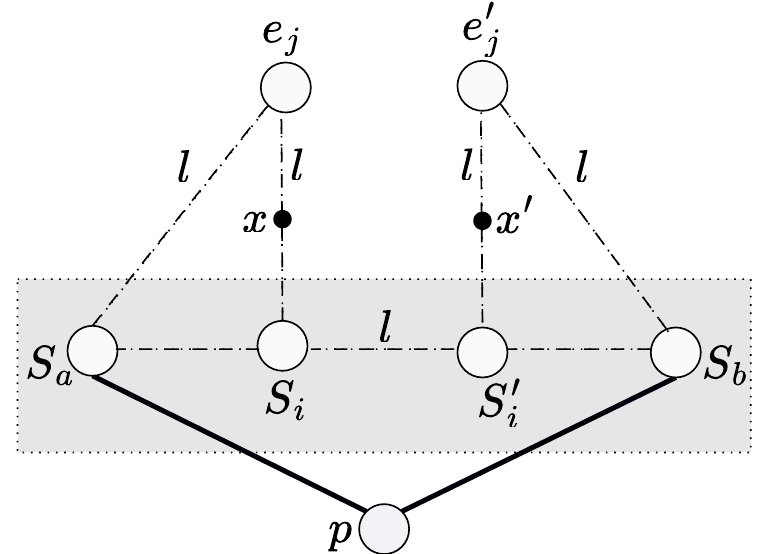}
    \caption{This figure shows the cycle of length $5l+2$ containing two set-element internal vertices $x$ and $x'$. This cycle shows that $x$ and $x'$ have proximity at least $\alpha^{2.5l + 1}$ (in the yes-case).}
    \label{fig:hardness-loop}
\end{figure}

This completes the proof of the claim. Thus, in the yes-case, $\broad{G} \ge \alpha^{2.5l+1}$. 

\vspace{6pt} 

\noindent \underline{\textbf{No-case:}}
By \Cref{asmp:hardgsc}, the union of any $ck$ sets among $S_1, \dots, S_n$  covers at most $(1 - \Omega(1))m$ elements, for any constant $c \ge 1$. 

Now consider adding $r = c'k$ edges, where $r \le ck/8$. Let $E'$ denote the set of added edges and let $G'$ be the graph obtained from $G$ by adding the edges in $E'$.

For any $\{u,v\} \in E'$, we define a set of ``\emph{involved}'' set and element vertices as follows. A set vertex $s_i$ is said to be involved in edge $\{u,v\}$ if for $x \in \{u,v\}$, we have (a) $s_i = x$, (b) $x$ is a set-set internal vertex and one of the end-points of the corresponding path (the one containing $x$) is $s_i$, or (c) $x$ is a set-element internal vertex and the set end-point of the corresponding path is $s_i$. Analogously, we say that an element vertex $e_i$ is involved in edge $\{u,v\}$ if for $x \in \{u,v\}$, either (a) $e_i = x$ or (b) $x$ is a set-element internal vertex and $e_i$ is the element end-point of the corresponding path.

We also generalize the notation slightly and say that a vertex is involved in a set of edges $E'$ if it is involved in at least one of the edges $e \in E'$. The main claim due to our choice of parameters is the following.

\begin{restatable}{claim}{nocaseclaim}
    There exist element vertices $e_i, e_j$ such that (a) neither of them is involved in $E'$, (b) none of the sets containing $i,j$ are involved in $E'$, and (c) there is no set $S_t$ that contains both $i$ and $j$.
\end{restatable}
\begin{proof}[Proof of Claim.]
First, note that any edge $e \in E'$ can have at most $4$ set vertices and $2$ element vertices involved in it. Thus, if we pick $r$ edges, we will have at most $4r$ set vertices and $2r$ element vertices involved. For each of the involved element vertices, choose an arbitrary set that covers that element, thus obtaining a set $I$ of $6r$ set vertices with the property that all the set and element vertices involved in $S'$ are either in $I$ or are covered by $I$. Now by our assumption for the no-case and the choice of $r$, this means there are $\Omega(m)$ element vertices that are not covered by the sets in $I$; Call this set $U$. By definition, for any $e_i \in U$, none of the sets covering $e_i$ are in $I$ (and so property (b) automatically holds). Finally, note that there must exist $e_i, e_j \in U$ such that no set vertex has length $l$ paths to both of them. This can be seen by a simple averaging argument as follows. Let $c_{ij}$ be the number of sets that contain both $i$ and $j$. Then, we have 
\[ \sum_{i,j \in U} c_{ij} \le \sum_{\ell \in [n]} \binom{|S_{\ell}|}{2} \le n \cdot \text{polylog}(n).  \] 
Indeed, if we take all pairs $i,j \in [m]$, the first inequality becomes equality. Now, if $c_{ij} \ge 1$ for all $i,j \in U$, then $\sum\limits_{i,j \in U} c_{ij} \ge \binom{|U|}{2}$, and since $|U| = \Omega(m)$ and $m = \Theta(n)$, this leads to a contradiction. Thus, there must exist $i,j \in U$ such that $c_{ij} = 0$, and this completes the proof of part (c) of the claim.
\end{proof}

For $e_i, e_j$ satisfying the claim, we show that $\proxG{G'}{e_i}{e_j}$ has to be small. First, note that the distance between them in $G'$ is $\ge 3l$. This is because any path from $e_i$ to $e_j$ must go via one of the set vertices containing $i,j$, and since they are not involved in $E'$, the shortest path between those set vertices has length $l$. Our goal now is to argue that the reach is also ``close'' to $\alpha^{3l}$. For this, we show how to simplify the graph for easier reasoning about reach.

\begin{claim}\label{lem:contraction}[Subset contraction]
Let $S$ be any subset of the vertices of $G'$. Define a \emph{contraction} as the process where we replace $S$ with a single ``hub'' vertex $h$, and replace every edge of the form $\{u,v\}$ where $u \in S$ and $v \not\in S$ with $\{h,v\}$ (forming parallel edges if appropriate). Let $G^c$ be the graph obtained after contraction. For any $u, v \not\in S$, we have $\proxG{G^c}{u}{v} \ge \proxG{G'}{u}{v}$.
\end{claim}

The claim then follows immediately from the sampling-based definition of proximity: suppose we sample edges with probability $\alpha$ each, then if a path exists in $G'$, it also exists in $G^c$ (because of us placing parallel paths).  Now given $E'$, define $S'$ as the set of all set and element vertices involved in $E'$, along with the pivot $p$. Then, define $S$ to be the union of $S'$ and all the internal vertices along paths between vertices of $S'$.  The crucial observation now is that every edge in $E'$ has both its end-points in $S$.

Now suppose we contract the set $S$ in $G'$ and obtain the graph $G^c$. By the claim, it suffices to show an upper bound on $\proxG{G^c}{e_i}{e_j}$ (where $e_i, e_j$ are the element vertices that we identified earlier). To do this, we make another observation about $G^c$: its vertices consist of $h$ (the new hub vertex), a subset $V_{set}$ of the original set vertices, a subset $V_{elt}$ of the element vertices, all the internal vertices of the paths between $V_{set} \cup V_{elt}$, and all the internal vertices of the paths between $h$ and $V_{set} \cup V_{elt}$. Thus it is natural to define a ``path compressed'' graph $H$, whose vertex set is $\{h\} \cup V_{set} \cup V_{elt}$, which has an edge iff there is a path of length $l$ in $G^c$. Note that there can be parallel edges in $H$. Now suppose we view $H$ as an probabilistic graph, where the sampling probability is $\alpha^l$ for every edge. Then we have the following easy observation:

\begin{observation}\label{obs:equality}
For all $u, v \in \{h\} \cup V_{set} \cup V_{elt}$, we have $\proxG{G^c}{u}{v} = \proxG{H}{u}{v}$.
\end{observation}

For our $e_i, e_j$ of interest, bounding $\proxG{H}{e_i}{e_j}$ will turn out to be simple, because the edge probability $\alpha^l$ will be chosen to be so small that only the shortest path between $e_i$ and $e_j$ matters. We formalize this in the following simple claim.

\begin{claim}\label{lem:puv-bound}
Let $H = (V_H, E_H)$ be an probabilistic graph on $n_H$ vertices in which every edge has sampling probability $\le \delta < \frac{1}{2 n_H}$. Let $u, v \in V$ such that dist$(u,v) = c$, for some integer $c \ge 1$. Then $\proxG{H}{u}{v} \le 2 n_H^{c-1} \delta^c $.
\end{claim}
\begin{proof}[Proof of Claim.]
For any path of length $\ell$, the probability that the path exists in a sampled graph is $\delta^{\ell}$. Between any two vertices, there are clearly at most $n_H^{\ell-1}$ paths of length $\ell$, and thus by a union bound, noting that dist$(u,v)= c$ and therefore there are no paths of length $<c$, we have:
\begin{align*}
\proxG{H}{u}{v} &\le \sum_{\ell = c}^{n_H - 1} n_H^{\ell-1} \delta^\ell \le n_H^{c-1} \delta^c \cdot \sum_{\ell \ge 0} (\delta n_{H})^\ell \le 2 n_H^{c-1} \delta^c.
\end{align*}
\end{proof}

The claim now implies a bound on $\proxG{H}{e_i}{e_j}$; we cannot use it directly since our $H$ has parallel edges, but we note that there are at most $6k + 1$ parallel edges between any two vertices (because that is a bound on the number of ``non-internal'' vertices, plus the pivot, in the contracted set $S$). Thus, we can replace parallel edges by a single edge with $\delta \le (6k+1) \alpha^l$, and use~\Cref{obs:equality} to obtain:
\[ \proxG{H}{e_i}{e_j} = \proxG{G'}{e_i}{e_j} \le\proxG{G^c}{e_i}{e_j} \le 2 (m+n)^2 \delta^3. \]
This is because the shortest path in $H$ between $e_i$ and $e_j$ has length three, as there is no set that contains both $i$ and $j$. Now, if we choose $l$ large enough (approximately $\frac{\log n}{\epsilon}$), we can make $2 (m+n)^2 \delta^3$ to be $< \alpha^{3l - \epsilon}$, for any $\epsilon > 0$.

Since the reach in the yes-case is $\ge \alpha^{2.5l + 1}$, the desired gap follows.

\end{proof}

\section{Proof of~\Cref{thm:sshardness}}\label{appendix:ss-hardness}

\sshardness*
\begin{proof}

    We once again reduce from \problem{Gap Set Cover}. Let $(S_1, S_2, \ldots S_n \subseteq [m], k)$ be an instance of
    this problem.
    Recall that it is \cclass{NP}-hard to distinguish between the cases
    $\emph{(i)}$ there exists a collection of at most $k$ sets which cover every element, and \emph{(ii)} any collection of
    $2ck$ sets leaves at least $2ck+1$ elements uncovered\footnote{
        We know from~\cite{feige1998,feige2010} that it is \cclass{NP}-hard to distinguish between the existence
        of a set cover with size $k$ and the non-existence of any set cover with size $2ck$. Our stronger hardness assumption
        can be obtained via a simple reduction: copy the instance of \problem{Gap Set Cover} and create $2ck$ additional
        replicas of each element, giving each replica the same set memberships as the original. 
    }~\cite{feige1998,feige2010}. 
    Let $d$ be the maximum size of any set, i.e., $d = \max_{i \in [n]} |S_i|$. Let
    $b$ be the maximum number of sets containing any single element, i.e., $b = \max_{j \in [m]} |\{i \in [n] \colon j \in S_i\}|$. Finally, we will
    need an additional value $l$, which can be thought of as an integer which is polynomial in $m + n$. We will show how to
    select the exact value of $l$ at the end of the proof.
    
    We construct
    an instance of \RI{} with uniform edge-sampling probability $\alpha$ as follows. First we introduce a vertex $\sourcev$, which will be the source vertex of 
    our constructed instance. Next, for each set $S_i$, $i \in [n]$, we introduce a vertex $v_i$. We call these vertices
    \emph{set vertices}. For each set vertex $v_i$, we introduce $l-1$ auxiliary vertices and $l$ edges, such that $v_i$ is connected to
    $\sourcev$ by a path of edge-length $l$ on these vertices and edges. We call these auxiliary vertices the \emph{set-path vertices} corresponding to set $S_i$.
    So far, we have added (less than) $1 + ln$ vertices and exactly $ln$ edges. Next, for each element $j \in [m]$, we add a new vertex $v_j$, which we call
    an \emph{element vertex}. For each set $S_i$ containing element $j$, we add $l - 1$ auxiliary vertices and $l$ edges
    such that $v_i$ and $v_j$ are connected by a path of edge-length $l$ on these vertices and edges. We call these auxiliary vertices
    the \emph{set-element-path} vertices corresponding to set $S_i$ and element $j$. This step adds at most $m\cdot d\cdot l$ vertices and edges.
    We call the probabilistic graph we have constructed $G$. We set the limit on edge additions to $k$, and the target value for the reach
    of $\sourcev$ to $\alpha^{\frac{3l}{2} + 1}$. This completes the construction of our instance of \RI{}.
    
    It remains to show that the reduction is correct. We begin by assuming that $(\{S_i\}_{i \in [n]}, k)$ is a yes-instance.
    That is, we assume that there exist $k$ sets $S_{i_1}, S_{i_2}, \ldots, S_{i_k}$ which cover every element. In this case,
    we add the $k$ edges between $\sourcev$ and the set vertices corresponding to this set cover. That is, we propose the solution
    $S = \{\sourcev v_{i_1}, \sourcev v_{i_2}, \ldots, \sourcev v_{i_k}\}$. Now we show that $\reach{G+S}{\sourcev} \geq \alpha^{\frac{3l}{2} + 1}$.
    Let $u \neq \sourcev$ be some vertex in $G$. If $u$ is a set vertex or a set-path vertex, then our initial construction guarantees that
    $\proxG{G+S}{u}{\sourcev} \geq \proxG{G}{u}{\sourcev} \geq \alpha^l$. If $u$ is an element vertex, then we identify a set $S_i$ which contains the element
    corresponding to $u$ and is part of the cover. Then 
    \[
        \proxG{G+S}{u}{\sourcev} \geq \proxG{G+S}{u}{S_i}\cdot\proxG{G+S}{S_i}{\sourcev} \geq \alpha^l\cdot \alpha = \alpha^{l+1}
    \]

    Finally, assume that $u$ is a set-element-path vertex corresponding to set $S_i$ and element $j$. Observe that
    either $\proxG{G}{u}{v_i} \geq \alpha^{\frac{l}{2}}$ or $\proxG{G}{u}{v_j} \geq \alpha^{\frac{l}{2}}$. In the former case,
    \[
        \proxG{G+S}{u}{\sourcev} \geq \alpha^{\frac{l}{2}}\cdot\proxG{G+S}{v_i}{\sourcev} \geq \alpha^{\frac{3l}{2}},
    \]
    and in the latter case
    \[
        \proxG{G+S}{u}{\sourcev} \geq \alpha^{\frac{l}{2}}\cdot\proxG{G+S}{v_j}{\sourcev} \geq \alpha^{\frac{3l}{2}+1}.
    \]

    Thus, all vertices have proximity at least $\alpha^{\frac{3l}{2}+1}$ to $\sourcev$ in $G+S$, so $S$ witnesses that
    $(G, \sourcev, k, \alpha^{\frac{3l}{2}+1})$ is a yes-instance of \RI{}.

    We now assume that $(\{S_i\}_{i \in [n]}, k)$ is a no-instance of \problem{Gap Set Cover}. 
    In this case, we let $S^*$ be an optimal solution to our constructed instance of \RI{}, and additionally allow that
    $S^*$ may contain up to $ck$ edges.
    That is, $S^*$ is a set of at most $ck$
    edge additions, with $\reach{G + S^*}{\sourcev} \geq \reach{G^*}{\sourcev}$. We will first give an upper bound on $\reach{G^*}{\sourcev}$, and then
    show how we could have chosen $l$ such that this upper bound yields the desired hardness result.
    We impose an arbitrary order on the (at most $2ck$) endpoints of the edges in
    $S^*$, $s_1, s_2, \ldots s_{ck}$. We then introduce a new solution $S$ of size at most $2ck$ such that
    $S = \{\sourcev s_1, \sourcev s_2, \ldots, \sourcev s_{2ck}\}$. By~\Cref{lemma:ss-star-lemma},
    $\reach{G+S}{\sourcev} \geq \frac{\reach{G^*}{\sourcev}}{2ck+2}$. We call the vertices $s_1, s_2, \ldots s_{2ck}$ the \emph{destinations} of the solution $S$,
    and we say that a set $S_i$ is \emph{involved} in solution $S$ if the destinations of $S$ include the set vertex
    $v_i$, any set-path vertex corresponding to $S_i$, or any set-element-path vertex corresponding to $S_i$.
    Note that every set-path vertex and every set-element-path vertex corresponds to exactly one set, so since $S$ has at most
    $2ck$ destinations we can conclude that at most $2ck$ sets are involved in $S$. Next, we say that an element $j$ is \emph{uncovered}
    if it is not contained in any set which is involved in $S$. Because only $2ck$ sets are involved in $S$, at least $2ck+1$ elements are
    uncovered. Moreover, because $S$ has at most $2ck$ destinations, there is at least one uncovered element $j$ for which the corresponding
    element vertex $v_j$ is not itself a destination of $S$. We will now show that $\proxG{G+S}{v_j}{\sourcev} \leq \alpha^{2l}(b+bd)$.

    Every path from $v_j$ to $\sourcev$ begins with $l$ edges from $v_j$ to some set vertex $v_i$, where $j \in S_i$ and $S_i$
    is not involved in $S$. From there, paths extend either via $l$ more edges to $\sourcev$, or via $l$ more edges to another
    element vertex. Each path of the former variety has contribution $\alpha^{2l}$, and $j$ is contained in at most $b$ sets, so these
    paths have contribution at most $b\alpha^{2l}$. Similarly, there are at most $bd$ paths of the latter variety, and each has contribution
    $\alpha^{2l}$. Hence, 
    
    \[
        \reach{G^*}{\sourcev} \leq \reach{G+S}{\sourcev}(2k+2) \leq \proxG{G+S}{v_j}{\sourcev}(2k+2) \leq \alpha^{2l}(b + bd)(2ck+2)
    \]

    We now claim that to achieve the desired hardness bound we need only set 
    \[
        l > \frac{2}{3\varepsilon}\big[-\log_{\alpha}(b+bd) -\log_{\alpha}(2ck+2) + \frac{4}{3} - \varepsilon \big]
    \]

    In this case, simple manipulations reveal that
    \begin{align*}
        \frac{3\varepsilon l}{2} &> -\log_{\alpha}(b+bd) -\log_{\alpha}(2ck+2) + \frac{4}{3} - \varepsilon \\
        \alpha^{\frac{3\varepsilon l}{2}}(b+bd)(2ck+2) &< \alpha^{\frac{4}{3}}\alpha^{-\varepsilon} \\
        \alpha^{2l}\alpha^{\frac{3\varepsilon l}{2}}(b+bd)(2ck+2) &< \alpha^{2l}\alpha^{\frac{4}{3}}\alpha^{-\varepsilon} \\
        \alpha^{2l}(b+bd)(2ck+2) &< \alpha^{2l}\alpha^{\frac{-3\varepsilon l}{2}}\alpha^{\frac{4}{3}}\alpha^{-\varepsilon} \\
        \reach{G^*}{\sourcev} &< (\alpha^{\frac{3l}{2} + 1})^{\frac{4}{3} - \varepsilon}
    \end{align*}

    Consequently, any algorithm which produces reach at least $(\reach{G^*}{\sourcev})^{\frac{4}{3} - \varepsilon}$ using at most $ck$ edges
    can also distinguish between yes- and no-instances of \problem{Gap Set Cover}. This completes the proof.
\end{proof}

\section{Omitted Proofs from~\Cref{sec:subset-source}}\label{appendix:subset-source}

\ssballgrowingtheorem*
\begin{proof}
    Let $v$ be an arbitrary vertex. Consider the set of all paths from $\sourcev$ to $v$ in $G'$.
    Partition these paths according to the last new edge (edge in $S$) encountered. Including the empty class, i.e., the class of paths which use no edges from $S$, there are at most $k + 1$ equivalence classes.
    Hence, some class has contribution at least $\beta'/(k+1)=2\proxradius$. If the empty class is one such class, then we observe that these paths also exist (and also have contribution at least $\beta'/(k+1)$) in $G$. Thus, in this case $\proxG{G}{\sourcev}{v} \geq 2\proxradius$ and the claim is true.
    Otherwise, let $e = s_is_j$ be the last new edge encountered by paths in the class with the largest contribution. Partition this class according to the orientation of $e$, i.e., $s_i \rightarrow s_j$ or $s_j \rightarrow s_i$.
    One set has contribution at least $\beta'/(2k+2) = \proxradius$. Without loss of generality, assume this is the set of paths for which $s_i$ is the trailing vertex of $e$, and call this set of paths $P$.
    Since $\Pr[P] \geq \proxradius$ and every path in $P$ passes through $s_i$, applying the Splitting Lemma (\Cref{lemma:splitting}) at $s_i$ we get that $\Pr[P[v, s_i]] \cdot \Pr[P[s_i, v_s]] \geq \Pr[P] \geq \proxradius$ which implies that $\Pr[P[v, s_i]] \geq \proxradius$.
    Thus, in this case, $\proxG{G}{v}{V(S)} \geq \proxG{G}{v}{s_i} \geq \proxradius$, as desired.
\end{proof}

\ssballpackingtheorem*
\begin{proof}
We assume as input to our algorithm a guess $\beta'$ for the value of $\reach{G^*}{\sourcev}$. Similar to the estimation of the optimal reach used in the proof of~\Cref{thm:ballpacking-theorem} (see \Cref{appendix:ballpacking}), we can obtain an arbitrarily good estimate for $\reach{G^*}{\sourcev}$ via a binary search. Here, we proceed as if we know $\reach{G^*}{\sourcev} \geq \beta' \geq \reach{G^*}{\sourcev}/(1+\varepsilon)$, and we set $\proxdiameter$ and $\proxradius$ accordingly, noting that the former inequality guarantees that~\Cref{thm:ss-ball-growing-cor} holds for $\beta'$.

We call Algorithm~\ref{alg:ballpacking} with $v=\sourcev$ and $\sourcev$ as the center of the star in step~\ref{step:picksmallercolorclass} as described in~\Cref{sec:subset-source}. By~\Cref{claim:ss-ballpacking-terminates-claim}, at the end of the while loop in step~\ref{step:pickw}, $C$ has at most $2k+1$ vertices. In the auxiliary graph $H_r^C$, for all $c_i, c_j \in C, i \neq j$, we put an edge between $f(c_i)$ and $f(c_j)$ if $\proxG{G}{c_i}{c_j} \geq \proxdiameter\pmin$. By property~(\ref{claim:ballpacking-closest-center-claim}) for every $c_i \in C$ at least one $c_j$ exists which satisfies this condition. In other words, the degree of every vertex in $H_r^C$  is at least one. Thus, the smaller color class $D$ picked in step~\ref{step:picksmallercolorclass} has at most $k$ vertices. Since we add $\sourcev$ to $D$, $D$ has at most $k+1$ vertices and the star centered at $\sourcev$ has at most $k$ edges. By~\Cref{claim:ballpacking-compactness-claim}, property~(\ref{claim:ballpacking-closest-center-claim}) and the triangle inequality in the implied metric of $G$, every vertex in $G+\hat{S}$ has proximity at least $\proxdiameter^2\pmin$ to some vertex in $D$. Thus, the resultant reach of $\sourcev$ is at least
\[\proxdiameter^2\pmin^2 = \frac{(\beta')^4\pmin^2}{(2k+2)^4} \geq \frac{\reach{G^*}{\sourcev}^4\pmin^2}{(2k+2)^4(1+\varepsilon)^4}\].
\end{proof}

\section{Constant Witnesses}\label{appendix:constant-witnesses}

In this section we introduce $(c, b)$-\emph{witnessing solutions}, and use them
to obtain linear approximations for \SRI{} using poly$(k)\cdot \log n$
edges. Our idea is to reduce our problem, a graph modification problem, to~\HS{}.
Observe that given an instance of \SRI{}, our task is to select $k$ \emph{modifications}, i.e., edge additions, from a set of polynomial size, i.e., $V^2 \setminus E$.
Given an optimal objective value $\reach{G^*}{\sourceV}$, we evaluate a candidate set $E'$ of modifications by checking a polynomial number of \emph{constraints}, i.e., we require that for each $(u, v) \in \sourceV \times V$, $\proxG{G+E'}{u}{v} \geq \reach{G^*}{\sourceV}$.
In the optimization context, our objective is a maximin over these constraints. Moreover, each of these constraints can be verified in polynomial time.
Many minimax and maximin graph modification problems can be phrased in this manner.

The key idea of our framework is to argue that there always exists a solution $S$ which is nearly optimal, i.e., $\reach{G+S}{\sourceV} \approx \reach{G^*}{\sourceV}$, with the additional property that each individual constraint, i.e., each vertex pair $(u, v) \in \sourceV \times V$, can be satisfied using only a \emph{constant} number of modifications in $S$.
In the context of~\SRI{}, the relevant definitions are as follows.

\begin{definition}\label{def:broadcast-constant-witness}
    Let $G = (V, E)$ be a probabilistic graph, $\sourceV \subseteq V$ a source set, $(u, v) \in \sourceV \times V$, $c \in \mathbb{N}$, and $0 \leq b \leq 1$.
    A $(u, v, b)$-\emph{witness} of size $c$ is a set $W_{uv} \subseteq V^2$ of size $c$ with the property that
    $\proxG{G+W_{uv}}{u}{v} \geq b$. If $S \subseteq V^2 \setminus E$ contains as a subset
    a $(u, v, b)$-witness of size at most $c$ for every vertex pair $u, v \in \sourceV \times V$, we say $S$ is
    a $(c, b)$-\emph{witnessing solution} to \SRI{} on $G$.    
\end{definition}

\Cref{def:broadcast-constant-witness} can be adapted in a natural way to any graph modification problem meeting the criteria described above.
To see why the idea is algorithmically useful, suppose that our problem of interest always admits a $(c, b)$-witnessing solution of size $k^d$, where $k$ is the budget for modifications. 
Then we construct an instance of~\HS{} as follows. For each possible combination of $c$ modifications (in our case, for each combination of $c$ edge additions) we create an element.
Because the number of possible modifications is polynomial and $c$ is a constant, we have a polynomial number of elements.
Next, for each constraint (in our case, each vertex pair $(u, v) \in \sourceV \times V$) we create a set. Again, because there are only a polynomial number of constraints, the size of our constructed instance is also polynomial in the input size.
Finally, for each constraint $(u, v)$, we add to the associated set the elements corresponding to all $(u, v, b)$-witnesses of size at most $c$.

By construction, the existence of a $(c, b)$-witnessing solution of size $k^d$ guarantees the existence of a hitting set of size $k^{cd}$. Each element in this hitting set corresponds to at most a constant $c$ number of modifications, so by running a $O(\log n)$-approximation for~\HS{}~\cite{johnson1974approximation}, we achieve objective value $b$ with $O(k^{cd}\log n)$ modifications.
\bigskip

In the rest of this section, we show how to apply the framework outlined above to~\SRI{}. The first step is to prove the existence of $(c, \reach{G^*}{\sourceV}\cdot\poly(k, \pmin, \pmax))$-witnessing solutions of size $\poly(k)$.
Actually, we have already done this. The Star Lemma (\Cref{thm:star-structure}) guarantees the existence of a $(2, \frac{\reach{G^*}{\sourceV}\pmin^2}{12k^2})$-witnessing solution of size $2k - 1$.
However, in this section we do not need the added edges to form a star, so we can achieve a slightly improved bound; see part $(ii)$ of the following lemma.
We include part $(i)$ because in the important special case of uniform activation probabilities, i.e., when $\pmin = \pmax$, this bound has no dependence on the activation parameter.

\begin{restatable}{lemma}{bcwexistence}\label{lemma:broadcast-constant-witness-existence}
    For any instance $(G = (V, E), k)$ of \SRI{} with optimum reach $\reach{G^*}{\sourceV}$, there exist both
    (i) a $(3, \frac{\reach{G^*}{\sourceV}\pmin^2}{3k^4\pmax^2})$-witnessing solution of size at most $7k - 6$, and
    (ii) a $(1, \frac{\reach{G^*}{\sourceV}\pmin}{12k^2})$-witnessing solution of size at most $\binom{2k}{2}$.
\end{restatable}
\begin{figure}[h]
    \centering
    \begin{tikzpicture}

    \node[circle,fill,inner sep=2pt,label=below:$u$] (u1) at (-6,0) {};
    \node[circle,fill,inner sep=2pt,label=below:$l$] (l1) at (-5,0) {};
    \node[circle,fill,inner sep=2pt,label=below:$t$] (t1) at (-2,0) {};
    \node[circle,fill,inner sep=2pt,label=below:$v$] (v1) at (-1,0) {};
    \node[circle,fill,inner sep=2pt,label=above:$w$] (w1) at (-3.5,1) {};
    \node[circle,fill=white, draw=white,inner sep=2pt] (t11) at (-3,0.3) {};
    \node[circle,fill=white, draw=white,inner sep=2pt] (l11) at (-4,0.3) {};
    \node[circle,fill=white, draw=white,inner sep=2pt] (t12) at (-3,0) {};
    \node[circle,fill=white, draw=white,inner sep=2pt] (l12) at (-4,0) {};
    \node[circle,fill=white, draw=white,inner sep=2pt] (t13) at (-3,-0.3) {};
    \node[circle,fill=white, draw=white,inner sep=2pt] (l13) at (-4,-0.3) {};
     \draw[dotted] (-6,0) sin (-5.8,.1) cos (-5.6,0) sin (-5.4,-.15) cos (-5.2,0) sin (-5,.1);
    \draw[dotted] (-2,0) sin (-1.8,.1) cos (-1.6,0) sin (-1.4,-.1) cos (-1.2,0) sin (-1,.1);
    \draw[dotted] (l12) -- (t12);
    \draw[dotted] (l11) -- (t11);
    \draw[dotted] (l13) -- (t13);
    \draw[blue] (l1) -- (w1);
    \draw[blue] (t1) -- (w1);
    \draw[red] (l1) -- (l11);
    \draw[red] (l1) -- (l12);
    \draw[red] (l1) -- (l13);
    \draw[red] (t1) -- (t11);
    \draw[red] (t1) -- (t12);
    \draw[red] (t1) -- (t13);

    \node[circle,fill,inner sep=2pt,label=below:$u$] (u2) at (1.3,0) {};
    \node[circle,fill,inner sep=2pt,label=below:$l$] (l2) at (3.7,0) {};
    \node[circle,fill,inner sep=2pt,label=below:$t$] (t2) at (6.1,0) {};
    \node[circle,fill,inner sep=2pt,label=below:$v$] (v2) at (7.1,0) {};
    \node[circle,fill,inner sep=2pt,label=below:$w$] (w2) at (2.5,0) {};
    \node[circle,fill=white, draw=white,inner sep=2pt] (t21) at (5.4,0.3) {};
    \node[circle,fill=white, draw=white,inner sep=2pt] (l21) at (4.4,0.3) {};
    \node[circle,fill=white, draw=white,inner sep=2pt] (t22) at (5.4,0) {};
    \node[circle,fill=white, draw=white,inner sep=2pt] (l22) at (4.4,0) {};
    \node[circle,fill=white, draw=white,inner sep=2pt] (t23) at (5.4,-0.3) {};
    \node[circle,fill=white, draw=white,inner sep=2pt] (l23) at (4.4,-0.3) {};
    \draw[dotted] (1.3,0) sin (1.5,.1) cos (1.7,0) sin (1.9,-.1) cos (2.1,0) sin (2.3,.1)cos (2.5,0);
    \draw[dotted] (2.5,0) sin (2.7,.1) cos (2.9,0) sin (3.1,-.1) cos (3.3,0) sin (3.5,.1)cos (3.7,0);
    \draw[dotted] (6.1,0) sin (6.3,.1) cos (6.5,0) sin (6.7,-.1) cos (6.9,0) sin (7.1,.1);
    \draw[dotted] (l22) -- (t22);
    \draw[dotted] (l21) -- (t21);
    \draw[dotted] (l23) -- (t23);
    \draw[red] (l2) -- (l21);
    \draw[red] (l2) -- (l22);
    \draw[red] (l2) -- (l23);
    \draw[red] (t2) -- (t21);
    \draw[red] (t2) -- (t22);
    \draw[red] (t2) -- (t23);
    \draw[blue] (w2) to[out=65, in=105] (t2);

\end{tikzpicture}
    \caption{Two figures illustrating the construction of $Q$ in the proof of part \emph{(i)} of~\Cref{lemma:broadcast-constant-witness-existence}. Both figures depict $\C'$, which is a set of paths from $u$ to $v$ in $G + S^*$. On the left, $\C' = \C_3$. On the right, $\C' = \C_1$ (the case in which $\C' = \C_2$ is symmetric). Red edges denote the fan-out and fan-in edges of $\C'$. Blue edges are the middle segment of $Q$.
    }
    \label{fig:constant_witness}
\end{figure}
\begin{proof}

    We begin with the proof of part \emph{(i)}.
    Let $S^*$ be an optimum solution producing reach $\reach{G^*}{\sourceV}$
    in the probabilistic graph $G^* = G + S^*$. We impose an arbitrary order on the endpoints of
    the edges in $S^*$, $s_1, s_2, \ldots, s_m$, where $m \leq 2k$. We construct a new solution $S$, where
    \[
        S = S^* \cup \{s_1s_2, s_1s_3, \ldots s_1s_m\} \cup \{s_2s_3, s_2s_4, \ldots s_2s_m\} \cup \{s_3s_4, s_3s_5, \ldots s_3s_m\} 
    \]

    Intuitively, we have chosen three\footnote{Note that the edges of $S^*$ always involve at least three distinct endpoints, unless $k = 1$. In this latter case, \SRI{} is polynomial-time solvable.}
    distinct endpoints of edges in $S^*$, and formed stars with these endpoints as the centers and all other endpoints as the leaves.
    It is easy to check that $S$ has size at most $k + 3m - 6 \leq 7k - 6$. To complete the proof, we must show that
    $S$ contains a $(u, v, \frac{4\reach{G^*}{\sourceV}\pmin^2}{12k^4\pmax^2})$-witness of size at most three for every pair of vertices $u, v \in V$.

    Let $P_{uv}$ be the set of paths from $u$ to $v$ in $G^*$, and recall (see~\Cref{sec:preliminaries}) that the contribution of
    $P_{uv}$ is exactly equal to $\proxG{G^*}{u}{v}$, and therefore an upper bound for $\reach{G^*}{\sourceV}$. Now, let $p_i$ and $p_j$ be two paths in $P_{uv}$. Let
    $l_i$ be the leading vertex of the first edge contained in $S^*$ (a \emph{new edge}) to appear along $p_i$. Also, let $t_i$ be
    the trailing vertex of the last new edge to appear along $p_i$. Define $l_j$ and $t_j$ similarly. We impose an equivalence relation
    on $P_{uv}$ by declaring that $p_i$ is similar to $p_j$ if $l_i = l_j$ and $t_i = t_j$. Note that we may reserve one equivalence
    class for the set of paths containing no new edges (the \emph{empty class}), so the equivalence relation remains well-defined. Note that the paths are simple, so $l_i \neq t_i$. It follows that the number of equivalence classes is at most $2k*(2k-1)+1 < 4k^2$. The sum of the contributions of these
    classes is an upper bound for $\reach{G^*}{\sourceV}$, so there must be at least one class of paths with contribution
    at least $\frac{\reach{G^*}{\sourceV}}{4k^2}$. If the empty class meets this criteria, then
    we are done, as the empty set is a $(u, v, \frac{\reach{G^*}{\sourceV}}{4k^2})$-witness, and $\frac{\reach{G^*}{\sourceV}}{4k^2} > \frac{4\reach{G^*}{\sourceV} \pmin^2}{(12k^4)\pmax^2}$ as long as $k > 1$.
    Otherwise, choose one such equivalence class, defined by
    vertices $l$ and $t$, and call this class $\C$. At least one of $s_1, s_2$, or $s_3$ is distinct from
    both $l$ and $t$. Call this vertex $w$.

    We now further partition $\C$ into three subsets. The first, denoted $\C_1$, is the set of paths in $\C$ on which
    $w$ precedes $l$. The second, $\C_2$, is those paths on which $t$ precedes $w$. The third, $\C_3$, is all other paths in $\C$. The sum
    of the contributions of these three subsets is an upper bound for the contribution of $\C$.
    It follows that at least one has contribution at least $\frac{\reach{G^*}{\sourceV}}{12k^2}$. Let $\C'$ denote whichever of $\C_1$, $\C_2$, or $\C_3$,
    has the largest contribution. We will now show how to replace $\C'$ with a new set of paths $Q$ which uses at most three edges from $S$.

    We begin by handling a special case, namely the case in which the edge $lt$ appears along at least one path in $\C'$.
    In this case, we use the fact that $lt \in S$ to include those paths in $Q$ with no modifications.
    Consequently, conditioned on the existence of the edge $lt$ in a sampled graph, the existence of any path in $\C'$ implies
    the existence of a path in $Q$. In other words, when conditioned on the existence of the edge $lt$, the contribution of
    $Q$ is at least as large as the contribution of $\C'$.   
    It is therefore sufficient for us to show that the contributions of $Q$ and $\C'$ are not too different when conditioned
    on the non-existence of $lt$ in a sampled graph. Hence, we proceed with the simplifying assumption that
    the edge $lt$ does not appear along any path in $\C'$. 
    
    We now show how to edit these paths to form $Q$. See~\Cref{fig:constant_witness} for a visual aid.
    If $\C' = \C_1$, then for each $p \in \C'$ we replace $p[w, t]$ with the edge $wt$.
    If $\C' = \C_2$, then for each $p \in \C'$ we replace $p[l, w]$ with the edge $lw$.
    Otherwise, we replace $p[l, t]$ with the segment $(l, lw, w, wt, t)$.
    We observe that if $\C' = \C_1$, then $\C'[u, w] = Q[u, w]$ and $\C'[t, v] = Q[t, v]$. Similarly,
    if $\C' = \C_2$ then $\C'[u, l] = Q[u, l]$ and $\C'[w, v] = Q[w, v]$, and
    if $\C' = \C_3$ then $\C'[u, l] = Q[u, l]$ and $\C'[t, v] = Q[t, v]$. We call the segment
    of $\C'$ on which $Q$ differs, namely either $\C'[w, t], \C'[l, w]$, or $\C'[l, t]$ the \emph{middle segment}
    of $\C'$, denoted $\C'_m$, and we call the other two segments the \emph{beginning} and \emph{ending} segments, written $\C'_b$ and $\C'_e$, respectively.
    We define $Q_b, Q_m$, and $Q_e$ similarly. Moreover, we note that whatever the value of $\C' \in \{\C_1, \C_2, \C_3\}$, $Q_b = \C'_b$ and $Q_e = \C'_e$.

    Thus far, we have identified a set of paths $\C'$ from $u$ to $v$ in $G + S^*$ with contribution at least $\frac{\reach{G^*}{\sourceV}}{12k^2}$,
    and we have used $\C'$ to construct a new set $Q$ of paths which use (in total) at most three edges from $S$.
    To complete the proof, it is sufficient to show that $\Pr[Q] \geq \frac{4\pmin^2}{k^2\pmax^2}\Pr[\C']$.
    Intuitively, we accomplish this first by arguing that since $Q$ and $\C'$ have identical beginning and ending segments, it is sufficient
    to compare their middle segments, and second by performing that comparison. However, the potential positive correlation between
    paths in different segments of $\C'$ necessitates a slightly more technical argument.

    Let $(p_1, p_2)$ be a pair of paths from the beginning and ending segments of $\C'$,
    i.e., $p_1 \in \C'_b = Q_b$ and $p_2 \in \C'_e = Q_e$. We say that $(p_1, p_2)$ is a \emph{nice path pair}
    if $p_1$ and $p_2$ are vertex-disjoint, and that $(p_1, p_2)$ \emph{exists} in a sampled graph if
    both paths exist. Let $\cE_1$ be the event that a nice path pair exists in a sampled graph\footnote{
        Similar to the proof of~\Cref{thm:star-structure}, $\cE_1$ as stated is an event in two sample spaces, i.e., the spaces defined by sampling from $G + S^*$ and
        $G + S$. However, since edges are sampled independently and the edges relevant to $\cE_1$ exist in both graphs,
        the event remains well-defined and has equal probability under both measures.
    }.
    Note that by construction, the vertex $w$ does not appear on any path in either $Q_b$ or $Q_e$. Then the edges
    of the paths in $Q_m$ are disjoint from the edges of paths in $Q_b$ and $Q_e$. Noting that edges are sampled independently,
    we now have that $\Pr[Q] = \Pr[Q_m]\cdot\Pr[\cE_1]$. Moreover, because $Q_m$ consists of a single path
    on at most two edges, i.e., either the edge $wt$, the edge $lw$, or the path $l, w, t$, we may write $\Pr[Q_m] \geq \pmin^2$,
    and conclude that $\Pr[Q] \geq \pmin^2\Pr[\cE_1]$.

    We now upper bound $\Pr[\C']$. 
    We call those new edges (edges in $S^*$) which are incident to $l$ and used by at least one path in $\C'$
    the \emph{fan-out} edges of $\C'$. Similarly, we call those new edges which are incident to $t$ and used by at least
    one path in $\C'$ the \emph{fan-in} edges of $\C'$. Let $\cE_{out}$ (respectively, $\cE_{in}$) be
    the event that at least one fan-out (respectively, fan-in) edge exists in a sampled graph.
    Observe that $\Pr[\C'] \leq \Pr[\cE_1 \cap \cE_{out} \cap \cE_{in}]$. Using the fact that, by construction, no new edges appear on
    any paths in $\C'_b$ or $\C'_e$, we have that $\cE_1$ and $(\cE_{out} \cap \cE_{in})$ are independent. 
    Then $\Pr[\C'] \leq \Pr[\cE_1]\cdot\Pr[\cE_{out} \cap \cE_{in}]$.

    We now need only to upper bound $\Pr[\cE_{out} \cap \cE_{in}]$. Observe that each pair of edges, one being a fan-out edge and the other being a fan-in edge,
    exists in a sampled graph with probability at most $\pmax^2$. Recall, according to our prior argument, that the edge $lt$ does not appear on any path in $\C'$, so
    the fan-out and fan-in edges are disjoint sets. Furthermore, their union has size at most $k$. We use these facts to obtain the following
    bound:
    \[
        \Pr[\cE_{out} \cap \cE_{in}] \leq \pmax^2 \cdot \max_{i \in [k]}\{i\cdot(k-i)\} \leq \frac{\pmax^2 k^2}{4}.
    \]
    Putting the whole proof together, we see that
    \begin{align*}
        \proxG{G+W_{uv}}{u}{v} \geq \Pr[Q] \geq \pmin^2\Pr[\cE_1] &= \pmin^2\Pr[\cE_1]\cdot\frac{\Pr[\cE_{out} \cap \cE_{in}]} {\Pr[\cE_{out} \cap \cE_{in}]} \\ &\geq \frac{4 \pmin^2}{k^2 \pmax^2}\Pr[\C'] = \frac{\reach{G^*}{\sourceV} \pmin^2}{3k^4 \pmax^2},  
    \end{align*}
    where $W_{uv}$ consists of the at most three edges from $S$ appearing along paths in $Q$. Hence, we have found the desired witness.

    To prove part \emph{(ii)} of the lemma, we use a larger solution. Recalling that we have labeled the endpoints of the edges in $S^*$ $s_1, s_2, \ldots$
    we create a new solution $S = \{s_is_j \ | i \neq j\}$ of size at most $\binom{2k}{2}$. That is, we add an edge between each pair of
    endpoints of edges in $S^*$. The next part of the proof proceeds as before, up to the definition of $\C'$. Note that this time we do not need to handle the
    special case concerning edge $lt$ separately. We now construct our
    set $Q$ by replacing the segment $\C'[l, t]$ with the edge $lt$. This is the only new edge used by $Q$, so all that remains
    is to show that $\Pr[Q] \geq \frac{\reach{G^*}{\sourceV}\pmin}{12k^2}$. To accomplish this, we use the same definitions as before for nice path pairs and the
    event $\cE_1$. We note that $\Pr[\C'] \leq \Pr[\cE_1]$ by definition, and $\Pr[Q] \ge \pmin\cdot \Pr[\cE_1]$ because the sampling of edge $lt$ occurs independently
    of event $\cE_1$. The claim follows.

\end{proof}

We now give the details of the reduction to~\HS{} outlined at the beginning of this section.

\begin{restatable}{theorem}{bcwalg}\label{thm:broadcast-constant-witness-algorithm}
    For any $\varepsilon > 0$ and source-set $\sourceV \subseteq V$, there exist polynomial-time algorithms which produce
    probabiliprobabilistic graphs with reach (of $\sourceV$) at least (i) $\frac{\reach{G^*}{\sourceV} \pmin^2}{(1 + \varepsilon)3k^4 \pmax^2}$ using $O(k^3\log n)$ edge additions,
    and (ii) $\frac{\reach{G^*}{\sourceV}\pmin}{(1 + \varepsilon)12k^2}$
    using $O(k^2\log n)$ edge additions.
\end{restatable}
\begin{proof}
    We prove part \emph{(i)} of the theorem. The proof for part \emph{(ii)} is conceptually identical and therefore omitted for brevity.
        We will begin by assuming that we already know the value of $\reach{G^*}{\sourceV}$. In this case, we reduce to
        \HS{} as follows. We define $\W = \binom{V^2\setminus E}{3} \cup \binom{V^2\setminus E}{2} \cup \binom{V^2\setminus E}{1}$
        as the set containing all groups of at most three potential edge additions. Note that $|\W| \in O(n^6)$.
        The elements
        of $\W$ are the elements of our hitting set instance. Then, for each pair of vertices $(u, v) \in \sourceV \times V$
        with $\proxG{G}{u}{v} < \frac{\reach{G^*}{\sourceV}\pmin^2}{3k^4\pmax^2}$, we add a set $\W_{uv}$ consisting of
        all $(u, v, \frac{\reach{G^*}{\sourceV}\pmin^2}{3k^4\pmax^2})$-witnesses of size at most three. This completes the construction.
        We use~\cref{lemma:broadcast-constant-witness-existence}
        to observe that there exists a hitting set of size at most
        $\binom{7k-6}{3} + \binom{7k - 6}{2} + \binom{7k - 6}{1} \in O(k^3)$.
    
        The algorithm proceeds by using the well-known greedy $O(\log n)$-approximation for \HS{}~\cite{johnson1974approximation} to generate a
        hitting set of size $O(k^3\log n)$. We return the union of all the witnesses contained in this hitting set.
        By construction, this set of edge additions contains as a subset a $(u, v, \frac{\reach{G^*}{\sourceV}\pmin^2}{3k^4\pmax^2})$-witness
        for every pair $u, v$ of vertices, and because every member of our hitting set contains at most three edges, our solution has
        size $O(k^3\log n)$.
    
        It remains to show how we can estimate $\reach{G^*}{\sourceV}$. We will do this via the same technique used in~\Cref{thm:ballpacking-theorem}.
        In the following, let $b$ denote the precise bound on edge additions
        given by the algorithm in the preceding paragraph. That is, $b$ is
        $\binom{7k-6}{3} + \binom{7k - 6}{2} + \binom{7k - 6}{1} \in O(k^3)$
        multiplied by the approximation factor given by~\cite{johnson1974approximation}.
        We note that $\reach{G^*}{\sourceV} \leq 1$, so $\frac{\reach{G^*}{\sourceV}}{\broad{G}} \leq \frac{1}{\broad{G}}$.
        Then for any $\varepsilon > 0$, there exists some integer $0 \leq i \leq \log_{1+\varepsilon} \frac{1}{\broad{G}}$ with the property that
        $\broad{G}(1 + \varepsilon)^i \leq \reach{G^*}{\sourceV} \leq \broad{G}(1+\varepsilon)^{i+1}$. We conduct a binary search of integers in the
        interval $[0, \log_{1+\varepsilon}\frac{1}{\broad{G}}]$. Note that $\broad{G} \geq \pmin^n$, so this interval has
        polynomial length (for fixed $\varepsilon$).
        For each tested integer $j$, we assume that $\reach{G^*}{\sourceV} = \broad{G}(1 + \varepsilon)^j$, and execute the
        algorithm described above. If the algorithm adds more than $b$ edges, then we conclude that $\broad{G}(1 + \varepsilon)^j \geq \reach{G^*}{\sourceV}$,
        and therefore that $j > i$. Let $j^*$ be the largest integer in the interval for which our algorithm adds at most $b$ edges.
        Then we can conclude that $\reach{G^*}{\sourceV} \leq \broad{G}(1+\varepsilon)^{j^*+1}$, and in this case our algorithm
        adds at most $b \in O(k^3\log n)$ edges to produce reach at least
        \[
          \frac{\broad{G}(1 + \varepsilon)^{j^*}\pmin^2}{3k^4\pmax^2} = \frac{\broad{G}(1 + \varepsilon)^{j^*+1}\pmin^2}{(1+\varepsilon)3k^4\pmax^2} \geq \frac{\reach{G^*}{\sourceV}\pmin^2}{(1+\varepsilon)3k^4\pmax^2},    
        \]
        as desired.
\end{proof}

We conclude this section by showing that we can obtain an improved bound for the special case of~\RI{}, i.e., we prove~\Cref{thm:ss-klogn}.

\ssstarlemma*
\begin{proof}
    Let $S_{star} = \{\sourcev s \colon s \in V(S) \setminus \{\sourcev\}\}$.
    Clearly, $S_{star}$ has size at most $2k$. Let $v$ be an arbitrary vertex, and consider the set of all paths from $\sourcev$ to $v$ in $G'$.
    Partition these paths according to the last new edge (edge in $S$) encountered. Including the empty class, i.e., the class of paths which use no edges from $S$, there are at most $k + 1$ equivalence classes.
    Hence, some class has contribution at least $\beta'/(k+1)$. If the empty class is one such class, then we observe that these paths also exist (and also have contribution at least $\beta'/(k+1)$) in $G+S_{star}$.
    Otherwise, let $e = s_is_j$ be the last new edge encountered by paths in the class with the largest contribution. Partition this class according to the orientation of $e$, i.e., $s_i \rightarrow s_j$ or $s_j \rightarrow s_i$.
    One set has contribution at least $\beta'/(2k+2)$. Without loss of generality, assume this is the set of paths for which $s_i$ is the trailing vertex of $e$, and call this set of paths $\mathcal C$.
    Now, consider the set of paths $Q$ from $\sourcev$ to $v$ in $G + S_{star}$ defined by prepending the edge $\sourcev s_i$ to each path in $\mathcal C[s_i, v]$.
    Note that $Q[s_i, v] = \mathcal C[s_i, v]$, and these sets have equal contribution. Then since every path in $\mathcal{C}[\sourcev, s_i]$ uses the edge $e$, we have that $\Pr[Q[s_i, v]] = \Pr[\mathcal C[s_i, v]] \geq \Pr[\mathcal C]/\pmax$.
    Meanwhile, by construction $\Pr[Q] \geq \pmin\cdot\Pr[Q[s_i, v]]$. Combining these inequalities, we have that $\Pr[Q] \geq \frac{\beta'\pmin}{(2k+2)\pmax} = \beta_{star}$, as desired.
\end{proof}

Observe that the second conclusion of the lemma implies that the solution $S_{star}$ is in fact a $(1, \frac{\reach{G^*}{\sourcev}\pmin}{(2k+2)\pmax^2})$-witnessing solution of size at most $2k$.
To achieve~\Cref{thm:ss-klogn}, we now need only form the reduction to \HS{} described at the beginning of this section.

\ssalg*
\begin{proof}[Proof of~\Cref{thm:ss-klogn}]
    Let $(G = (V, E), \sourcev, k)$ be an instance of \RI{}.
    We begin by assuming that we already know the optimum achievable reach $\reach{G^*}{\sourcev}$. We
    reduce to \HS{}.
    The elements of our \HS{} instance are $V^2 \setminus E$, i.e., all possible edge additions.
    For each vertex $u \in V$ with $\proxG{G}{\sourcev}{u} < \reach{G^*}{\sourcev}$, we create a set $\W_u$
    consisting of all single edge-additions which improve the proximity of $\sourcev$ to $u$ to at least $\frac{\reach{G^*}{\sourcev}\pmin}{(2k + 2)\pmax}$.
    That is, $\W_u$ is the set of all $(u, \frac{\reach{G^*}{\sourcev}\pmin}{(2k + 2)\pmax})$-witnesses of size 1.
    According to~\cref{lemma:ss-star-lemma}, there exists a hitting set of size at most $2k$. We use
    the well-known greedy approximation for \HS{}~\cite{johnson1974approximation} to obtain a hitting set of size $O(k\log n)$, and we return these edges as
    our solution. It follows from the construction that this solution achieves reach at least $\frac{\reach{G^*}{\sourcev}\pmin}{(2k + 2)\pmax}$.
    The procedure to estimate the value of $\reach{G^*}{\sourcev}$ is similar to~\Cref{thm:broadcast-constant-witness-algorithm}.
\end{proof}

\bibliographystyle{abbrvnat}
\bibliography{refs}

\end{document}